%% file: main.tex
\documentclass[10pt,usenames,dvipsnames]{article}
\usepackage{enumerate}
\usepackage{amsmath,amssymb}
\usepackage{natbib}
\usepackage{caption}
\usepackage[usenames]{color}
\usepackage{ying}
\usepackage{multirow}
\usepackage{rotating}
\usepackage{fullpage}
\usepackage{authblk}
\usepackage{enumerate}
\usepackage[margin=1in,footskip=0.25in]{geometry}
\usepackage{comment}

\usepackage{mathtools}
\mathtoolsset{showonlyrefs}

\usepackage{graphicx,amssymb}
\usepackage{xcolor}
\usepackage[colorlinks,
            linkcolor=red,
            anchorcolor=blue,
            citecolor=blue
            ]{hyperref}
\usepackage{algorithm}
\usepackage{algorithmic}

\let\hat\widehat
\let\tilde\widetilde

\def \iid {\stackrel{\text{i.i.d.}}{\sim}}

\newcommand{\dm}{{\textrm{DiM}}}
\newcommand{\deb}{{\textrm{Debias}}}

\theoremstyle{plain}

\def \pto {\stackrel{P}{\to}}

\usepackage{hyperref}
\def\##1\#{\begin{align}#1\end{align}}
\def\$#1\${\begin{align*}#1\end{align*}}

\usepackage{tikz,tabularx}
\usetikzlibrary{patterns}
\usetikzlibrary{arrows}
\usetikzlibrary{tikzmark, positioning, fit, shapes.misc}
\usetikzlibrary{decorations.pathreplacing, calc}
\tikzset{brace/.style={decorate, decoration={brace}},
  brace mirrored/.style={decorate, decoration={brace,mirror}},
}
\def\keywords{\vspace{.5em}
{\textit{Keywords}:\,\relax%
}}

\usepackage{xcolor}

\title{\LARGE Towards Optimal Variance Reduction in Online Controlled Experiments}
\author{Ying Jin}
\author[2]{Shan Ba}

\affil{Department of Statistics, Stanford University
}
\affil[2]{Data Science Applied Research, LinkedIn Corporation}
\begin{document}

\maketitle

\begin{abstract}
We study optimal variance reduction solutions for 
count and ratio metrics in online controlled experiments.
Our methods leverage flexible machine learning tools to incorporate covariates that are independent from the treatment but have predictive power for the outcomes, and employ 
the cross-fitting technique to remove the bias in complex machine learning models. 
We establish CLT-type asymptotic inference based on our estimators under mild convergence conditions. 
Our procedures 
are optimal (efficient) for the corresponding targets
as long as the machine learning estimators are consistent, 
without any requirement for their convergence rates. 
In complement to the general optimal procedure, 
we also derive a linear adjustment method for ratio metrics 
as a special case  that 
is computationally efficient and can flexibly incorporate any pre-treatment covariates. 
We evaluate the proposed variance reduction procedures with
comprehensive simulation studies and provide  practical suggestions regarding 
commonly adopted assumptions in computing ratio metrics. 
When tested on real online experiment data from LinkedIn, 
the proposed optimal procedure for ratio metrics can reduce 
up to 80\% of variance compared to the standard difference-in-mean estimator and also further reduce up to 30\% of variance compared to the CUPED approach by going beyond linearity and incorporating a large number of extra covariates.
\end{abstract}


\section{Introduction}
Online controlled experiments, also known as A/B tests, are extensively used in tech companies to assess the impacts of product changes on business metrics. 
Although online experiments can involve millions of users, their sensitivity is still a major challenge because 
the treatment effect is often small 
compared to the noise.  
Failing to detect even small differences in key metrics can have significant business implications~\citep{kohavi_tang_xu_2020} and it is crucial to develop 
powerful 
statistical tools to quickly capture nonzero effects
with fewer samples 
and  shorter
experimentation turn-around time.

To improve the sensitivity of online controlled experiments, variance reduction techniques are commonly used which leverage relevant covariates to 
remove explainable variance in 
the outcomes.  
For count metrics, 
many variance reduction solutions have been developed in the literature based on classical linear adjustment
\citep{yang2001efficiency,freedman2008regression,lin2013agnostic,deng2013improving} 
and 
machine learning tools~\citep{hosseini2019unbiased,guo2021machine}.
%
However, 
the optimality of 
variance reduction procedures has not been thoroughly studied and using suboptimal procedures may lead 
to unnecessary costs. 
In addition to count metrics, online experiments also often involve ratio metrics~\citep{deng2017trustworthy} whose variance reduction is more complex but much less studied. As we will discuss in Section~\ref{subsec:metric}, ratio metrics are similar to the setting of cluster randomized experiments. Existing variance reduction solutions for ratio metric are mostly extensions to those of count metrics \citep{deng2013improving} and 
a rigorous 
statistical framework for (optimal) variance reduction 
remains absent.

In this paper, we target at the natural but unanswered question:
\emph{
    With access to a set of covariates that are independent of the treatment in the experiment,  what is the optimal estimator for comparing the outcomes of treatment and control groups?}
We study the optimal variance reduction procedures for both
count and ratio metrics that are ubiquitous 
in online controlled experiments in the industry. 
The optimality we focus on is 
semiparametric efficiency~\citep{bickel1993efficient}. 
Given an estimand that 
arises from the comparison of experiment outcomes, 
our goal is to 
develop an estimator with the smallest 
asymptotic variance among all asymptotically unbiased estimators. 
We propose procedures that reduce variance 
of treatment effect estimation by incorporating flexible ML regressors with rigorous statistical  guarantee. 
Based on classical semiparametric statistics theory, 
we establish the optimality of our procedures under mild conditions. 
For ratio metrics,
in addition to the optimal (and perhaps nonlinear) 
approaches, 
we also propose a computationally efficient 
linear adjustment method 
which, to the best of our knowledge, are not available in the literature. 

The rest of the paper is organized as follows. In section~\ref{sec:review}, we introduce the definition of count and ratio metrics in online experiments and provide an overview of 
the related literature. 
We develop the variance reduction procedures, asymptotic inference and optimality properties for count metrics in Section~\ref{sec:count}, and for ratio metrics in Section~\ref{sec:rr} and \ref{sec:r}. Section~\ref{sec:simulation} is devoted to simulation studies and Section~\ref{sec:example} illustrates the performance of our methods using real online experiments from LinkedIn. 
Finally, we give concluding remarks in Section~\ref{sec:conclusion}. 

\section{Problem setting and related work} \label{sec:review}

The \emph{randomization unit} and the \emph{analysis unit}
are two important concepts for online experiments. The most common online experiments are randomized by users, while sometimes the experiments can also use alternative randomization units.
For example, in online experiments for enterprise products, the randomization unit typically needs to be a cluster of users (such as
an enterprise account or contract) because users in the same contract must have access to the same product feature.
In another scenario where it is infeasible to identify/track users in a web service, the randomization unit is often chosen to be a service request or a pageview.
The analysis unit of an experiment, on the other hand, may not necessarily be the same as the experiment's randomization unit. For example, in a cluster randomized  experiment for enterprise products, we can choose the analysis unit at either the cluster level (e.g., revenue per cluster) or the individual user level (e.g., revenue per user). Depending on whether the randomization 
and the analysis units coincide, metrics in online experiments can be classified into different types, for which different analysis procedures are needed. 
 
%
%

\subsection{Count and ratio metrics}\label{subsec:metric}

To formalize the 
problem setting, we adopt the potential outcomes framework 
and follow the standard Stable Unit Treatment Value Assumption (SUTVA)~\citep{imbens2015causal} so that there is no interference among
the randomization units. 
We also take a super-population perspective 
where the randomization units can be viewed as i.i.d.,~but the 
analysis units may not. 

\paragraph{Count metrics.} 
The most common metrics are count metrics, whose analysis unit matches the randomization unit in the experiment. For example, in online experiments that are randomized by users, count metrics are those defined on the user level such as revenue per user, pageviews per user, number of clicks per user, etc. Because the analysis units 
are the same as the the i.i.d.~randomization units, 
the variance of count metrics 
can be estimated directly by the sample variance formula.

Formally, suppose there are $n$ units in the experiment, 
    for which we have access to i.i.d.~observations $\{(X_i,Y_i,T_i)\}_{i=1}^n$ from an unknown distribution, where 
  $X_i$ is the pre-treatment covariates, 
$Y_i$ is the measured metric and $T_i$ is the treatment indicator.
Following the standard practice 
of online controlled experiments, 
we assume the treatment indicators $T_i\iid$ Bernoulli($p$) 
for $p\in (0,1)$ and 
are independent of all other information. 
The units in the treated group ($T_i=1$) receive the treatment, 
and the units in the control group ($T_i=0$) do not. 
%
The outcomes $\{Y_i\}_{i=1}^n$ are then measured 
after the experiment. 
Under the potential outcome framework, 
each unit has two potential outcomes $(Y_i(1),Y_i(0))$, 
where $Y_i(1)$ is the outcome that unit $i$ exhibits 
under treatment, and $Y_i(0)$ is that under control. 
For each unit, we only observe one potential outcome $Y_i = Y_i(T_i)$ 
under SUTVA~\citep{imbens2015causal}. 
Typically, count metrics are aggregated by sample means 
such as $\frac{1}{\sum_{i}T_i} \sum_{T_i=1} Y_i$ for 
the units in the treatment or control groups, 
and the  difference between two groups 
shows the causal effect of the treatment. 
The corresponding estimand is 
\$
\tau = \EE\big[Y_i(1)\big] - \EE\big[Y_i(0)\big],
\$ 
where the expectations are taken with respect to 
the distribution that the units are from. 

\paragraph{Ratio metrics.} 
When the analysis unit is at a lower level
than the experiments' randomization unit, the metric of interest is a ratio metric (because it can be expressed as a ratio of two count metrics), whose 
unique structure requires distinct techniques for 
inference. 

In user-randomized experiments, the click-through rate (average number of clicks per pageview, computed as number of clicks /number of pageviews) is an example of ratio metric 
whose analysis unit is at the pageview level. 
Because the experiment is randomized by users and 
different pageviews of the same user are correlated,
only the outcomes aggregated at the user level (randomization unit level) can be viewed as i.i.d..
For analysis purposes, the click-through rate 
can be equivalently viewed as a ratio of 
two user-level count metrics
(number of clicks per user
/number of pageviews per user).
Consider another experiment for enterprise products which needs to be randomized by contracts to ensure that all users within each contract receive the same treatment assignment. Revenue per contract is a standard count metric in this experiment, but in practice we are often more interested in analyzing revenue per user, which is a ratio metric. Because only the contracts are i.i.d.~and users under each contract are not independent, variance of the revenue per user cannot be directly calculated by the sample variance formula. Instead, by viewing revenue per user as a ``ratio'' of two contract-level count metrics (revenue per contract/number of users per contract), we can estimate its variance based on the delta method \citep{deng2017trustworthy}. 
More broadly, this setting is similar to the cluster randomized experiments in causal inference~\citep{green2008analysis,middleton2015unbiased} as each randomization unit can be viewed as a cluster of analysis units.

Formally, let $i=1,\dots,n$ be i.i.d.~randomization units in the experiment from some distribution $\PP$. They are randomly allocated to treated or control groups, indicated by $T_i\in\{0,1\}$. Here we assume $T_i \iid \text{Bernoulli}(p)$ for some $p\in (0,1)$ and are independent from all other information. 
Suppose two count metrics $Y_i$ and $Z_i$ are measured for each randomization unit $i$ (e.g., $Y_i$ is the total number of clicks of user $i$ and $Z_i$ is the total number of pageviews of user $i$). The ratio metrics (e.g., click through rate) for the treated and control groups are defined as 
\$
\frac{  \sum_{i~\textrm{treated}} Y_i}{ \sum_{i~\textrm{treated}} Z_i}  
= ~
\frac{\frac{1}{n_t} \sum_{i~\textrm{treated}}Y_i }{\frac{1}{n_t} \sum_{i~\textrm{treated}} Z_i},\quad 
\frac{ \sum_{i~\textrm{control}}Y_i}{ \sum_{i~\textrm{control}} Z_i}
= ~
\frac{\frac{1}{n_c} \sum_{i~\textrm{control}}Y_i }{\frac{1}{n_c} \sum_{i~\textrm{control}} Z_i},
\$
where $Y_i$ and $Z_i$ are aggregated
across all randomization units 
in the treatment/control groups first before taking the ratio.  
Because randomization units are i.i.d., 
the law of large numbers implies that the 
population-level comparison 
target (the limit 
of difference in ratio metrics) 
is \emph{the difference in ratios of expectations}
$
\frac{\EE[Y_i\given T_i=1]}{\EE[Z_i\given T_i=1]} - \frac{\EE[Y_i\given T_i=0]}{\EE[Z_i\given T_i=0]}.
$ 
In practice, the expectation of $Z_i$ is always nonzero and hence both the population-level and sample-level ratios are well defined.
Note that we do not define ratio metrics for the treatment and control groups as $\frac{1}{n_t}\sum_{i~\textrm{treated}} \frac{Y_i}{Z_i}$ and $\frac{1}{n_c}\sum_{i~\textrm{control}} \frac{Y_i}{Z_i}$, because some individual $Z_i$ may be zero and the corresponding individual ratio $Y_i/Z_i$ is not well-defined.
The previous definition of ratio metrics can actually be viewed as a weighted average of $Y_i/Z_i$ with 
the weights proportional to $Z_i$ (e.g., 
giving more weights to more active users):
\$
    \frac{\frac{1}{n_c} \sum_{i~\textrm{control}}Y_i }{\frac{1}{n_c} \sum_{i~\textrm{control}} Z_i}
    = \sum_{i~\textrm{control}}\underbrace{\frac{Z_i}{\sum_{i~\textrm{control}} Z_i}}_{\textrm{"weights"}} \times \underbrace{\frac{Y_i}{Z_i}}_{\textrm{ratio for randomization unit $i$}} .
    \$

Ratio metrics can also be grouped into two types depending on their assumptions: 
(1) The first type allows both the metric $Z$ in the denominator and the metric $Y$ in the numerator to be changed by the treatment. 
The above click-through rate example belongs to this type.  
(2) The second type assumes that the metric $Z$ in the denominator is a random variable associated with the experiment unit, but is \emph{not} changed by the treatment. 
\cite{deng2017trustworthy} refers to (2) 
as the \emph{stable denominator assumption} (SDA in the following). 
SDA is plausible when the ratio metric mainly uses $Z$ in the denominator as a normalization factor 
to standardize the changes in the numerator metric $Y$. 
For example, when we are interested in a user-level metric ``revenue per user'' but the experiment needs to be randomized by clusters of users, we can choose $Z_i$ as the number of active users in an cluster. 
This assumption always needs to be checked in practice using a separate test on $Z_i$.
If the SDA is violated, the ratio metric itself is hard to interpret: we do not know whether an increase in the ratio is good (e.g., due to an increase in revenue) or bad (e.g., due to a decrease in the number of active users).
For instance, a bad treatment which decreases the number of active users in the contracts can instead yield a higher revenue-per-user ratio as the remaining users most likely are the most active ones.
When the SDA is violated, the analysis should also emphasize on studying the count metric change in the numerator $Y$ and in the denominator $Z$ separately before drawing conclusions.
%

We now define the estimands for the two types of ratio metrics under the potential outcomes framework. 
For type (1) ratio metric, each unit has potential outcomes ($Y_i(1)$, $Y_i(0)$, $Z_i(1)$, $Z_i(0)$), where we observe $(Y_i,Z_i) = (Y_i(T_i),Z_i(T_i))$. 
The estimand is 
\#
\delta = \frac{\EE[Y_i(1)]}{\EE[Z_i(1)]} - \frac{\EE[Y_i(0)]}{\EE[Z_i(0)]}, \label{eq:rr_target}
\# 
where the expectation is with respect to 
the distribution the units are from. 
In the type (2) ratio metric, each unit has potential outcomes $(Y(1),Y(0))$ so that $Y_i=Y_i(T_i)$ while $Z$ is a plain random variable, and the estimand is 
\#
\delta' = \frac{\EE[Y_i(1)]}{\EE[Z_i ]} - \frac{\EE[Y_i(0)]}{\EE[Z_i ]}. \label{eq:r_target}
\#
In this work, we will consider both types of estimands for the ratio metrics and offer practical recommendations. As we will show in our numerical experiments, 
the assumption (2) should be made with caution.

\subsection{Related work}\label{subsec:review}


This work falls into the general randomized experiment setting in causal inference \citep{imbens2015causal},
whereas we study the post-hoc variance reduction 
instead of the 
experimental design strategies. 
In addition, our method leverages 
machine learning models to
learn the conditional relation between the potential outcomes 
and the covariates, 
which is similar to but with distinct goal from 
 the investigation 
of treatment effect heterogeneity in causal inference
\citep{athey2016recursive,chernozhukov2017generic,imai2013estimating,kunzel2019metalearners,nie2020quasi,wager2018estimation}:  we fit 
the conditional mean functions of the potential outcomes to 
remove 
explainable variations in the outcomes, rather than 
to investigate the heterogeneity of the conditional treatment effect.

This work adds to the rich literature of 
variance reduction with covariate adjustment
in randomized experiments. 
Analysis of covariance (ANCOVA) has a long history of 
application in physical experiments \citep{wu_hamada_2009}.
The classic linear regression adjustment \citep{yang2001efficiency,freedman2008regression,lin2013agnostic} is shown to work even when the linear model is misspecified.  
\cite{deng2013improving} proposes to use pre-treatment data as the regression covariates and the corresponding variance reduction method, called CUPED (Controlled-experiment Using Pre-Experiment Data), has been widely used in the industry. 
The blossom of ML research also inspires
a line of recent work on using ML tools for variance reduction, including \cite{hosseini2019unbiased,cohen2020no,guo2021machine} and the references therein,
but few of them focus on the optimality of the variance reduction procedure. 
In particular, among  the existing works for count metrics, 
both \cite{guo2021machine} and \cite{cohen2020no} uses predicted outcomes from ML estimators as covariates in a linear regression adjustment. 
Their intuitions are to create more relevant features with ML tools to improve upon CUPED, but not aimed at optimality. 
As would be discussed in details at the end of Section~\ref{subsec:ct_optimal}, careful considerations are needed for optimality and the method in \cite{guo2021machine} can not be semiparametrically efficient in general.
\cite{hosseini2019unbiased} discusses 
ratio metrics and provides methods 
to obtain unbiased estimators with machine learning tools. 
However, to our knowledge, 
rigorous and explicit statistical inference guarantees and 
the optimality of the procedures are not provided.

Our approach for count metrics is asymptotically
the same as the augmented inverse propensity weighting (AIPW) estimator \citep{robins1994estimation}, 
whose well-established semiparametric efficiency~\citep{hahn1998role} result 
forms the basis of our optimality guarantee. 
For count metrics, we develop valid inference procedures
in randomized experiments ($L_2$
convergence in probability to any fixed function), 
which is  much weaker than 
the pointwise convergence condition 
to true conditional mean functions 
as is often required in observational studies \citep{nichols2007causal,schuler2017targeted,chernozhukov2018double}
or the investigation of heterogeneous treatment effects \citep{chernozhukov2017generic,kunzel2019metalearners,athey2019estimating,nie2020quasi,kennedy2020optimal},
and also differs from the traditional approach of Donsker conditions and empirical process theory \citep{andrews1994empirical,van2000asymptotic,van1996weak} to 
control errors in estimating nuisance components (the conditional mean functions in our setting).

Our estimators for ratio metrics 
as well as its inference and optimality results  
are new to the literature. 
Intuitively, they all have a fit-and-debias flavor related to the AIPW estimator~\citep{robins1994estimation}.
The optimality theories we develop for ratio metrics finds roots in  
the classical semiparametric efficiency theory~\citep{bickel1993efficient,hahn1998role}.


\section{Variance reduction for count metrics}\label{sec:count}

\input{count}

\section{Variance reduction for ratio metrics: without SDA} \label{sec:rr}

\input{ratio_rr}

\section{Variance reduction for ratio metrics: with SDA}\label{sec:r}
\input{ratio_rcut}

\section{Simulation studies} \label{sec:simulation}
\input{simu}

\section{Real examples} \label{sec:example}
\input{real}

\section{Conclusions} \label{sec:conclusion}
We establish a rigorous statistical framework for variance reduction of count and ratio metrics that are popular in online controlled experiments.
We 
propose variance reduction methods that utilize  flexible 
ML tools to incorporate large numbers of covariates, 
yielding unbiased estimators for treatment effects under mild conditions. 
Based on semiparametric efficiency theory, we establish the optimality of the proposed procedures, sheding light on ideal solutions for variance reduction. 
Simulation studies 
illustrate the performance of our methods and also give practical suggestions for ratio metrics under different assumptions. 
Finally, 
two real online experiments from LinkedIn are used to illustrate the applicability of our methods.

Our current work opens  several interesting avenues for future work. 
Firstly, 
direct extension of the current work is 
to add a linear regression step on top of our methods to 
utilize the ``agnostic'' property of linear adjustment. 
This might help in practice when the ML estimators are much off; 
however, this does not help with optimality, since 
under the   consistency condition for the ML predictor,
adding a linear regression step 
does not change the asymptotic behavior of the estimator. 
Another line of future work is variance reduction with limited capacity of ML estimators
 --- the consistency condition 
may not be satisfied 
especially when the ML estimator is chosen among a function class that does not cover the true conditional mean function.
In this case, it would be of interest to investigate optimal solutions given such constraints. 
Another possible extension is to study robust solutions 
for ratio metrics. 
As illustrated in the simulations, erroneously making the SDA might hurt the validity of the optimal procedure for $\delta'$ in~\eqref{eq:r_target}. 
A procedure that yields valid inference for the target $\delta'$ even when the SDA is violated might be desired. Instead of targeting at optimality, it would be interesting to find a robust solution that is ``optimal'' in a possibly more restrictive sense.

\section{Acknowledgements}

This work is done during the first author's internship at LinkedIn Applied Research team. 
The authors would like to thank Dominik Rothenh\"{a}usler for helpful discussions and thank Weitao Duan, Rina Friedberg, Reza Hosseini, Juanyan Li, Min Liu, Jackie Zhao, Sishi Tang and Parvez Ahammad for their suggestions and feedbacks.  


\bibliographystyle{ims}
\bibliography{reference}

\newpage 

\appendix

\input{appendix}

\end{document}

%% file: count.tex

We begin our discussion with 
intuitions on how machine learning can assist 
variance reduction for count metrics. 
To estimate the treatment effect  $\tau = \EE[Y_i(1)]-\EE[Y_i(0)]$,  
 an ideal ``estimator'' is
$
\hat\theta_{\textrm{ideal}} = \frac{1}{n}\sum_{i=1}^n (Y_i(1)-Y_i(0))
$
which directly compares pairs of potential outcomes, 
although it is not rigorously an estimator since 
$Y_i(1)$ and $Y_i(0)$ are not simultaneously observable. 
This ``estimator'' is ideal for two reasons: (1) large sample size: both the treated and the control group leverage $n$ observations; (2) paired comparison: the variance of the individual treatment effect $Y_i(1)-Y_i(0)$ is typically smaller than the variances of $Y_i(1)$ and $Y_i(0)$ alone.
A standard Difference-in-Mean (DiM) estimator 
compares the averages of the two groups:
$
\hat\theta_{\textrm{DiM}}
= \frac{1}{n_t} \sum_{i~\textrm{treated}}  Y_i(1)  - \frac{1}{n_c} \sum_{i~\textrm{control}}  Y_i(0),
$
which has larger variance than $\hat\theta_{\textrm{ideal}}$
by
Cauchy-Schwarz inequality. 
While $\hat\theta_{\textrm{ideal}}$ is not computable, 
a natural idea is to impute the unobserved potential outcomes 
based on covariates $X_i$.
Let 
$\hat\mu_1(\cdot)$ and $\hat\mu_0(\cdot)$ be some machine learning 
predictors for $Y(1)$ and $Y(0)$ based on $X$, respectively. 
By naively plugging in the predictors, we obtain 
$
\hat\theta_{\textrm{plug-in}} = \frac{1}{n} \sum_{i=1}^n \big(\hat{Y}_i(1) - \hat{Y}_i(0)\big),
$
where 
 \$
 \hat{Y}_i(1) = \begin{cases} Y_i(1),\quad &T_i=1\\ \hat\mu_1(X_i),\quad &T_i=0    \end{cases},\qquad \hat{Y}_i(0) = \begin{cases} Y_i(0),\quad &T_i=0\\ \hat\mu_0(X_i),\quad &T_i=1    \end{cases}.
 \$ 
The simplicity of $\hat\theta_{\textrm{plug-in}}$ 
comes with two problems. First, without a well-posed parametric model, the convergence rate of $\hat\mu_1(\cdot)$ and $\hat\mu_0(\cdot)$ is often slower than $n^{-1/2}$, 
hence $\hat\theta_{\textrm{plug-in}}$ can introduce ``\emph{regressor bias}'' 
that is even larger than the variance. 
Secondly, another ``\emph{double-dipping bias}'' occurs as the same dataset is used both for model-fitting and for prediction,
hence
$\{\hat\mu_w(X_i)\}_{1\leq i\leq n}$ are no longer independent copies, posing challenges for statistical analysis. 
We will develop debiasing techniques for these issues.

We also note that 
the above intuition
is related to 
\cite{guo2021generalized,cohen2020no} 
where machine learning predictions 
are used to impute the counterfactuals. 
However, our work takes a specific approach 
that aims at optimality. 
The conditions for valid inference also differ from them.

\subsection{Estimation procedure}\label{subsec:count_procedure}

Our proposed variance reduction procedure 
improves upon the naive $\hat\theta_{\textrm{plug-in}}$ 
in two aspects: it eliminates the ``regressor bias'' by adding a de-biasing term and employs the cross-fitting technique  \citep{chernozhukov2018double} to 
correct for the ``double-dipping bias''. 

The first step is to randomly split 
the original dataset $\cD = (Y_i,T_i,X_i)_{i=1}^n$ into $K$ (roughly) equal-sized folds $\{\cD^{(k)}\}_{1\leq k\leq K}$ with sizes $n_k=|\cD^{(k)}|$, each fold containing $n_{k,t} = \sum_{i\in \cD^{(k)}}T_i$  treated samples and $n_{k,c} = n_k - n_{k,t}$ control samples. 
Here the random splitting is conducted separately in the treated and control groups to achieve 
balanced numbers of treated samples 
across folds. In practice, $K=2$ generally works well.

The second step is cross-fitting \citep{chernozhukov2018double}. Let $\cD^{(-k)} = \cD\backslash \cD^{(k)}$ denote all data after holding out the $k$-th fold.
For each $k\in[K]$, fit a function $\hat\mu_1^{(k)}(x)$ for $\EE[Y(1)\given X=x]$ using $\big\{(X_i,Y_i)\colon T_i=1, i\in\cD^{(-k)} \big\}$,  and fit a function $\hat\mu_0^{(k)}(x)$ for $\EE[Y(0)\given X=x]$ using $\big\{(X_i,Y_i)\colon T_i=0, i\in\cD^{(-k)} \big\}$. Then apply the fitted functions to the held-out samples in $\cD^{(k)}$ to generate out-of-sample predictions  $\hat\mu_0(X_i) = \hat\mu_0^{(k)}(X_i)$ and $\hat\mu_1(X_i) = \hat\mu_1^{(k)}(X_i)$ for $i\in \cD^{(k)}$. 
These out-of-sample predictions are independent conditional on $\cD^{(-k)}$, helping alleviate the ``double-dipping bias''.

Finally, we estimate $\tau$ via
\#\label{eq:db_count}
& \hat\theta_\deb = \frac{1}{K}\sum_{k=1}^K ~\hat\theta_{\deb}^{(k)},\quad \text{where}\\
\hat\theta_{\deb}^{(k)} = 
\frac{1}{n_k}\sum_{i\in \cD^{(k)}} &  \big(\hat\mu_1(X_i) - \hat\mu_0(X_i)\big) + \frac{1}{n_{k,t}}\sum_{\substack{T_i=1,\\i\in \cD^{(k)}}} \big( Y_i - \hat\mu_1(X_i) \big) - \frac{1}{n_{k,c}} \sum_{\substack{T_i=0,\\i\in \cD^{(k)}}} \big( Y_i - \hat\mu_0(X_i)  \big). \notag
\#
As will be shown later, this estimator is finite-sample unbiased and one could also obtain similar performance by the following estimator
\#\label{eq:db_count_0}
&\tilde\theta_\deb =  
\frac{1}{n }\sum_{i=1}^n   \big(\hat\mu_1(X_i) - \hat\mu_0(X_i)\big) + \frac{1}{n_t}\sum_{i~\textrm{treated}} \big( Y_i(1) - \hat\mu_1(X_i) \big) - \frac{1}{n_{c}} \sum_{i~\textrm{control}} \big( Y_i(0) - \hat\mu_0(X_i)  \big),
\#
which might have an $O(1/n)$ bias from unequal fold sizes. 
Note that $\tilde\theta_\deb$ adds a de-biasing term to the naive plug-in estimator 
\$
\tilde\theta_\deb \approx \hat\theta_{\textrm{plug-in}} + \frac{ n_c}{n\cdot n_t}\sum_{i~{\scriptsize \textrm{treated}}} \big( Y_i(1) - \hat\mu_1(X_i) \big) - \frac{ n_t}{n\cdot n_c} \sum_{i~{\scriptsize \textrm{control}}} \big( Y_i(0) - \hat\mu_0(X_i)  \big).
\$
In this equation, $\hat\theta_{\textrm{plug-in}}$ uses machine learning tools to 
impute the unobserved potential outcomes, and the remaining two terms serve as corrections to the ``regressor bias'' in fitting the mean functions. 
As will be seen in Section~\ref{sec:rr} and~\ref{sec:r}, 
the debiasing technique 
is also useful when developing estimators 
for ratio metrics with careful considerations on the imputation of the predicted values.  
The proposed procedure is summarized in Algorithm~\ref{alg:debias}. 
\begin{algorithm}[H]
\caption{Debiased Variance Reduction}\label{alg:debias}
\begin{algorithmic}[1]
\STATE Input: Dataset $\cD =\{(Y_i,X_i,T_i\}_{i=1}^{n}$, number of folds $K$.
\STATE Randomly split $\cD$ into $K$ folds $\cD^{(k)}$, $k=1,\dots,K$.
\FOR{$k=1,\dots,K$} 
\STATE Use all $(X_i,Y_i)$ with $T_i=1$ and $i\notin \cD^{(k)}$ to obtain $\hat\mu_1^{(k)}(x)$ for $\EE[Y(1)\given X=x]$; 
\STATE Use all $(X_i,Y_i)$ with $T_i=0$ and $i\notin \cD^{(k)}$ to obtain $\hat\mu_0^{(k)}(x)$ for $\EE[Y(0)\given X=x]$;
\STATE For all $i\in \cD^{(k)}$, compute $\hat\mu_1(X_i) = \hat\mu_1^{(k)}(X_i)$ and  $\hat\mu_0(X_i) = \hat\mu_0^{(k)}(X_i)$.
\ENDFOR
\STATE Compute estimator $\hat\theta_{\deb}$ or $\tilde\theta_\deb$ on $\cD$ according to~\eqref{eq:db_count} or~\eqref{eq:db_count_0}.
\end{algorithmic}
\end{algorithm}

Ignoring the nuisance in estimation and 
assuming $\hat\mu_w(x)\to \mu_w(x):=\EE[Y(w)\given X=x]$ for $w\in \{0,1 \}$,
one can show that the asymptotic variance of $\tilde\theta_\deb$ and $\hat\theta_\deb$ satisfies
\$
\Var(\tilde\theta_\deb) \approx  
\underbrace{\frac{1}{n }\Var\big(\mu_1(X ) - \mu_0(X )\big)}_{\textrm{(i) predictable part}} + \underbrace{\frac{1}{n_t}\Var\big( Y (1) - \mu_1(X ) \big) + \frac{1}{n_{c}} \Var\big( Y (0) - \mu_0(X )  \big)}_{\textrm{(ii) irreducible variance}}
\$
In the above decomposition, the (i) predictable part echoes our intuitions of imputing $Y_i(1)-Y_i(0)$ with predicted values, thereby inheriting the advantages of $\hat\theta_{\textrm{ideal}}$ of increased sample size and decreased single-term variance when the treatment effects are typically small. 
In the meantime, (i) is the variance of the projection of $Y_i(1)-Y_i(0)$ on the $X$-space, which is 
the best effort to predict $Y_i(1)-Y_i(0)$ with $X_i$.
The (ii) is the variance that cannot be eliminated by the information of $X$ embodied in the $n_t$ treated samples and $n_c$ control samples. 
 In other words, one might view $\tilde\theta_{\deb}$ as the best efforts towards exploiting the information in $X$. 
 As the variance in (i) is deflated by a larger factor $n$,  
 predictors with smaller variance in (ii), i.e., more accurate ML predictors $\mu_w(x)$, are desirable. 
\begin{remark}
For the case of count metrics, our estimator is 
asymptotically equivalent to the AIPW estimator \citep{robins1994estimation}, whose
semiparametric efficiency \citep{hahn1998role} directly implies the 
optimality of our estimator. 
As will be shown in Section \ref{subsec:ct_unbias} and \ref{subsec:ct_optimal}, the differences here include 1) we establish 
valid inference as long as the ML estimator 
converges to deterministic functions 
without consistency, 
and 2) we establish the optimality under a 
weaker condition of $L_2$ convergence in probability. 
Our debiasing term 
is also related to the ``prediction unbiasedness'' condition in~\cite{guo2021generalized}.
\end{remark}

When we restrict $\hat\mu_w$ in Algorithm~\ref{alg:debias} to be a linear function of $x$, 
our estimator is connected to the well-known CUPED estimator~\citep{deng2013improving}
\#\label{eq:vanilla_cuped}
\hat\theta_{\textrm{CUPED}} = \frac{1}{n_t}\sum_{i~\textrm{treated}} \big(Y_i(1) - \hat\theta (X_i- \bar{X} )\big) - \frac{1}{n_c} \sum_{i~\textrm{control}} \big(Y_i(0) - \hat\theta (X_i- \bar{X} )\big),
\#
where $\hat\theta$ is the OLS projection of $Y_i$ (the pooled outcomes of control and treated groups) on (centered) pre-treatment metric $X_i$. The connection can be easily established as follows. By replacing $\mu_w(x)$ with the linear predictor $\beta_w^\top x$ ($w$ = 0 or 1) where $\beta_w$ is the least-square linear coefficient of $Y(w)$ on $X$, the proposed optimal estimator
$\tilde\theta_\deb$ in (\ref{eq:db_count_0}) can be simplified as
\$
\tilde\theta_\deb 
&= \frac{1}{n_t}\sum_{i~\textrm{treated}} \big( Y_i(1) - \beta_1^\top (X_i - \bar{X}) \big) - \frac{1}{n_{c}} \sum_{i~\textrm{control}} \big( Y_i(0) - \beta_0^\top (X_i - \bar{X})  \big). 
\$
This improves upon CUPED by running separate regressions in the treated and control groups. This estimator enjoys the agnostic property~\citep{lin2013agnostic}: it leads to no larger variance than the diff-in-mean estimator without any assumptions, while the vanilla CUPED in~\eqref{eq:vanilla_cuped} may not. In such low-dimensional case, due to the low complexity of linear function classes, the double-dipping bias is not a concern and cross-fitting is not necessary.
But for high-dimensional linear regression such as the LASSO, 
it is important to use our proposed algorithm to cross-fit and debias with $\hat\mu$ being the LASSO estimator, otherwise the bias induced by 
the high dimensional regression could be of the same order as the variance.

\subsection{Unbiasedness and asymptotic inference} \label{subsec:ct_unbias} 
In this part, we establish the finite-sample unbiasedness and asymptotic inference for our procedures for count metrics.
In the following theorem, we show that $\hat\theta_{\deb}$ is unbiased and 
$\tilde\theta_\deb$ has a negligible $O(1/n)$ bias due to potentially
unequally-sized folds. 
Such finite-sample unbiasedness holds for any machine learning regressors.

\begin{theorem}[Finite-sample unbiasedness]\label{thm:deb_unbiased}
Suppose $(X_i,Y_i(0),Y_i(1))\iid\PP$  
are independent of the treatment assignments $\cT = \{T_i\}_{i=1}^n$. Then $\EE[\hat\theta_{\deb}\given \cT] = \tau$ 
for any $n$. Furthermore,
suppose there exists an absolute constant $c_0>0$ such that 
$\big|\EE[\hat\mu_w^{(k)}(X_i)\given \cD^{(-k)}]\big|\leq c_0$ 
and $\big|\EE[Y_i(w)]\big|\leq c_0$ for all $w\in\{0,1\}$.  
Then $\big|\EE[\tilde\theta_{\deb}\given \cT] - \tau \big| \leq c/\min\{n_t,n_c\}$ for some absolute constant $c>0$.
\end{theorem}



We now study the asymptotic inference for our estimators. 
We assume the convergence of $\hat\mu_1$, $\hat\mu_0$ to fixed functions (this is mild because we do not require it to converge to the true conditional mean functions) and the treatment assignment mechanism. 

\begin{assumption}[Convergence]\label{assump:conv}
There exists two fixed functions $\mu^*_1(\cdot)$ and $\mu^*_0(\cdot)$, so that $\|\hat\mu^{(k)}_1 - \mu^*_1\|_2 \stackrel{P}{\to}0$ and $\|\hat\mu^{(k)}_0 - \mu^*_0\|_2 \stackrel{P}{\to}0$ for all $k\in[K]$.
\end{assumption}

\begin{assumption}[Treatment assignment mechanism]\label{assump:treat}
Assume $n_t/n\stackrel{P}{\to}p$ for some fixed $p\in(0,1)$, so that $n_{k,t}/n_k \stackrel{P}{\to } p$ for all $k\in[K]$.
\end{assumption}

\begin{theorem}[Asymptotic confidence intervals]\label{thm:asymp_var}
Suppose Assumptions~\ref{assump:conv} and~\ref{assump:treat} hold. Then $\sqrt{n}(\hat\theta_\deb - \tau) \stackrel{d}{\to} N(0,\sigma_\deb^2)$ and $\sqrt{n}(\tilde\theta_\deb - \tau) \stackrel{d}{\to} N(0,\sigma_\deb^2)$, where 
$
\sigma_{\deb}^2 = \frac{1}{p}\Var\big(Y_i(1) - (1-p) \mu_1^*(X_i) - p \mu_0^*(X_i)\big) + \frac{1}{1-p}\Var\big( Y_i(0) - (1-p) \mu_1^*(X_i) - p \mu_0^*(X_i)\big).
$
Furthermore, define the variance estimator 
$
\hat \sigma_\deb^2 = \frac{n}{n_t^2} \sum_{i=1}^nT_i (A_i - \bar{A} )^2 +  \frac{n}{n_c^2} \sum_{i=1}^n(1-T_i)  (B_i - \bar{B} )^2,
$ 
where $\bar{A} = \frac{1}{n_t} \sum_{i=1}^nT_i  A_i$, $\bar{B} = \frac{1}{n_c} \sum_{i=1}^n(1-T_i) B_i $, and 
$
A_i = Y_i(1) - \frac{n_c}{n} \hat\mu_1 (X_i) - \frac{n_t}{n }\hat\mu_0 (X_i),
$  
$
B_i = Y_i(0) - \frac{n_c}{n} \hat\mu_1(X_i) - \frac{n_t}{n} \hat\mu_0(X_i). 
$
Then $\hat\sigma_{\deb}^2\stackrel{P}{\to}\sigma_{\deb}^2$, 
and $\hat\theta_\deb \pm z_{1-\alpha/2} \hat\sigma_\deb/\sqrt{n}$ and $\tilde\theta_\deb \pm z_{1-\alpha/2} \hat\sigma_\deb/\sqrt{n}$ are both asymptotically valid $(1-\alpha)$ confidence intervals for $\tau = \EE[Y(1)]-\EE[Y(0)]$.
\end{theorem}

\subsection{Optimality: semiparametric efficiency} \label{subsec:ct_optimal}

Returning to our motivating question, we establish the optimality of our procedure, building upon semiparametric statistics theory~\citep{bickel1993efficient,robins1994estimation,hahn1998role}. We are to show that the asymptotic variance $\sigma_\deb^2$ is no larger than any `regular' estimators (roughly speaking, those asymptotically linear ones with the form $\frac{1}{n}\sum_{i=1}^n \phi(X_i,Y_i,T_i)+o_P(1/\sqrt{n})$ for some function $\phi$,  including the ones obtained from linear regression as in~\cite{deng2013improving,guo2021machine,cohen2020no}). 
Due to the limit of paper length, we omit the formal backgrounds on semiparametric efficiency in the main text and defer the discussions to 
Appendix~\ref{app:var_bd_background}, 
which includes notions of regular non-parametric space and efficient influence functions. 

We impose the following condition on the consistency of the estimators $\hat\mu_1(\cdot),~\hat\mu_0(\cdot)$.

\begin{assumption}[Consistency]\label{assump:consist}
Let $\mu_1(x) = \EE[Y(1)\given X=x]$ and $\mu_0(x) = \EE[Y(0)\given X=x]$ be the two true mean functions. Suppose $\|\hat\mu^{(k)}_w - \mu_w\|_2 \stackrel{P}{\to}0$ for $w\in\{0,1\}$ and all $k\in[K]$.
\end{assumption}

Assumption~\ref{assump:consist} is a mild condition 
without any requirement of the convergence rates. 
This is in contrast to 
the conditions on convergence rates (though doubly-robust properties help relax them to lower orders) for observational studies \citep{nichols2007causal,schuler2017targeted,chernozhukov2018double} 
where the propensity score function $e(x) = \PP(T=1\given X=x)$ also needs to be estimated. Also, we do not need the convergence to be point-wise. 



\begin{theorem}[Semiparametric efficiency]
\label{thm:efficiency}
Suppose Assumptions~\ref{assump:treat} and~\ref{assump:consist} hold. Then the asymptotic variance of $\hat\theta_\deb$ and $\tilde\theta_\deb$ in Theorem~\ref{thm:asymp_var} 
is the semiparametric  variance bound for $\tau = \EE[Y(1)] - \EE[Y(0)]$. 
\end{theorem}

Theorem~\ref{thm:efficiency}, together with Theorem~\ref{thm:asymp_var},
indicates that the covariates $X_i$ should 
be chosen to be powerful predictors for the outcomes, 
such that $\Var(Y(w)-\mu_w(X))$ is relatively small for $w\in \{0,1\}$. 
Since our estimator is finite-sample unbiased conditional on $\cT$, if we view $\cT$ as fixed, 
it has the smallest variance among all asymptotically unbiased regular estimators. 
This result might be of interest for two-sample mean test 
and completely randomized experiments as well; 
in the latter case, 
units in the treated and control 
groups still have i.i.d.~potential outcomes 
if we assume the units are i.i.d.~before treatment assignment.

%
We take a moment here to compare Theorem~\ref{thm:efficiency} 
to other machine-learning empowered methods in the literature~\citep{cohen2020no,guo2021machine,hosseini2019unbiased}. 
As indicated by~\eqref{eq:db_count} and~\eqref{eq:db_count_0}, our estimator is essentially a linear combination of $Y_i$, $\hat\mu_1(X_i)$ and $\hat\mu_0(X_i)$ that achieves smallest variance. 
Thus, it can be obtained from a linear regression 
when Assumption~\ref{assump:consist} holds. 
Referring to~\eqref{eq:db_count} and~\eqref{eq:db_count_0}, for optimality, one needs to include both $\hat\mu_1(X_i)$ and $\hat\mu_0(X_i)$ in the regression terms and run the regression separately on treated and control groups. Therefore, although machine learning helps to  exploit nonlinearity, the regression in \cite{guo2021machine} 
with one single predictor for $Y_i$ (not separately for $Y_i(1)$ and $Y_i(0)$) would not achieve optimality in general, unless $\mu_1(x)$ and $\mu_0(x)$ are completely colinear. This issue is the same for \cite{hosseini2019unbiased} where only one predictor is used. 
The method in \cite{cohen2020no} may be optimal, 
which, however, 
requires a slightly stronger condition of $L_4$-distance convergence.

%% file: ratio_rr.tex

We now 
study variance reduction procedures for ratio metrics introduced in Section~\ref{subsec:metric} whose denominator $Z_i=Z_i(T_i)$ can  be  changed  by  the  treatment (without SDA). 
The  results are new to the literature while sharing similar ideas to 
our results for count metrics.

\subsection{Estimation procedure}
The first step  
is  the $K$-fold sample splitting for $\cD = (Y_i,Z_i,T_i,X_i)_{i=1}^n$
as introduced in Section~\ref{subsec:count_procedure}. 
The second step, cross-fitting, needs 
careful consideration: for each $k\in[K]$, we use the data $\big\{(X_i,Z_i ,Y_i)\colon T_i=1, i\in\cD^{(-k)} \big\}$ to obtain estimators $\hat\mu_1^{Y,(k)}(x)$ for $\EE[Y(1)\given X=x]$ and $\hat\mu_1^{Z,(k)}(x)$ for $\EE[Z(1)\given X=x]$. Likewise, we use $\big\{(X_i,Z_i,Y_i)\colon T_i=0, i\in\cD^{(-k)} \big\}$ to obtain $\hat\mu_0^{Y,(k)}(x)$ and $\hat\mu_0^{Z,(k)}(x)$. Then, we 
calculate predictions $\hat\mu_w^Y(X_i) = \hat\mu_w^{Y,(k)}(X_i)$ and  $\hat\mu_w^Z(X_i) = \hat\mu_w^{Z,(k)}(X_i)$ for all $i\in \cD^{(k)}$, $w\in \{0,1\}$. 
Finally, we estimate $\delta$ in~\eqref{eq:rr_target} by
\#\label{eq:est_delta}
 \hat\delta &= \frac{ \sum_{i=1}^n A_i  }{ \sum_{i=1}^n B_i}- \frac{ \sum_{i=1}^n C_i  }{ \sum_{i=1}^n D_i}, 
\#
where $A_i = \hat\mu_1^Y(X_i ) + \frac{T_i}{\hat p}  \big(Y_i - \hat\mu_1^Y(X_i ) \big)$, 
$B_i =  \hat\mu_1^Z(X_i ) + \frac{T_i}{\hat p}  \big(Z_i - \hat\mu_1^Z(X_i ) \big)$,
$C_i =  \hat\mu_0^Y(X_i ) + \frac{1-T_i}{1-\hat p}  \big(Y_i - \hat\mu_0^Y(X_i ) \big)$,  and 
$D_i =  \hat\mu_0^Z(X_i ) + \frac{1-T_i}{1-\hat p}  \big(Z_i - \hat\mu_0^Z(X_i ) \big)$ for  $\hat p = n_t/n$.  
The procedure is summarized in Algorithm~\ref{alg:ratio}. 
Compared to  
$
\hat\delta_{\dm} := 
\frac{ \sum_{i~\textrm{treatment}}Y_i}{ \sum_{i~\textrm{treatment}} Z_i}
-
\frac{ \sum_{i~\textrm{control}}Y_i}{ \sum_{i~\textrm{control}} Z_i},
$
our estimator 
substitutes the sample means of  the treated and control groups with average of the fit-and-debias predictions for all $n$ units, which is similar to our estimator for count metrics.

\begin{algorithm}[h]
\caption{Debiased Variance Reduction for Ratio Metric without SDA}\label{alg:ratio}
\begin{algorithmic}[1]
\STATE Input: Dataset $\cD =\{(Y_i,X_i,Z_i,T_i\}_{i=1}^{n}$, number of folds $K$.
\STATE Randomly split $\cD$ into $K$ folds $\cD^{(k)}$, $k=1,\dots,K$.
\FOR{$k=1,\dots,K$} 
\STATE Use all $(X_i,Z_i,Y_i)$ with $T_i=1$ and $i\notin \cD^{(k)}$ to obtain $\hat\mu_1^{Y,(k)}(x)$ and $\hat\mu_1^{Z,(k)}(x)$; 
\STATE Use all $(X_i,Z_i,Y_i)$ with $T_i=0$ and $i\notin \cD^{(k)}$ to obtain $\hat\mu_0^{Y,(k)}(x)$ and $\hat\mu_0^{Z,(k)}(x)$;
\STATE Compute $\hat\mu_w^Y(X_i) = \hat\mu_w^{Y,(k)}(X_i)$ and  $\hat\mu_w^{Z}(X_i) = \hat\mu_w^{Z,(k)}(X_i)$ for all $i\in \cD^{(k)}$ and $w\in\{0,1\}$.
\ENDFOR
\STATE Compute estimator $\hat\delta$ on $\cD$ according to~\eqref{eq:est_delta}.
\end{algorithmic}
\end{algorithm}

\subsection{Asymptotic inference}

The analysis of ratio metrics is naturally asymptotic (\cite{deng2017trustworthy}), and thus 
we focus more on the asymptotic properties of the proposed estimator. 

We impose the following conditions on the treatment assignment mechanism and convergence of cross-fitted functions. In Assumption~\ref{assump:rr_conv}, we only require the convergence of estimated functions to deterministic functions, not the true conditional mean functions. This is a mild condition that holds for general machine learning regression methods. 

\begin{assumption}[Data Generating Process]
  \label{assump:rr_data}
  $(X_i,Z_i(0),Z_i(1),Y_i(0),Y_i(1))\iid \PP$ and the treatment assignments $T_i\iid \textrm{Bernoulli}(p)$ are independent of all other random variables. 
\end{assumption}

\begin{assumption}[Convergence]
\label{assump:rr_conv}
There exists some fixed fuctions $\mu_{1,Y}^*(\cdot)$, $\mu_{0,Y}^*(\cdot), \mu_{1,Z}^*(\cdot)$, $\mu_{0,Z}^*(\cdot)$, so that both $\|\hat\mu_{w}^{Y,(k)} - \mu_{w,Y}^*\|_2, \|\hat\mu_{w}^{Z,(k)} - \mu_{w,Z}^*\|_2\stackrel{P}{\to} 0$ for $w\in \{0,1\}$. 
\end{assumption}

In preparation for  inferential guarantees, 
we define the influence function 
\#
\phi_\delta(Y_i,Z_i,X_i,T_i) &= \frac{A_i^*}{\EE[Z_i(1)]} - \frac{\EE[Y_i(1)]}{\EE[Z_i(1)]^2} B_i^* - \frac{C_i^*}{\EE[Z_i(0)]} + \frac{\EE[Y_i(0)]}{\EE[Z_i(0)]^2} D_i^*,\label{eq:rr_var}
\# 
where $A_i^* =  \tilde\mu_{1,Y}^*(X_i ) + \frac{T_i}{p}  \big(Y_i - \EE[Y_i(1)]- \tilde\mu_{1,Y}^*(X_i ) \big),$
$B_i^*  =  \tilde\mu_{1,Z}^*(X_i ) + \frac{T_i}{p}  \big(Z_i -\EE[Z_i(1)] - \tilde\mu_{1,Z}^*(X_i ) \big)$,
$C_i^*  =  \tilde\mu_{0,Y}^*(X_i ) + \frac{1-T_i}{1-p}  \big(Y_i  - \EE[Y_i(0)] -  \tilde\mu_{0,Y}^*(X_i ) \big)$,
and $D_i^* =   \tilde\mu_{0,Z}^*(X_i ) + \frac{1-T_i}{1-p}  \big(Z_i - \EE[Z_i(0)]- \tilde\mu_{0,Z}^*(X_i ) \big)$. 
Here $\tilde \mu_{w,Y}^*(X_i) = \mu_{w,Y}^*(X_i) - \EE[\mu_{w,Y}^*(X_i)]$ and $\tilde \mu_{w,Z}^*(X_i) = \mu_{w,Z}^*(X_i) - \EE[\mu_{w,Z}^*(X_i)]$, $w\in \{0,1\}$ are the centered limiting functions. The following theorem establishes the asymptotic confidence intervals.

\begin{theorem}\label{thm:rr_inference}
Suppose Assumptions~\ref{assump:rr_data} and~\ref{assump:rr_conv} hold, and let $\hat\delta$ be the output of Algorithm~\ref{alg:ratio}. Then $\sqrt{n}(\hat\delta - \delta) \stackrel{d}{\to} N(0,\sigma_\delta^2)$, where $\sigma_\delta^2 = \Var(\phi_\delta(Y_i,Z_i,X_i,T_i))$. 
Moreover, define the variance estimator 
$\hat\sigma_\delta^2 = \frac{1}{n } \sum_{T_i=1} \big(d_{1,i}  - \bar{d}_1\big)^2 + \frac{1}{n} \sum_{T_i=0} \big(d_{0,i}  - \bar{d}_0\big)^2$, where 
$
d_{1,i} = -\frac{n_c}{n_t \bar{Z}(1)}  \hat\mu_{1}^Y(X_i) + \frac{n}{n_t \bar{Z}(1)} Y_i
 + \frac{n_c}{n_t} \frac{\bar{Y}(1)}{\bar{Z}(1)^2}   \hat \mu_{1}^Z (X_i) 
 - \frac{n}{n_t} \frac{\bar{Y}(1)}{\bar{Z}(1)^2}  Z_i
 - \frac{1}{\bar{Z}(0)}  \hat\mu_{0}^Y (X_i) + \frac{\bar{Y}(0)}{\bar{Z}(0)^2}  \hat\mu_{0}^Z (X_i),
$ and
$
d_{0,i} =  \frac{1}{\bar{Z}(1)}  \hat\mu_{1}^Y   (X_i) - \frac{\bar{Y}(1)}{\bar{Z}(1)^2}  \hat\mu_{1 }^Z(X_i) - \frac{n}{n_c\bar{Z}(0)} Y_i
+ \frac{n_t}{n_c \bar{Z}(0)} \hat\mu_0^Y(X_i) + \frac{n }{n_c} \frac{\bar{Y}(0)}{\bar{Z}(0)^2} Z_i
- \frac{n_t}{n_c}\frac{\bar{Y}(0)}{\bar{Z}(0)^2} \hat\mu_0^Z(X_i).
$
Then $\hat\delta \pm \hat\sigma_\delta \cdot z_{1-\alpha/2}/\sqrt{n}$ is an asymptotically valid $(1-\alpha)$ confidence interval for $\delta$.
\end{theorem}


\subsection{Optimality}
Our estimator is optimal (semiparametric efficient) when the estimated functions are consistent. 
We begin with  mild convergence assumptions. 



\begin{assumption}\label{assump:rr_consistency}
$\|\hat\mu^{Y,(k)}_w - \mu_{w,Y}\|_2 \stackrel{P}{\to}0$ and $\|\hat\mu^{Z,(k)}_0 - \mu_{w,Z}\|_2 \stackrel{P}{\to}0$ for all $k\in[K]$ and $w\in \{0,1\}$, where 
$\mu_{w,Y}(x) = \EE[Y(w)\given X=x]$ and $\mu_{w,Z}(x) = \EE[Z(w)\given X=x]$.
\end{assumption}

We are to show that the efficient influence function for the estimation of $\delta$ is given by  
\#\label{eq:rr_eff_infl}
\phi_{\delta}^\dag&(Y_i,Z_i,X_i,T_i) = \frac{A_i^{\dag}}{\EE[Z_i(1)]} - \frac{\EE[Y_i(1)]}{\EE[Z_i(1)]^2} B_i^{\dag} - \frac{C_i^{\dag}}{\EE[Z_i(0)]} + \frac{\EE[Y_i(0)]}{\EE[Z_i(0)]^2} D_i^{\dag},
\#
where
$A_i^{\dag}  =   \mu_{1,Y} (X_i ) + \frac{T_i}{p}  \big(Y_i  - \mu_{1,Y} (X_i ) \big) - \EE[Y_i(1)]$, 
$B_i^{\dag} =   \mu_{1,Z} (X_i ) + \frac{T_i}{p}  \big(Z_i  - \mu_{1,Z} (X_i ) \big)-\EE[Z_i(1)]$, 
$C_i^{\dag}  = \mu_{0,Y}^*(X_i )  + \frac{1-T_i}{1-p}  \big(Y_i   -  \mu_{0,Y} (X_i ) \big)- \EE[Y_i(0)]$, 
$D_i^{\dag} = \mu_{0,Z}^*(X_i ) + \frac{1-T_i}{1-p}  \big(Z_i -  \mu_{0,Z} (X_i ) \big)- \EE[Z_i(0)]$. 
The following theorem shows that under consistency, the asymptotic variance of $\hat\delta$ coincides with the variance of $\phi_{\delta,\dag}$, which is also the efficient variance bound. Thus, the optimality of the proposed procedure is established under appropriate conditions. 

\begin{theorem}\label{thm:rr_optimal}
Suppose Assumptions~\ref{assump:rr_data} and~\ref{assump:rr_consistency} hold. Then it holds that $\sqrt{n}(\hat\delta - \delta)\stackrel{d}{\to} N(0,\sigma_{\delta,\dag}^2)$, where 
$
\sigma_{\delta,\dag}^2 = \Var\big( \phi_{\delta,\dag}(Y_i,Z_i,X_i,T_i)\big) 
$ (c.f.~\eqref{eq:rr_var})
is the semiparametric asymptotic variance bound for $\delta = \EE[Y(1)]/\EE[Z(1)] - \EE[Y(0)]/\EE[Z(0)]$.  
\end{theorem}

Theorem~\ref{thm:rr_optimal} is based on
a more general result (Theorem~\ref{thm:rr_var_bound} in the appendix) in situations 
where the treatment assignment might  depend on $X$. 
It may be of independent interest for 
efficient estimation of ratio metrics in 
stratified experiments and observational studies. 


\subsection{Special case: optimal linear adjustment for ratio metrics} \label{subsec:linear_rr}

As a special case of Alg.~\ref{alg:ratio}, 
we derive an optimal \emph{linear adjustment} method 
for ratio metrics, which is computationally efficient and has several advantages over the existing approaches in the literature.

We suppose $X_i\in \RR^p$ for some fixed $p$ (with intercept). In Alg.~\ref{alg:ratio},
let $\hat\mu_w^{Y,(k)}(x) = \hat\beta_{Y,w,(k)}^\top x$ for all $w\in\{0,1\}$ and $k=1,\dots,K$, where $\hat\beta_{Y,w,(k)}$ is the OLS coefficient of $\{Y_i\colon T_i=w,i\in \cD^{(-k)}\}$ on $\{X_i\colon T_i=w,i\in \cD^{(-k)}\}$. The fitted function $\hat\mu_w^{Z,(k)}(x)$ can be similarly obtained. 
As we have discussed at the end of Section~\ref{subsec:count_procedure},
in this fixed-$p$ setting, because of 
the low complexity of linear function classes,
  the cross-fitting step may not be needed and we can further simplify the approach by letting $\hat\mu_w^Y(x) = \hat\beta_{Y, w}^\top x$ in Algorithm~\ref{alg:ratio}, 
 where $\hat\beta_{Y,w}$ is the empirical OLS coefficient of $\{Y_i\colon T_i=w\}$ 
 on $\{X_i\colon T_i=w\}$ for $w\in \{0,1\}$, i.e.,
 running linear regressions separately on treated and control groups 
 without sample splitting. 
 The imputations $\hat\mu_w^Z(x)$ can be similarly obtained. 
 One can show that the two estimators with or without sample splitting are asymptotically equivalent due to the convergence property of linear regression coefficients. 
Both estimators admit the asymptotic linear expansion (up to additive constants)
\$
\frac{1}{n} \sum_{i=1}^n \Big( \alpha_{1,*}^\top X_i - \alpha_{0,*}^\top X_i + \frac{T_i}{p}\big(\Gamma_i - \alpha_{1,*}^\top X_i )   - \frac{1-T_i}{1-p}\big(\Gamma_i - \alpha_{0,*}^\top X_i )  \Big) + o_P(1/\sqrt{n}),
\$
where $\Gamma_i = \frac{T_i Y_i}{\EE[Z(1)]} - \frac{T_i Z_i \EE[Y(1)]}{\EE[Z(1)]^2} - \frac{(1-T_i) Y_i}{\EE[Z(0)]} + \frac{(1-T_i) Z_i \EE[Y(0)]}{\EE[Z(0)]^2} $, 
and $\alpha_{1,*}$, $\alpha_{0,*}$ are the population OLS coefficients 
of $\Gamma_i$ on $X_i$ in treated and control groups. 
Because of the ``agnostic'' property of linear adjustment~\citep{lin2013agnostic}, 
the variance of our estimator is always no larger than that of the diff-in-mean estimator. 
Moreover, our estimator achieves 
semiparametric efficiency if the actual conditional means are linear. 
Our linear adjustment approach improves  
in several aspects upon 
the ratio-metric extention of CUPED~\citep{deng2013improving}: 
\#\label{eq:rr_cuped}
\frac{\sum_{T_i=1}Y_i}{\sum_{T_i=1}Z_i} -  \frac{\sum_{T_i=0}Y_i}{\sum_{T_i=0}Z_i} -  \hat\theta \cdot \Bigg( \frac{\sum_{T_i=1}\tilde Y_i}{\sum_{T_i=1}\tilde Z_i} -\frac{\sum_{T_i=0}\tilde Y_i}{\sum_{T_i=0}\tilde Z_i} \Bigg)
\#
for some $\hat\theta\in \RR$, where $\tilde{Y}_i$ and $\tilde{Z}_i$ are pre-treatment versions of $Y_i$ and $Z_i$. 

\begin{enumerate}[(i)]
\item 
Firstly, the estimator~\eqref{eq:rr_cuped} only works with pre-treatment metrics as covariates, while our method incorporates arbitrary covariates, hence more flexible and powerful. 
\item
Secondly, our method provably reduces variance compared to the diff-in-mean estimator because it runs
separate regressions in the treated and control groups. 
A single $\hat\theta$ in~\eqref{eq:rr_cuped} is not guaranteed to reduce variance in some adversarial settings~\citep{lin2013agnostic}.
\item 
Thirdly, the linear expansion of the estimator in~\eqref{eq:rr_cuped} (assuming $\hat\theta\stackrel{P}{\to} \theta$) is 
$
\frac{1}{n} \sum_{i=1}^n \big[ \frac{T_i}{p}\big(\Gamma_i - \theta \tilde\Gamma_i\big) - \frac{1-T_i}{1-p} \big(\Gamma_i - \theta\tilde\Gamma_i \big) \big]+ o_P(1/\sqrt{n}),
$
where $\tilde\Gamma_i$ is similarly defined as $\Gamma_i$ with $Y_i,Z_i$ replaced by $\tilde{Y}_i,\tilde{Z}_i$, and $\theta$ is the OLS coefficient of $\Gamma_i$ on $\tilde\Gamma_i$.
Our estimator achieves the OLS projection of $\Gamma_i$ on the whole linear space of $X_i$, 
while~\eqref{eq:rr_cuped} only projects on the linear space of $\tilde\Gamma_i$, 
a subspace of $X_i$ 
when $X_i$ contains $\tilde{Y}_i,\tilde{Z}_i$. Therefore, our estimator achieves more reduction of variance. 
\end{enumerate}


%% file: ratio_rcut.tex

In this section, we consider the second type of ratio metrics introduced in Section~\ref{subsec:metric} whose denominator $Z$ is assumed to be stable (SDA). 
Due to the page limit, 
we only describe the procedure here in Algorithm~\ref{alg:ratio'}
for the reference of practitioners. 
Inference and optimality guarantees 
can be found in Appendix~\ref{app:sda}.
We adopt the similar plug-in-and-debias idea 
as in Sections~\ref{sec:count} and~\ref{sec:rr},
while here we pool all samples to estimate $\EE[Z]$, 
and fit the conditional mean functions for $Y$ based on $(X,Z)$ separately in two groups. 

\begin{algorithm}[h]
\caption{Debiased Variance Reduction for Ratio Metric with Stable $Z$}\label{alg:ratio'}
\begin{algorithmic}[1]
\STATE Input: Dataset $\cD =\{(Y_i,X_i,Z_i,T_i\}_{i=1}^{n}$, number of folds $K$.
\STATE Randomly split $\cD$ into $K$ folds $\cD^{(k)}$, $k=1,\dots,K$.
\FOR{$k=1,\dots,K$} 
\STATE Use $\{(X_i,Z_i,Y_i)\colon T_i=1,i\notin \cD^{(k)}\}$ to obtain $\hat\mu_1^{(k)}(x,z)$ for $\EE[Y(1)\given X=x,Z=z]$; 
\STATE Use  $\{(X_i,Z_i,Y_i)\colon T_i=0,i\notin \cD^{(k)}\}$ to obtain $\hat\mu_0^{(k)}(x,z)$ for $\EE[Y(0)\given X=x,Z=z]$;
\STATE Compute $\hat\mu_1(X_i,Z_i) = \hat\mu_1^{(k)}(X_i,Z_i)$ and  $\hat\mu_0(X_i) = \hat\mu_0^{(k)}(X_i,Z_i)$ for all $i\in \cD^{(k)}$.
\ENDFOR
\STATE Compute $\hat p= n_t/n$ and $\Gamma_i = \hat\mu_1(X_i,Z_i) - \hat\mu_0(X_i,Z_i)  + \frac{T_i}{\hat p}  \big(Y_i - \hat\mu_1(X_i,Z_i) \big) - \frac{1-T_i}{1-\hat p}  \big(Y_i - \hat\mu_0(X_i,Z_i)\big)$ for all $i\in \cD$.
\STATE Compute estimator $\hat\delta' = \frac{ \sum_{i=1}^n \Gamma_i  }{ \sum_{i=1}^n Z_i}$.
\end{algorithmic}
\end{algorithm}

%% file: simu.tex

We conduct simulations to demonstrate the performance 
of all proposed optimal variance reduction procedures, 
based on which we offer practical suggestions.

\subsection{Count metrics}\label{subsec:simu_count}

In this section, we use simulations to validate  
whether Alg.~\ref{alg:debias} can outperform state-of-the-art
methods when $Y$ and $X$ have a nonlinear relationship,
while achieving comparable performance when the relationship is indeed linear (in which case linear adjustment is optimal). 

We design one nonlinear and one linear data-generating processes where $X\in \RR^d$
for $d\in\{10,100\}$.
The sample size is fixed at $n=10000$ and we generate $X_i\iid N(0,I_d)$ for $d\in\{10,100\}$ except for a categorical variable $X_6\sim \text{Unif}\{1,2,\dots,10\}$ to represent both continuous and categorical covariates. 
Let treatments $T_i\iid$ Bernoulli$(0.5)$ and the outcomes are generated by
$
Y_i = b(X_i) + T_i \cdot \tau(X_i) + \epsilon_i,
$
where $\epsilon_i \iid N(0,1)$ is the random noise. 
In the nonlinear setting  
inspired by~\cite{guo2021machine,friedman1991multivariate}, 
we define the conditional treatment effect $\tau(x) =10 x_1 +5 \log(1+\exp(x_2))
+ \ind\{x_6\in\{1,5,9\}\}$
and baseline  $b(x) = 10\sin(\pi \cdot x_1x_2) + 20(x_3-0.5)^2 + 10 x_4 + 5 \ind\{x_6\in\{1,5,9\}\}$.
In the linear setting, 
we specify  $b(x) = \beta^\top x $ and 
$\tau(x) = 1+ \delta^\top x+ \ind\{x_6\in\{1,5,9\}\}$, where 
$
\beta = (5.31, 1.26, 3.12, -0.85,0,\dots,0)^\top,$ $\delta = (1.26,-3.14,0,\dots,0)^\top \in \RR^d. 
$
The ground truth 
is $\tau=4.34$ for the nonlinear setting 
and $\tau=0.303$ for the linear setting. 
In $N=1000$ independent runs, 
we evaluate the estimated 
standard deviation and 
the empirical coverage of the 0.95-confidence interval. 
Valid empirical coverage would certify the validity of the inference procedure,
under which smaller estimated standard deviation indicates shorter confidence intervals and higher efficiency.

We compare our method to the diff-in-mean (DiM) estimator 
and the popular CUPED estimator.
To demonstrate the performance with different types of machine learning algorithms, 
we use the 
random forest regressor (RF), gradient boosting (GB) 
and neural networks (NN, two layers) 
from \texttt{scikit-learn} python library, 
all without manual model tuning. 
The CUPED method we implement 
is equivalent to~\cite{lin2013agnostic} which uses
multivariate
covariates and runs
separate regressions in the treated and control groups.
It has better performance than 
the vanilla version of~\cite{deng2013improving}, 
and is  optimal in the  linear setting.
Results averaged over $1000$ replicates are summarized in Table~\ref{tab:simu}. 



\begin{table}[htbp]
\centering
{\footnotesize
\begin{tabular}{c|c |c | c|c |c |c|c |c|c}
\toprule
\multirow{2}{*}{Setting} 
&    \multicolumn{6}{c|}{Algorithm~\ref{alg:debias} } 
&\multicolumn{2}{c|}{ CUPED } 
  &   DiM   \\
\cline{2-10}
&  \multicolumn{3}{c|}{ Var.Red\%.} & \multicolumn{3}{c|}{ Emp.Cov. }
& \multirow{2}{*}{Var.Red\%.} &  \multirow{2}{*}{Emp.Cov.} 
  & \multirow{2}{*}{Emp.Cov.} 
\\
\cline{2-7}
& RF & GB & NN & RF & GB & NN & & &  \\
\hline 
Lin, $d=10$ &   \textbf{90.51 } 
&  91.91  & 91.80 
& 0.950  & 0.958   &  0.947 
& 92.12 & 0.941    &  0.951   \\
  \hline
Lin, $d=100$ &   \textbf{89.92} 
&   91.48  & 88.31 
& 0.951   & 0.944 & 0.951
& 92.12 
& 0.945   &  0.949   \\
  \hline
Nonlin, $d=10$ & \textbf{88.51}  
& 89.59   & 67.57 
&   0.953     &   0.943  & 0.949
&  34.11  & 0.944   &  0.945  \\
  \hline
Nonlin, $d=100$ & \textbf{88.00} 
& 90.42   & 73.48 
& 0.956   &   0.953  & 0.960
& 37.62 & 0.936     &  0.941   \\
\bottomrule
\end{tabular}
}%

\caption{ Variance reduction \% compared to DiM (Var. red\%.) and the 
empirical coverage of the 0.95-C.I.s (Emp.\,cov.) for all methods 
in the linear (Lin) and nonliner (Nonlin) settings. }
\label{tab:simu}
\end{table}

From Table~\ref{tab:simu}, the results for linear setting shows that our proposed Alg.~\ref{alg:debias} is valid
and achieves comparable efficiency as the optimal linear method.
The results for nonlinear setting shows Alg.~\ref{alg:debias} has much higher efficiency than the state-of-the-art 
linear method, as machine learning regressors are capable 
to capture nonlinear dependencies. 
The variance reduction of NN 
is smaller than that of RF and GB, which probably is because NN has lower prediction accuracy due to the lack of model tuning. 

\subsection{Ratio metrics}\label{subsec:simu_ratio}

Variance reduction for ratio metrics is rarely studied in the literature and existing benchmarks are scarce.
In this section, we conduct simulations to evaluate the performance of our 
methods and some alternative solutions.

We design two data generating processes,
one satisfying the SDA introduced in Section~\ref{subsec:metric} and one does not. 
In both settings, 
the  marginal distributions of $X_i$ and $Z_i$ 
are the same, 
and $Y_i =b(X_i,Z_i) + T_i \cdot \tau(X_i,Z_i) + \epsilon_i$ for 
$\epsilon_i\iid N(0,1)$, where we define 
$
b(x,z) = \big(1.5+\sin(\pi x_1 x_2)\big) \cdot z + 0.5 x_4^2$ and
$
\tau(x,z) = 0.5 z\cdot \big( x_1 + \log(1+e^{x_3})\big) 
+ 0.2 \ind\{x_6\in\{1,5,9\}\}.
$ 
The only difference is whether $Z_i$ is influenced by the treatment. 
In both settings, we generate $X_i\iid N(0,I_d)$ for $d\in\{10,100\}$ 
except for a categorical variable $X_6\sim \text{Unif}\{1,2,\dots,10\}$.
To illustrate the impact of imposing the SDA on $Z$, we test Alg.~\ref{alg:ratio} and~\ref{alg:ratio'} in both settings. 


Setting 1 satisfies the SDA, 
and 
$
\delta = \delta' =  \EE[Y(1)-Y(0)] / \EE[Z].
$
We generate  $Z_i=\log(1+\exp(1+X_{i,1})) +D_i ( 0.2X_{i,3}^2+0.1\ind\{x_6\in\{1,5,9\}\})$ for 
 $D_i\iid$ Bernoulli$(0.5)$ to ensure $Z_i>0$.    
We compare our Alg.~\ref{alg:ratio} and~\ref{alg:ratio'} to 
the difference-in-mean estimator 
$\hat\delta'_{\textrm{DiM}} = \big(\frac{1}{n_t}\sum_{T_i=1}Y_i -  \frac{1}{n_c}\sum_{T_i=0}Y_i\big) /\big(\frac{1}{n}\sum_{i=1}^nZ_i\big)$  
as well as $\hat\delta'_{\textrm{CUPED}} = \hat\theta_{\textrm{CUPED}}/\big(\frac{1}{n}\sum_{i=1}^nZ_i\big)$, 
where $\hat\theta_{\textrm{CUPED}}$ is the 
linear adjustment estimator for count metrics (multivariate and separate regression version). 
Both of them pool all $Z_i$ to estimate the denominator and 
apply methods for count metrics to estimate the numerator,
hence comparable to Alg.~\ref{alg:ratio'}. 
The true estimands are 
$\delta = \delta'= 0.641$. 

Setting 2 does not satisfy the SDA, and 
the two estimands $\delta$ and $\delta'$ are distinct.
The only difference to setting 1 is that we generate 
$Z_i = \log(1+\exp(X_{i,1})) + T_i (0.2 X_{i,3}^2+0.1\ind\{x_6\in\{1,5,9\}\})$, 
where we replace $D_i$ with $T_i$ 
so that $Z_i$ is a realized potential outcome. 
As the two algorithms are not directly comparable,
our evaluation focuses  on Alg.~\ref{alg:ratio}, 
while we still evaluate Alg.~\ref{alg:ratio'} 
to show the consequences of erroneously assuming SDA. 
We compare these two nonlinear algorithms to the difference-in-mean estimator 
$\hat\theta_{\textrm{DiM}} = \sum_{T_i=1}Y_i /\sum_{T_i=1}Z_i -  \sum_{T_i=0}Y_i /\sum_{T_i=0}Z_i$ as well as our linear adjustment method proposed
in Section~\ref{subsec:linear_rr}. 
(Note that we do not include the ratio metric extension of the CUPED method outlined in the appendix of \cite{deng2013improving} because it cannot incorporate covariates other than pre-treatment metrics.)
The true estimands are $\delta=0.463$ and 
$\delta'=0.740$.

In both settings, our procedures are 
implemented with random forest,
gradient boosting, and neural network 
from \texttt{scikit-learn} python library. 
We evaluate the estimated variance and 
the empirical coverage of the $0.95$-confidence intervals over $N=1000$ independent runs with sample size $n=10000$, and we also assess the proportions of reduced variance compared 
to the diff-in-mean estimator, as well as 
the coverage of confidence intervals for estimands 
(i.e., $\delta$ for $\hat\delta$, while $\delta'$ for $\hat\delta'$). 
Due to the page limit, Table~\ref{tab:simu_ratio} only 
summarizes the results with random forest, while others follow similar patterns.

\begin{table}[h]
\centering
{\footnotesize
\renewcommand\arraystretch{1.0}
\begin{tabular}{c|c |c | c|c |c|c  }
\toprule
\multirow{2}{*}{Setting} &    \multicolumn{2}{c|}{Algorithm~\ref{alg:ratio'}} &  \multicolumn{2}{c|}{Algorithm~\ref{alg:ratio}} &\multicolumn{2}{c}{ Linear } \\
\cline{2-7}
 & Var.Red\%& Emp.Cov. 
& Var.Red\%&  Emp.Cov. 
& Var.Red\%&  Emp.Cov. 
\\
\hline 
SDA, $d=10$ & \textbf{71.39$\%$} & 0.977 &  70.73$\%$   &  0.947   & 46.71$\%$ & 0.959       \\
  \hline
  SDA, $d=100$ & \textbf{71.05$\%$}  & 0.991 & 70.74$\%$  &   0.939  & 42.43$\%$ & 0.817   \\
\hline 
non-SDA, $d=10$ & 42.23$\%$ & 0.000  & \textbf{41.30$\%$} & \textbf{0.942}  & 1.25$\%$ & 0.954     \\
  \hline
non-SDA, $d=100$ & 48.01$\%$  & 0.000 & \textbf{39.64$\%$}  &  \textbf{0.941}  & 2.63$\%$ & 0.954   \\
\bottomrule
\end{tabular}
}%
\caption{Variance reduction (relative to the 
corresponding diff-in-mean estimator) and empirical coverage for ratio metrics.
``Linear'' is $\hat\delta'_{\textrm{CUPED}}$ for SDA settings 
and the linear adjustment method in Section~\ref{subsec:linear_rr} for non-SDA settings. }
\label{tab:simu_ratio}
\end{table}

In setting 1 with SDA, 
the first two lines in Table~\ref{tab:simu_ratio}
confirms the optimality of 
Alg.~\ref{alg:ratio'}:  
its 
coverage is above the nominal level $0.95$, 
and it achieves the best variance reduction performance. 
As $\delta=\delta'$, Alg.~\ref{alg:ratio} is also valid for $\delta$ 
with comparable variance reduction to 
Alg.~\ref{alg:ratio'}; 
this is due to the nature of the simulation design, 
not necessarily true in general. 
For linear methods, CUPED (linear regression without cross-fitting) 
for the SDA, $d=100$ setting (line 2)
does not provide valid coverage, showing problems
with linear regression asymptotics under 
relatively high dimensionality.
On the contrary, our cross-fitting based linear method from Section~\ref{subsec:linear_rr} reliably 
handles high dimensionality ($d=100$) and achieves the desired coverage (line 4).

The last two lines in Table~\ref{tab:simu_ratio} 
illustrate the 
performance of different methods
in absence of SDA (setting 2) and Alg.~\ref{alg:ratio} is clearly the best.
Linear adjustment does not reduce much 
variance due to the nonlinearity of the data. 
Although Alg.~\ref{alg:ratio'} achieves more variance reduction 
due to the nature of the data-generating process, 
the inference based on Alg.~\ref{alg:ratio'} 
for $\delta'$ 
is not valid (coverage is zero).  
This happens because 
the asymptotic unbiasedness of $\hat\delta'$ relies crucially on 
$\PP_{(X_i,Z_i)\given T_i=1} = \PP_{(X_i,Z_i)\given T_i=0}$. 
However, without SDA, 
the debiasing term $\frac{1}{n_c}\sum_{T_i=0}(Y_i(0)-\hat\mu_0(X_i,Z_i(0)))$ cannot correct for the bias of $\frac{1}{n_t}\sum_{T_i=1}\hat\mu_0(X_i,Z_i(1))$. 
In fact, the confidence intervals derived from Alg.~\ref{alg:ratio'} covers another quantity 
that is neither $\delta$ nor $\delta'$ and 
whose 
practical interpretation is unclear (the quantity 
equals $\EE[\Gamma_i]/\EE[Z_i]$ for $\Gamma_i$ 
defined in Appendix~\ref{app:sda}.
Thus, the smaller variance of Alg.~\ref{alg:ratio'} 
does not make it more attractive than Alg.~\ref{alg:ratio}. 

Some takeaway messages and practical suggestions are summarized below.
\begin{enumerate}[(i)]
\item When the SDA does \emph{not} hold, the estimand is 
 $\delta = \EE[Y(1)]/\EE[Z(1)] - \EE[Y(0)]/\EE[Z(0)]$ 
and $\hat\delta$ from Alg.~\ref{alg:ratio} is optimal.
\item When the SDA holds, 
the target is $\delta' = \EE[Y(1)-Y(0)]/\EE[Z]$ and 
Alg.~\ref{alg:ratio'} is optimal. 
\vspace{-0.5em}
\item In practice, the SDA for variance reduction needs to be made with caution, because if it is violated, 
$\hat\delta'$ from Alg.~\ref{alg:ratio'} 
may be invalid and its actual target 
lacks clear interpretation. 
There should always be a separate test (such as applying Alg.~\ref{alg:debias} for count metrics) on whether $\EE[Z(1)]=\EE[Z(0)]$. 
Without strong evidence for SDA, 
we recommend dropping
the SDA and using Alg.~\ref{alg:ratio} for optimal variance reduction.  
Even if SDA \emph{actually} holds, Alg.~\ref{alg:ratio} is still valid and in some cases only slightly inferior to Alg.~\ref{alg:ratio'}.

\end{enumerate}

%% file: real.tex

In this section, we provide two real examples at LinkedIn: 
one applying Alg.~\ref{alg:debias} to analyze 
the count metrics in a LinkedIn feed experiment, 
and the other one applying Alg.~\ref{alg:ratio}  
to analyze ratio metrics in 
an enterprise experiment for LinkedIn learning.

\subsection{Count metrics in a LinkedIn feed experiment}

The LinkedIn feed is an online system exposing members to contents posted in their network, including career news, ideas, questions, and jobs in 
the form of short text, articles, images, and videos. ML algorithms are used to rank the tens of thousands of candidate updates for each member to help them discover the most relevant contents. 
We consider an experiment on LinkedIn's feed homepage where members are randomly assigned into a treatment group and a control group. The treatment group is assigned a new version of feed relevance algorithm, 
which would be compared to the baseline algorithm in the control group. 
The goal of the experiment is to understand how the new ranking algorithm would impact the revenue from feed. As discussed in \cite{deng2013improving}, the effectiveness of variance reduction using CUPED depends on the linear correlation of the experiment outcome with the pre-experiment metric. 
It is challenging to reduce the variance for revenue 
because member's revenue is a volatile metric whose autocorrelation
across different time periods is weak. 
We hope to improve the performance of variance reduction by incorporating 
more covariates and 
exploiting nonlinearity with ML methods. 

The experiment takes a random sample of $n=400000$ members from the LinkedIn online feed traffic 
and remove outliers whose revenues are  above the $99.5\%$ quantile. We implement Alg.~\ref{alg:debias} with gradient boosting in the \texttt{scikit-learn} Python library and compare it to CUPED (with separate linear regressions on treated and control outcomes) as well as the diff-in-mean estimator. 
To illustrate the advantage of incorporating side information, 
the CUPED method only uses pre-treatment revenue metric, whilst 
two sets of covariates are used in Alg~\ref{alg:debias}: 
1) pre-treatment revenue metric; 
2) pre-treatment revenue metric and other member attributes including country code, industry, membership status, job seeker class, profile viewer count, connection count, network density, etc., where categorical features are transformed into binary variables with one-hot encoding.

With only pre-treatment revenue metrics, CUPED reduces $15.91\%$ of variance compared to the diff-in-mean estimator, whereas Alg.~\ref{alg:debias} reduces $19.77\%$ by exploiting nonlinear relations with gradient boosting models. 
By further incorporating member attribute covariates, Alg.~\ref{alg:debias} reduces $22.22\%$ of variance compared to the diff-in-mean estimator, 
showing the efficiency gained from more side information.

\subsection{Ratio metrics in a LinkedIn learning experiment}

LinkedIn Learning platform is launching a new notification center from which learners can receive customized course recommendations based on what they watched, saved or the trending courses. It also ensures that the learners would not miss assignments from their managers, active conversations with instructors and learners on their favorite contents.
 
 \begin{figure}[htb] 
   \centering
   \includegraphics[width=400pt]{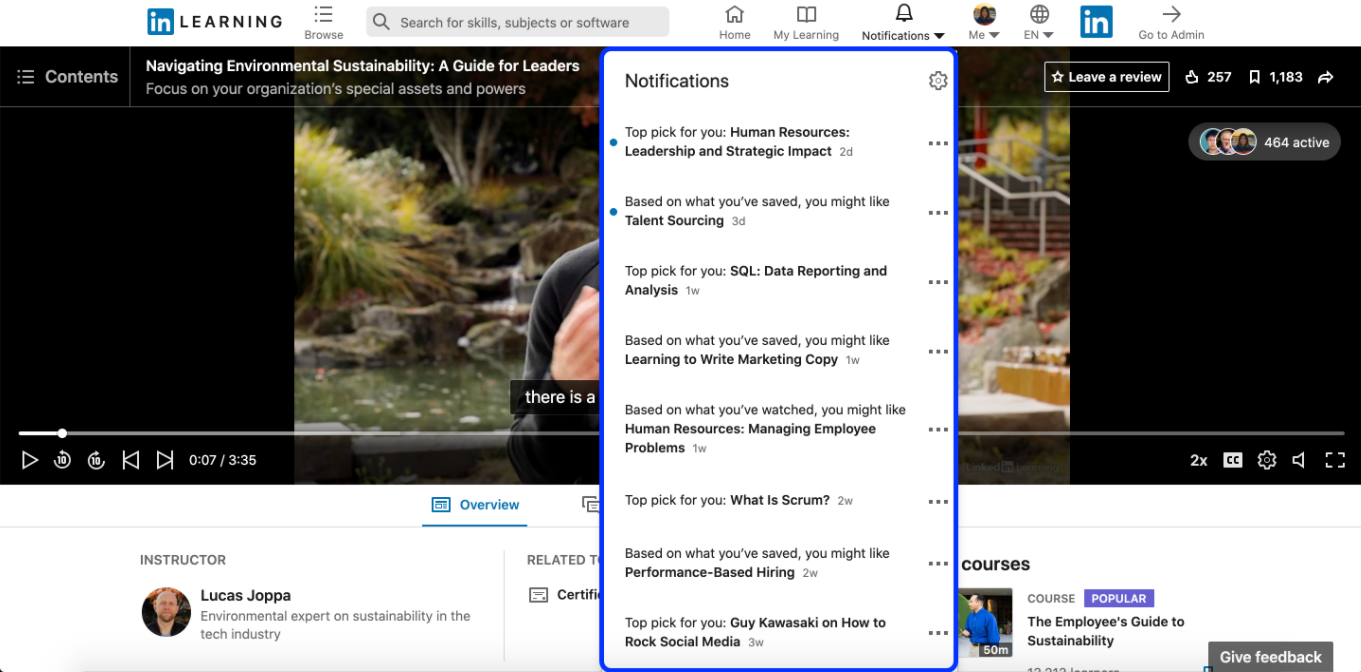}
   \caption{The new LinkedIn Learning notification center.}\label{Fig:notification}
 \end{figure}
 
An online experiment is conducted to assess the impact of this new notification center on the enterprise learners. Because the notification center is an explicit new feature with UI changes (as shown in Figure~\ref{Fig:notification}), 
the experiment requires that learners from the same enterprise account will either be all in control or all in treatment to avoid
jeopardizing customer trust.
In order to ensure ``same account same experience'', the experiment is randomized by enterprise accounts (instead of by individual learners). 
We focus on the ``learning engagement per learner'' metric, which measures 
the contribution of one day's engagement  
to future engagement (video time watched). As introduced in Section~\ref{subsec:metric}, 
since the analysis unit (learner) is at a lower level than the randomization unit (enterprise account), our metric of interest is a ratio metric ``learning engagement per account/number of active learners per account''. This experiment can also be viewed as a cluster randomized experiment.

For target $\delta$, we compare Alg.~\ref{alg:ratio} to an extension (separate regressions on each group)
of
the CUPED outlined in the appendix of~\cite{deng2013improving}.  
When the target is $\delta'$, CUPED 
means $\hat\delta'_{\textrm{CUPED}} = \hat\theta_{\textrm{CUPED}}/\big(\frac{1}{n}\sum_{i=1}^nZ_i\big)$, 
where $\hat\theta_{\textrm{CUPED}}$ is the CUPED estimator for count metrics ``learning engagement per account''. 
The corresponding diff-in-mean estimators are the same as introduced 
in Section~\ref{subsec:simu_ratio}. 
To illustrate the advantage of 
incorporating extra covariates, the CUPED
method only uses pre-treatment versions of $Y,Z$, whereas
for each target, 
two sets of covariates are used in 
Algs~\ref{alg:ratio} and~\ref{alg:ratio'}: 
1) only pre-treatment versions of $Y,Z$; 
2) pre-treatment versions of $Y,Z$ and features of contracts, 
including industry segment, 
SSO information, 
country code, account size category, etc.
In both settings, we implement our procedures
with the \texttt{XGBoost} python library without special tuning.
The results 
are summarized in Table~\ref{tab:real}.

\begin{table}[htbp]
\centering
{\footnotesize
\renewcommand\arraystretch{1.0}
\begin{tabular}{c|c |c | c |c   }
\toprule
\multirow{2}{*}{Covariates} &    \multicolumn{2}{c|}{ Target: $\delta$ } &\multicolumn{2}{c }{ Target: $\delta'$ }\\
\cline{2-5}
 & Alg.~\ref{alg:ratio} & CUPED  
& Alg.~\ref{alg:ratio'} & CUPED  
\\
\hline 
1) & 3.13$\%$ (1.40$\%$) & \multirow{2}{*}{1.76$\%$} 
&  74.1$\%$ (-11.13$\%$)  &  \multirow{2}{*}{76.62$\%$}     \\
  \cline{1-2}\cline{4-4}
2) & 12.35$\%$ (10.78$\%$) &  &   83.6$\%$ (29.58$\%$)  &        \\
\bottomrule
\end{tabular}
}%
\caption{ 
Proportion of reduced variance compared to DiM 
(compared to CUPED). }
\label{tab:real}
\end{table}

For both targets, we see the efficacy
of using machine learning tools and incorporating 
large numbers of covariates.  
For target $\delta$ when CUPED does not reduce much variance, 
flexible machine learning tools improves CUPED by
going beyond linearity as in Line 1
and incorporating a large number of extra covariates as in Line 2. 
For $\delta'$, we see the substantial improvement
(about $70\%$ to $80\%$ of variance reduced)
compared to the diff-in-mean estimator
and our optimal estimator can further achieve 30\% lower variance than CUPED by exploiting the information in the many covariates.

%% file: appendix.tex
\section{Deferred details for ratio metrics under SDA}
\label{app:sda}

\subsection{Asymptotic inference }

To establish asymptotically valid inference procedure, we impose several conditions on the data generating process and convergence of the estimators. 

\begin{assumption}[Data Generating Process]
  \label{assump:r_data}
  Suppose $(X_i,Z_i,Y_i(0),Y_i(1))\iid \PP$ and $T_i\iid \textrm{Bernoulli}(p)$ which are independent of all other random variables. 
\end{assumption}

\begin{assumption}[Convergence]
\label{assump:r_conv}
There exists some fixed fuctions $\mu_{1 }^*(\cdot,\cdot)$, $\mu_{0 }^*(\cdot,\cdot)$ so that $\|\hat\mu_w^{(k)} - \mu_{w}^*\|_2\pto 0$ for all $k\in [K]$ and $w\in\{0,1\}$, where $\|\cdot\|_2$ denotes the $L_2$ norm on the probability space of $(X_i,Z_i)$. 
\end{assumption}

Under the above two conditions, we are to show that the estimator $\hat\delta'$ has influence function
\#\label{eq:r_infl}
\phi_{\delta}'(Y_i,Z_i,X_i,T_i) = \frac{\Gamma_i^*}{\EE[Z_i]}  - \frac{\EE[Y_i(1)]-\EE[Y_i(0)]}{\EE[Z_i]^2} \big(Z_i - \EE[Z_i]\big),\quad \text{where}
\#
\$
\Gamma_i^* &=  \tilde\mu_1^*(X_i,Z_i) - \tilde\mu_0^*(X_i,Z_i)  +  \frac{T_i}{  p}
\big(\tilde Y_i(1) - \tilde\mu_1^*(X_i,Z_i)\big) - \frac{1-T_i}{1- p} \big(\tilde Y_i(0) - \tilde\mu_0^*(X_i,Z_i)\big), 
\$
and $\tilde{Y}_i(w) = Y_i(w) - \EE[Y_i(w)]$, $\tilde\mu_w^*(X_i,Z_i) = \mu_w^*(X_i,Z_i) - \EE[\mu_w^*(X_i,Z_i)]$ for $w\in \{0,1\}$ are centered random variables. 
The following theorem establishes the asymptotic behavior of the proposed estimator, whose proof is deferred to 
Appendix~\ref{app:r_inference}.
\begin{theorem}\label{thm:r_inference}
Suppose Assumptions~\ref{assump:r_data} and~\ref{assump:r_conv} hold, and let $\hat\delta'$ be the output of Algorithm~\ref{alg:ratio'}. Then $\sqrt{n}(\hat\delta' - \delta')\stackrel{d}{\to} N(0,\sigma_{\delta'}^2)$, where $\sigma_{\delta'}^2 = \Var\big(\phi_{\delta}'(Y_i,Z_i,X_i,T_i)\big)$ 
for the influence function~\eqref{eq:r_infl}. 
Furthermore, define the variance estimator 
\#\label{eq:r_var_est}
&\hat\delta_{\delta'}^2  =  
\frac{1}{n}\sum_{T_i=1} \big(g_{1,i} - \bar{g}_1(1)\big)^2 + \frac{1}{n}\sum_{T_i=0} \big( g_{0,i} - \bar{g}_0(0)\big)^2 \\
\textrm{where}~~ g_{1,i} &= \frac{1}{\bar{Z}} \big(   \hat \mu_1(X_i,Z_i) - \hat  \mu_0 (X_i,Z_i)\big) + \frac{n}{n_t \bar{Z}}  \big( {Y}_i(1) -  \hat \mu_1 (X_i,Z_i)\big)  - \frac{\bar{Y}(1)-\bar{Y}(0)}{\bar{Z}^2} Z_i, \notag \\
g_{0,i} &= \frac{1}{\bar{Z}} \big(   \hat \mu_1(X_i,Z_i) - \hat  \mu_0 (X_i,Z_i)\big) - \frac{n}{n_c \bar{Z}}  \big( {Y}_i(0) -  \hat \mu_0 (X_i,Z_i)\big)  - \frac{\bar{Y}(1)-\bar{Y}(0)}{\bar{Z}^2} Z_i. \notag
\#
Then $\hat\delta' \pm \hat\sigma_{\delta'} \cdot z_{1-\alpha/2}/\sqrt{n}$ is an asymptotically valid $(1-\alpha)$ confidence interval for $\delta'$.
\end{theorem}

\subsection{Optimality}

In this part, we show the semiparametric efficiency of 
the proposed estimator under similar consistency conditions. 
To begin with, we impose the following two assumptions on the data generating process and consistency. 


\begin{assumption}\label{assump:r_consistency}
Denote $\mu_w(x,z) = \EE[Y(w)\given X=x,Z=z]$ for all $w\in \{0,1\}$. Suppose $\|\hat\mu_w^{(k)} - \mu_{w}\|_2\pto 0$ for all $k\in[K]$ and $w\in\{0,1\}$, where $\|\cdot\|_2$ denotes the $L_2$-norm in the probability space of $(X_i,Z_i)$.
\end{assumption}

To begin with, we define the efficient influence function for estimating $\delta'$ as 
\#
&\phi_{\delta',\dag}(Y_i,Z_i,X_i,T_i) = \frac{\Gamma_i^{\dag}}{\EE[Z_i]} - \frac{\EE[Y_i(1)]-\EE[Y_i(0)]}{\EE[Z_i]^2}\big(Z_i-\EE[Z_i]), \label{eq:r_eff_infl}\\
\textrm{where}\quad &\Gamma_i^{\dag} =   \mu_{1} (X_i,Z_i ) - \mu_0(X_i,Z_i) - \big( \EE[Y_i(1)] - \EE[Y_i(0)] \big)\notag  \\
&\qquad \qquad + \frac{T_i}{p}  \big(Y_i  - \mu_{1 } (X_i,Z_i ) \big) - \frac{1-T_i}{1-p}\big(Y_i  - \mu_{0 } (X_i,Z_i ) \big) . \notag 
\#

The following theorem establishes the semiparametric efficiency of $\hat\delta'$ under the consistency assumption~\ref{assump:r_consistency} by referring to a more general result (c.f.~Theorem~\ref{thm:r_optimality})
which provides efficiency bounds under $T_i\indep (Y_i(1),Y_i(0))\given X_i,Z_i$. The latter works for more general settings than the randomized experiment protocol, which might be of independent interest for observational studies. 

\begin{theorem}\label{thm:r_optimality}
Suppose Assumptions~\ref{assump:r_data} and~\ref{assump:r_consistency} hold. Then it holds that $\sqrt{n}(\hat\delta' - \delta')\stackrel{d}{\to}N(0,\sigma_{\delta',\dag}^2)$, where 
$
\sigma_{\delta',\dag}^2 = \Var\big(\phi_{\delta',\dag}(Y_i,Z_i,X_i,T_i)\big)
$
for $\phi_{\delta',\dag}$ defined in~\eqref{eq:r_eff_infl} 
is the semiparametric asymptotic variance bound for $\delta' = \EE[Y(1)-  Y(0)]/\EE[Z]$. 
\end{theorem}

Theorem~\ref{thm:r_optimality} sheds light on the ideal procedure of variance reduction
for ratio metrics with stable denominators: one should aim to fit the two conditional mean functions of the numerator $Y$ on both $(X,Z)$ in the fit-and-plugin framework, where the denominator $Z_i$ should be pooled to estimate the expectation. 

The procedure proposed here is different from the optimal one for a changing $Z$ in Section~\ref{sec:rr} due to differeces in the estimands and 
probabilistic nature of the data-generating processes. 
Although it seems plausible that pooling all data to estimate the denominator (and making more assumptions) should lead
to an estimator with smaller variance, it is not always the case
and might even hurt the validity of the procedure if the assumption is false imposed. 
We will elaborate on this issue by simulation studies in Section~\ref{subsec:simu_ratio}.

\section{Proof of results on count metrics}

\subsection{Proof of unbiasedness} \label{app:unbiased}
\begin{proof}[Proof of Theorem~\ref{thm:deb_unbiased}]
Recall that  $\cD^{(k)}$, $k=1,\dots,K$ are the $K$ randomly-split folds, with $n_k = |\cD^{(k)}|$, $n_{k,t} = \sum_{i\in \cD^{(k)}}T_i$ and $n_{k,c} = n_k - n_{k,t}$. Here $\lfloor n/K\rfloor \leq n_k \leq \lfloor n/K\rfloor+1$. For each $k\in [K]$, since the functions $\hat\mu_1^{(k)}$, $\hat\mu_0^{(k)}$ are fitted with data in $\cD\backslash \cD^{(k)}$, we have 
\$
&\EE\bigg[\frac{1}{n_k}\sum_{i\in \cD^{(k)}}   \big(\hat\mu_1(X_i) - \hat\mu_0(X_i)\big) + \frac{1}{n_{k,t}}\sum_{i\in \cD^{(k)}, T_i=1} \big( Y_i(1) - \hat\mu_1(X_i) \big) - \frac{1}{n_{k,c}} \sum_{i\in \cD^{(k)}, T_i=0} \big( Y_i(0) - \hat\mu_0(X_i)  \big)\bigggiven \cT\bigg]\\
&= \EE\bigg[ \EE\Big[  \big(\hat\mu_1(X_i) - \hat\mu_0(X_i)\big) +  \big( Y_i(1) - \hat\mu_1(X_i) \big) -  \big( Y_i(0) - \hat\mu_0(X_i)  \big)\Biggiven \cT, \cD^{(-k)}\Big]\bigggiven \cT\bigg]\\
&= \EE\big[Y_i(1)-Y_i(0)\biggiven \cT\big] = \EE\big[Y_i(1)-Y_i(0) \big].
\$
Here the first equality follows from the tower property of conditional expectations, as well as the i.i.d.~of $(Y_i(1),Y_i(0),\hat\mu_1(X_i),\hat\mu_0(X_i))$ conditional on $\cD^{(-k)},\cT$ thanks to the cross-fitting technique and the independence of $\cT$. The second equality follows from the independence of $\hat\mu_1,\hat\mu_0$ and $X_i$ for $i\in \cD^{(k)}$ due to cross-fitting, and the last equality follows from the independence of $\cT$ and $(Y_i(1),Y_i(0))$. Since $\hat\theta_{\deb}$ in \eqref{eq:db_count} is the average of all $\hat\theta_{\deb}^{(k)}$, we obtain $\EE[\hat\theta_{\deb}\given \cT] = \delta$ for any sample size $n$. 

We then consider $\tilde\theta_{\deb}$. Note that 
\$
\tilde\theta_{\deb} =& \sum_{k=1}^K \frac{n_k}{n} \Big( \frac{1}{n_k}\sum_{i\in \cD^{(k)}}   \big(\hat\mu_1(X_i) - \hat\mu_0(X_i)\big) \Big) +  \sum_{k=1}^K \frac{n_{k,t}}{n_t}\Big(\frac{1}{n_{k,t}}\sum_{i\in \cD^{(k)}, T_i=1} \big( Y_i(1) - \hat\mu_1(X_i) \big)\Big) \\
&\qquad \qquad  - \sum_{k=1}^K \frac{n_{k,c}}{n_c}\Big(\frac{1}{n_{k,c}} \sum_{i\in \cD^{(k)}, T_i=0} \big( Y_i(0) - \hat\mu_0(X_i)  \big)\Big)\\
=& \sum_{k=1}^K \frac{n_k}{n} ~\hat\theta_{\deb}^{(k)} +  \sum_{k=1}^K \Big(  \frac{n_{k,t}}{n_t} - \frac{n_k}{n}\Big) \Big(\frac{1}{n_{k,t}}\sum_{i\in \cD^{(k)}, T_i=1} \big( Y_i(1) - \hat\mu_1(X_i) \big)\Big) \\
&\qquad \qquad  - \sum_{k=1}^K \Big(\frac{n_{k,c}}{n_c} - \frac{n_{k }}{n }\Big) \Big(\frac{1}{n_{k,c}} \sum_{i\in \cD^{(k)}, T_i=0} \big( Y_i(0) - \hat\mu_0(X_i)  \big)\Big).
\$
By the conditions in (i), we know 
\$
\Big| \frac{n_{k,t}}{n_t} - \frac{n_k}{n} \Big| \leq \frac{2}{n_t},\quad \Big| \frac{n_{k,c}}{n_c} - \frac{n_k}{n} \Big| \leq \frac{2}{n_c}.
\$
Therefore
\$
\Big|\EE\big[\tilde\theta_{\deb} \biggiven \cT\big] - \delta \Big| \leq & ~\Bigg|\EE\bigg[ \sum_{k=1}^K \Big(  \frac{n_{k,t}}{n_t} - \frac{n_k}{n}\Big) \Big(\frac{1}{n_{k,t}}\sum_{i\in \cD^{(k)}, T_i=1} \big( Y_i(1) - \hat\mu_1(X_i) \big)\Big) \bigggiven \cT\bigg]\\
&\qquad \qquad  - \EE\bigg[ \sum_{k=1}^K \Big(\frac{n_{k,c}}{n_c} - \frac{n_{k,t}}{n_t}\Big) \Big(\frac{1}{n_{k,c}} \sum_{i\in \cD^{(k)}, T_i=0} \big( Y_i(0) - \hat\mu_0(X_i)  \big)\Big)\bigggiven \cT\bigg]\Bigg| \\
= & ~\Bigg|\EE\bigg[ \sum_{k=1}^K \Big(  \frac{n_{k,t}}{n_t} - \frac{n_k}{n}\Big) \EE\big[  Y_i(1) - \hat\mu_1(X_i) \biggiven \cD^{(-k)},\cT\big]  \bigggiven \cT\bigg]\\
&\qquad \qquad  - \EE\bigg[ \sum_{k=1}^K \Big(\frac{n_{k,c}}{n_c} - \frac{n_{k }}{n }\Big) \EE\big[  Y_i(0) - \hat\mu_0(X_i) \biggiven \cD^{(-k)},\cT\big] \bigggiven \cT\bigg]\Bigg| \\
\leq& K \cdot  \frac{2}{\min\{n_t,n_c\}} \cdot 2c_0 := \frac{c}{\min\{n_t,n_c\}} 
\$
for absolute constant $c:= 4Kc_0$. Here the first equality follows from the tower property and the i.i.d.~of $(Y_i(1),Y_i(0),\hat\mu_1(X_i),\hat\mu_0(X_i))$ conditional on $\cD^{(-k)},\cT$ thanks to the cross-fitting technique and the independence of $\cT$. The second inequality follows from the conditions of the bounded expectations. Thus we complete the proof.
\end{proof}

\subsection{Proof of asymptotic inference} \label{app:asymp_var}
\begin{proof}[Proof of Theorem~\ref{thm:asymp_var}]
We first investigate the behavior of $\hat\theta_{\deb}^{(k)}$ for each $k\in[K]$. Define its counterpart
\$
\hat\theta_{\deb}^{(*,k)} = 
\frac{1}{n_k}\sum_{i\in \cD^{(k)}}   \big( \mu_1^*(X_i) - \mu_0^*(X_i)\big) + \frac{1}{n_{k,t}}\sum_{\substack{T_i=1,\\i\in \cD^{(k)}}} \big( Y_i(1) -  \mu_1^*(X_i) \big) - \frac{1}{n_{k,c}} \sum_{\substack{T_i=0,\\i\in \cD^{(k)}}} \big( Y_i(0) -  \mu_0^*(X_i)  \big).
\$
Denoting $\Delta_1(X_i) = \hat\mu_1(X_i) - \mu_1^*(X_i)$ and $\Delta_0(X_i) = \hat\mu_0(X_i) - \mu_0^*(X_i)$, we have 
\#\label{eq:diff1}
\hat\theta_{\deb}^{(k)} - \hat\theta_{\deb}^{(*,k)} &= 
\frac{1}{n_k}\sum_{i\in \cD^{(k)}}   \big( \Delta_1(X_i) - \Delta_0(X_i) \big) - \frac{1}{n_{k,t}}\sum_{\substack{T_i=1,\\i\in \cD^{(k)}}}  \Delta_1(X_i)   + \frac{1}{n_{k,c}} \sum_{\substack{T_i=0,\\i\in \cD^{(k)}}} \Delta_0(X_i)
\#
We denote $\EE_c$ as the expectation taken conditional on the sample splitting, $\cD^{(-k)}$ and $\cT$, under which $\Delta_0(X_i)$ for $i\in \cD^{(k)}$ are i.i.d., separately. Then 
\$
&\EE_c\bigg[ \frac{1}{n_k}\sum_{i\in \cD^{(k)}}   \Big( \Delta_1(X_i) - \Delta_0(X_i) - \frac{n_k}{n_{k,t}} T_i \Delta_1(X_i) + \frac{n_k}{n_{k,c}}(1-T_i)\Delta_0(X_i) \Big) \bigg] \\
&= \EE_c\big[ \Delta_1(X_i) - \Delta_0(X_i)   \big] -  \frac{1}{n_k} \frac{n_k}{n_{k,t}} \sum_{i\in \cD^{(k)}} T_i  \EE_c\big[ \Delta_1(X_i)\big] ~+ \frac{1}{n_k} \frac{n_k}{n_{k,c}} \sum_{i\in \cD^{(k)}} (1-T_i)  \EE_c\big[ \Delta_0(X_i)\big] = 0.
\$
On the other hand, by~\eqref{eq:diff1} and spreading out all cross-terms, we have
\$
&n_k \cdot \EE_c\Big[  \big(\hat\theta_{\deb}^{(k)} - \hat\theta_{\deb}^{(*,k)} \big)^2 \Big] \\
&= \frac{1}{n_k} \EE_c\Bigg[ \bigg(  \sum_{\substack{T_i=1,i\in \cD^{(k)}}} \Big(-\frac{n_{k,c}}{n_{k,t}} \Delta_1(X_i) - \Delta_0(X_i)\Big)  + \frac{1}{n_{k,c}} \sum_{\substack{T_i=0,i\in \cD^{(k)}}}\Big( \Delta_1(X_i) + \frac{n_{k,t}}{n_{k,c}} \Delta_0(X_i)   \Big)  \bigg)^2 \Bigg] \\
&= \frac{n_{k,t}}{n_k}\EE_c\bigg[ \Big(-\frac{n_{k,c}}{n_{k,t}} \Delta_1(X_i) - \Delta_0(X_i)\Big)^2 \bigg] + \frac{n_{k,c}}{n_k} \EE_c\bigg[ \Big( \Delta_1(X_i) + \frac{n_{k,t}}{n_{k,c}} \Delta_0(X_i)   \Big) ^2\bigg]\\
&\qquad + \frac{n_{k,t}(n_{k,t}-1)}{n_{k}} \bigg( \EE_c\Big[-\frac{n_{k,c}}{n_{k,t}} \Delta_1(X_i) - \Delta_0(X_i)\Big] \bigg)^2 + \frac{n_{k,c}(n_{k,c}-1)}{n_k}\bigg( \EE_c\Big[\Delta_1(X_i) + \frac{n_{k,t}}{n_{k,c}} \Delta_0(X_i) \Big]\bigg)^2 \\
&\qquad \qquad + \frac{2n_{k,t}n_{k,c}}{n_k} \EE_c\Big[-\frac{n_{k,c}}{n_{k,t}} \Delta_1(X_i) - \Delta_0(X_i)\Big]\cdot  \EE_c\Big[\Delta_1(X_i) + \frac{n_{k,t}}{n_{k,c}} \Delta_0(X_i) \Big].
\$
Here by the Cauchy-Schwarz inequality, we have 
\$
&\frac{n_{k,t}}{n_k}\EE_c\bigg[ \Big(-\frac{n_{k,c}}{n_{k,t}} \Delta_1(X_i) - \Delta_0(X_i)\Big)^2 \bigg] \leq \frac{n_{k,t}}{n_k} \frac{2n_{k,c}^2}{n_{k,t}^2} \|\Delta\|_2^2 + \frac{2n_{k,t}}{n_{k,t}} \|\Delta_0\|_2^2 = o_P(1),\\
&\frac{n_{k,c}}{n_k} \EE_c\bigg[ \Big( \Delta_1(X_i) + \frac{n_{k,t}}{n_{k,c}} \Delta_0(X_i)   \Big) ^2\bigg] \leq \frac{2n_{k,c}}{n_k} \|\Delta_1\|_2^2 + \frac{n_{k,c}}{n_k} \frac{n_{k,t}^2}{n_{k,c}^2} \|\Delta_0\|_2^2 = o_P(1). 
\$
Meanwhile, letting $e_1 = \EE_c[\Delta_1(X_i)]$ and $e_0 = \EE_c[\Delta_0(X_i)]$, the remaining terms can be written as 
\$
& \frac{n_{k,t}(n_{k,t}-1)}{n_{k}} \bigg( \EE_c\Big[-\frac{n_{k,c}}{n_{k,t}} \Delta_1(X_i) - \Delta_0(X_i)\Big] \bigg)^2 + \frac{n_{k,c}(n_{k,c}-1)}{n_k}\bigg( \EE_c\Big[\Delta_1(X_i) + \frac{ n_{k,t}}{n_{k,c}} \Delta_0(X_i) \Big]\bigg)^2 \\
&  \qquad + \frac{2n_{k,t}n_{k,c}}{n_k} \EE_c\Big[-\frac{n_{k,c}}{n_{k,t}} \Delta_1(X_i) - \Delta_0(X_i)\Big]\cdot  \EE_c\Big[\Delta_1(X_i) + \frac{n_{k,t}}{n_{k,c}} \Delta_0(X_i) \Big] \\
&= \frac{n_{k,t}^2 - n_{k,t}}{n_{k}} \bigg( \frac{n_{k,c}^2}{n_{k,t}^2} e_1^2 - \frac{2n_{k,c}}{n_{k,t}} e_1e_0 + e_0^2 \Big)  + \frac{n_{k,c}^2 - n_{k,c}}{n_k} \bigg( e_1^2 + \frac{2n_{k,t}}{n_{k,c}} e_1e_0 + \frac{n_{k,t}^2}{n_{k,c}^2} e_0^2  \bigg) \\
&\qquad + \frac{2n_{k,t}n_{k,c}}{n_k}\bigg( -\frac{n_{k,c}}{n_{k,t}}  e_1^2 - 2e_1e_0 - \frac{n_{k,t}}{n_{k,c}} e_0^2 \bigg) \\
&= - \frac{ n_{k,t}}{n_k}\bigg( \frac{n_{k,c}^2}{n_{k,t}^2} e_1^2 - \frac{2n_{k,c}}{n_{k,t}} e_1e_0 + e_0^2 \Big) - \frac{n_{k,c}}{n_k}\bigg( e_1^2 + \frac{2n_{k,t}}{n_{k,c}} e_1e_0 + \frac{n_{k,t}^2}{n_{k,c}^2} e_0^2  \bigg) = o_P(1)
\$
since $e_1,e_0=o_P(1)$. Therefore, by Markov's inequality, for any $\epsilon>0$, it holds that 
\$
\PP_c\Big( \sqrt{n_k} \big|\hat\theta_{\deb}^{(k)} - \hat\theta_{\deb}^{(*,k)} \big|> \epsilon   \Big) \stackrel{P}{\to} 0.
\$
Applying Lemma~\ref{lem:cond_to_op} with $\cF_n$ being the $\sigma$-algebra of $\PP_c$, we know $\sqrt{n} \big(\hat\theta_{\deb}^{(k)} - \hat\theta_{\deb}^{(*,k)} \big) = o_P(1)$, both marginally and conditionally on $\cT$. 

We now consider the behavior of $\tilde\theta_{\deb}$ by comparing to its counterpart
\$
\tilde\theta_{\deb}^* = \frac{1}{n }\sum_{i=1}^n   \big( \mu^*_1(X_i) -  \mu^*_0(X_i)\big) + \frac{1}{n_t}\sum_{i~\textrm{treated}} \big( Y_i(1) - \mu^*_1(X_i) \big) - \frac{1}{n_{c}} \sum_{i~\textrm{control}} \big( Y_i(0) -  \mu^*_0(X_i)  \big).
\$
Note that 
\$
\tilde\theta_{\deb} - \tilde\theta_{\deb}^* &= \frac{1}{n}\sum_{i~\textrm{treated}} \Big( -\frac{n_c}{n_t} \Delta_1(X_i) - \Delta_0(X_i) \Big) + \frac{1}{n}\sum_{i~\textrm{control}} \Big( \Delta_1(X_i) + \frac{n_t}{n_c} \Delta_0(X_i)\Big)\\
&=\frac{1}{n}\sum_{k=1}^K \sum_{i\in \cD^{(k)}} T_i\Big( -\frac{n_c}{n_t} \Delta_1(X_i) - \Delta_0(X_i) \Big) + (1-T_i)\Big( \Delta_1(X_i) + \frac{n_t}{n_c} \Delta_0(X_i)\Big)\\
&:= \sum_{k=1}^K  \big(\tilde\theta_{\deb}^{(k)} - \tilde\theta_{\deb}^{(*,k)}\big).
\$
Using the aforementioned notations, for each $k\in [K]$, 
\$
&n\cdot\EE_c\Big[ \big(\tilde\theta_{\deb}^{(k)} - \tilde\theta_{\deb}^{(*,k)}\big)^2  \Big] \\
&= \frac{n_{k,t}}{n}\EE_c\bigg[ \Big(-\frac{n_c}{n_t} \Delta_1(X_i) - \Delta_0(X_i)\Big)^2 \bigg] + \frac{n_{k,c}}{n}\EE_c\bigg[ \Big( \Delta_1(X_i) + \frac{n_t}{n_c} \Delta_0(X_i)\Big)^2 \bigg] \\
&\qquad + \frac{n_{k,t}(n_{k,t}-1)}{n} \bigg( \EE_c\Big[-\frac{n_{c}}{n_{t}} \Delta_1(X_i) - \Delta_0(X_i)\Big] \bigg)^2 + \frac{n_{k,c}(n_{k,c}-1)}{n}\bigg( \EE_c\Big[\Delta_1(X_i) + \frac{n_{t}}{n_{c}} \Delta_0(X_i) \Big]\bigg)^2 \\
&\qquad \qquad + \frac{2n_{k,t}n_{k,c}}{n} \EE_c\Big[-\frac{n_{c}}{n_{t}} \Delta_1(X_i) - \Delta_0(X_i)\Big]\cdot  \EE_c\Big[\Delta_1(X_i) + \frac{n_{t}}{n_{c}} \Delta_0(X_i) \Big] = o_P(1)
\$
by similar arguments as before. Therefore, we have $\sqrt{n} \big(\tilde\theta_{\deb} - \tilde\theta_{\deb}^{*} \big) = o_P(1)$, both marginally and conditionally on $\cT$. 

We now show the asymptotic equivalence of $\tilde\theta_\deb^*$ and $\hat\theta_\deb^*$. Note that 
\$
\sqrt{n}\bigg(\frac{\hat\theta_{\deb}^{(*,k)}}{K} - \tilde\theta_{\deb}^{(*,k)}\bigg) &= \sqrt{n}\bigg(\frac{1}{Kn_k} - \frac{1}{n}\bigg)\sum_{i\in \cD^{(k)}}\big( \mu^*_1(X_i) -  \mu^*_0(X_i)\big) \\
&\qquad + \sqrt{n}\bigg(\frac{1}{Kn_{k,t}} - \frac{1}{n_t}\bigg) \sum_{i\in \cD^{(k)}} T_i\big(Y_i(1) - \mu^*_1(X_i) \big)\\
&\qquad \qquad -   \sqrt{n}\bigg(\frac{1}{Kn_{k,c}} - \frac{1}{n_c}\bigg) \sum_{i\in \cD^{(k)}} T_i\big(Y_i(0) - \mu^*_0(X_i) \big).
\$
Here 
\$
\sqrt{n}\bigg(\frac{1}{Kn_k} - \frac{1}{n}\bigg) =\frac{n -Kn_k}{\sqrt{n}\cdot Kn_k} = O\big(n^{-3/2}\big)
\$
since $|n-Kn_k|\leq 1$. Therefore we have (under $\PP_c$)
\$
\sqrt{n}\bigg(\frac{1}{Kn_k} - \frac{1}{n}\bigg)\sum_{i\in \cD^{(k)}}\big( \mu^*_1(X_i) -  \mu^*_0(X_i)\big) = o_P(1).
\$
With similar arguments for the other two terms, we have $\sqrt{n}\big(\frac{\hat\theta_{\deb}^{(*,k)}}{K} - \tilde\theta_{\deb}^{(*,k)}\big) = o_P(1)$ both marginally and conditionally on $\cT$. Adding up the $K$ terms, we have $\sqrt{n}\big(\tilde\theta_\deb^* - \hat\theta_\deb^*\big) = o_P(1)$. To summarize, so far we have shown that 
\$
\sqrt{n}\big(\tilde\theta_\deb  - \hat\theta_\deb \big),~ \sqrt{n}\big(\tilde\theta_\deb^* - \hat\theta_\deb^*\big),~ \sqrt{n}\big(\tilde\theta_\deb^* -  \hat\theta_\deb\big) = o_P(1)
\$
both marginally and conditionally on $\cT$. Therefore, to study the asymptotic behavior of $\hat\theta_\deb$ and $\tilde\theta_\deb$, it suffices to focus on the oracle $\tilde\theta_\deb^*$. Conditionally on $\cT$, i.e., fixing the treatment and control indices, we know 
\$
\tilde\theta_\deb^* = \frac{1}{n_t} \sum_{i~\textrm{treated}} \Big( Y_i(1) - \frac{n_c}{n} \mu_1^*(X_i) - \frac{n_t}{n} \mu_0^*(X_i) \Big) + \frac{1}{n_c} \sum_{i~\textrm{control}} \Big(  - Y_i(0) + \frac{n_c}{n} \mu_1^*(X_i) + \frac{n_t}{n} \mu_0^*(X_i) \Big),
\$
where the two summations are mutually independent. Furthermore, by law of large numbers and $n_t/n\stackrel{P}{\to}p$, 
\$
\frac{1}{n_t} \sum_{i~\textrm{treated}} \Big(   \frac{n_c}{n} - (1-p)\Big) \Big(\mu_1^*(X_i) - \EE\big[\mu_1^*(X_i)\big]   \Big) = o_P(1/\sqrt{n}).
\$
With similar arguments for other terms, we have 
\$
\sqrt{n}\big(\tilde\theta_\deb^* - \delta\big) & = o_P(1) + \frac{1}{n_t} \sum_{i~\textrm{treated}} \Big( Y_i(1) - (1-p) \mu_1^*(X_i) - p \mu_0^*(X_i) - \EE\big[Y_i(1)\big]  \\
&\qquad \qquad \qquad + (1-p) \EE\big[\mu_1^*(X_i)\big] - p\EE\big[ \mu_1^*(X_i)\big] \Big) \\
&\qquad + \frac{1}{n_c} \sum_{i~\textrm{control}} \Big(  - Y_i(0) + (1-p) \mu_1^*(X_i) + p \mu_0^*(X_i) + \EE\big[Y_i(0)\big]  \\
&\qquad \qquad \qquad + (1-p) \EE\big[\mu_1^*(X_i)\big] - p \EE\big[ \mu_1^*(X_i)\big] \Big).
\$
Since $n_t/n\stackrel{P}{\to} p$, by the central limit theorem, we know 
\#
&\sqrt{n}\big(\tilde\theta_\deb^* - \delta\big) \stackrel{d}{\to} N(0,\sigma_\deb^2),  \label{eq:sig_deb}\\
\textrm{where}~~&\sigma_{\deb}^2 = \frac{1}{p}\Var\Big(Y_i(1) - (1-p) \mu_1^*(X_i) - p \mu_0^*(X_i)\Big) + \frac{1}{1-p}\Var\Big( Y_i(0) - (1-p) \mu_1^*(X_i) - p \mu_0^*(X_i)\Big).\notag
\#
By the asymptotic equivalences established before, we also have 
\$
&\sqrt{n}\big(\tilde\theta_\deb  - \delta\big) \stackrel{d}{\to} N(0,\sigma_\deb^2),\quad \sqrt{n}\big(\hat\theta_\deb  - \delta\big) \stackrel{d}{\to} N(0,\sigma_\deb^2)
\$
for the same $\sigma_\deb^2$ in~\eqref{eq:sig_deb}. Such variance can be consistently estimated via 
\$
&\hat \sigma_\deb^2 = \frac{n}{n_t^2} \sum_{T_i=1} \big(A_i - \bar{A}(1)\big)^2 +  \frac{n}{n_c^2} \sum_{T_i=0} \big(B_i - \bar{B}(0)\big)^2,
\$ 
where, recalling the notations, $\bar{A}(1) = \frac{1}{n_t} \sum_{T_i=1}  A_i$, $\bar{B}(0) = \frac{1}{n_c} \sum_{T_i=0} B_i $, and 
\$
A_i = Y_i(1) - \frac{n_c}{n} \hat\mu_1 (X_i) - \frac{n_t}{n }\hat\mu_0 (X_i),\quad B_i = Y_i(0) - \frac{n_c}{n} \hat\mu_1(X_i) - \frac{n_t}{n} \hat\mu_0(X_i). 
\$
To see the consistency, we write $A_i=A_{1,i} + E_{1,i}$, where 
\$
A_{1,i} = Y_i(1) - \frac{n_c}{n}  \mu^*_1 (X_i) - \frac{n_t}{n } \mu^*_0 (X_i).
\$
By Cauchy-Schwarz inequality, we have 
\$
\bigg|\frac{1}{n_t} \sum_{T_i=1} A_i^2 - \frac{1}{n_t} \sum_{T_i=1} A_{1,i}^2 \bigg|= \bigg|\frac{2}{n_t} \sum_{T_i=1} A_i E_{1,i} + \frac{1}{n_t}\sum_{T_i=1} E_{1,i}^2\bigg| \leq \sqrt{\frac{1}{n_t} \sum_{T_i=1} A_{1,i}^2}\cdot \sqrt{\frac{1}{n_t}\sum_{T_i=1} E_{1,i}^2} + \frac{1}{n_t}\sum_{T_i=1} E_{1,i}^2.
\$
Again by Cauchy-Schwarz inequality, 
\$
\frac{1}{n_t}\sum_{T_i=1} E_{1,i}^2 
= 
\frac{1}{n_t}\sum_{T_i=1} \Big( \frac{n_c}{n} \Delta_1(X_i) + \frac{n_t}{n } \Delta_0(X_i) \Big)^2 
\leq 
\frac{n_c^2}{n^2} \frac{2}{n_t}  \sum_{T_i=1} \Delta_1(X_i)^2 + \frac{n_t^2}{n^2} \frac{2}{n_t} \sum_{T_i=1} \Delta_1(X_i)^2.
\$
For each $k\in[K]$, by Assumption~\ref{assump:conv} we have 
\$
\EE_c\bigg[\frac{2}{n_t}  \sum_{T_i=1} \Delta_1(X_i)^2\bigg] = 2\big\| \Delta_1\big\|_2^2 = o_P(1).
\$
Applying Lemma~\ref{lem:cond_to_op} and with similar arguments, we have 
\$
\frac{1}{n_t}\sum_{T_i=1} E_{1,i}^2 \leq \frac{n_c^2}{n^2} \cdot o_P(1) + \frac{n_t^2}{n^2} \cdot  o_P(1)
= o_P(1).
\$
It's not hard to see that $\frac{1}{n_t} \sum_{T_i=1} A_{1,i}^2 = O_P(1)$, hence 
\#\label{eq:ct_A2}
\frac{1}{n_t} \sum_{T_i=1} A_i^2 = \frac{1}{n_t} \sum_{T_i=1} A_{1,i}^2 + o_P(1).
\#
Meanwhile, by Cauchy-Schwarz inequality, we have 
\$
\bigg|\frac{1}{n_t}\sum_{T_i=1} E_{1,i} \bigg| \leq \sqrt{\frac{1}{n_t}\sum_{T_i=1} E_{1,i}^2} = o_P(1),
\$
hence 
\#\label{eq:ct_A}
\bar{A}(1) = \frac{1}{n_t} \sum_{T_i=1} A_i  = \frac{1}{n_t} \sum_{T_i=1}  A_{1,i}  + o_P(1) = \bar{A^*}(0)+o_P(1).
\#
Combining equations~\eqref{eq:ct_A2} and~\eqref{eq:ct_A}, we know 
\$
\frac{1}{n_t}\sum_{T_i=1} \big(A_i - \bar{A}(1)\big)^2 = \frac{1}{n_t} \sum_{T_i=1} A_i^2 - \Big(\frac{1}{n_t}\sum_{T_i=1} A_{i}\Big)^2 = \frac{1}{n_t} \sum_{T_i=1} A_{1,i}^2 - \Big(\frac{1}{n_t}\sum_{T_i=1} A_{1,i}\Big)^2 + o_P(1).
\$
Furthermore, we write $A_{1,i}= A_{2,i} + E_{2,i}$, where 
$
A_{2,i} = Y_i(1) - (1-p) \mu_1^*(X_i) - p \mu_0^*(X_i).
$
Here 
\$
\frac{1}{n_t}\sum_{T_i=1}E_{2,i}^2 = \frac{1}{n_t}\sum_{T_i=1}  \Big( \frac{n_t}{n} - p\Big)^2 \big(\mu_1^*(X_i) - \mu_0^*(X_i)\big)^2 = o_P(1),
\$
since $n_t/n - p =o_P(1)$. By similar arguments as for $A_i$ and $A_{i,1}$, we see that 
\$
&\frac{1}{n_t} \sum_{T_i=1} A_i^2 = \frac{1}{n_t} \sum_{T_i=1} A_{1,i}^2 + o_P(1) = \frac{1}{n_t} \sum_{T_i=1} A_{2,i}^2 + o_P(1),\\
&\frac{1}{n_t} \sum_{T_i=1} A_i  = \frac{1}{n_t} \sum_{T_i=1}  A_{1,i}  + o_P(1) = \frac{1}{n_t} \sum_{T_i=1}  A_{2,i}  + o_P(1).
\$
Therefore by the law of large numbers, we have 
\$
\frac{1}{n_t}\sum_{T_i=1} \big(A_i - \bar{A}(1)\big)^2 = \frac{1}{n_t} \sum_{T_i=1} A_{2,i}^2 - \Big(\frac{1}{n_t}\sum_{T_i=1} A_{2,i}\Big)^2 + o_P(1) = \Var(A_{2,i}) +o_P(1).
\$
By same arguments to $B_i$, we see that 
\$
\frac{1}{n_c}\sum_{T_i=0} \big(B_i - \bar{B}(0)\big)^2 = \Var(B_{2,i}) + o_P(1),
\$
where $B_{2,i} = Y_i(0) - (1-p) \mu_1^*(X_i) - p \mu_0^*(X_i)$. Further since $n/n_t\pto 1/p$ and $n/n_c\pto 1/(1-p)$, we arrive at 
\$
\hat \sigma_\deb^2 &= \frac{n}{n_t } \cdot \frac{1}{n_t} \sum_{T_i=1} \big(A_i - \bar{A}(1)\big)^2 +  \frac{n}{n_c } \cdot \frac{1}{n_c} \sum_{T_i=0} \big(B_i - \bar{B}(0)\big)^2 \\
&= \frac{1}{p}\Var(A_{2,i}) + \frac{1}{1-p}\Var(B_{2,i}) + o_P(1) = \sigma_{\deb}^2 +o_P(1).
\$
Therefore, we complete the proof of Theorem~\ref{thm:asymp_var}.
\end{proof}

\section{Proof of results on ratio metrics}

\subsection{Proof of asymptotic inference: without SDA} \label{app:rr_inference}
\begin{proof}[Proof of Theorem~\ref{thm:rr_inference}]
We first establish the asymptotics of $\hat\delta$ by connecting to its oracle counterpart 
\$
\hat\delta^* = \frac{ \EE[Y_i(1)] + \frac{1}{n} \sum_{i=1}^n A_i^*  }{ \EE[Z_i(1)] + \frac{1}{n} \sum_{i=1}^n B_i^*}- \frac{\EE[Y_i(0)] +  \sum_{i=1}^n C_i^*  }{ \EE[Z_i(0)] + \sum_{i=1}^n D_i^*} = \frac{\EE[Y_i(1)] + \bar{A^*}}{\EE[Z_i(1)] +\bar{B^*}} - \frac{\EE[Y_i(0)] + \bar{C^*}}{\EE[Z_i(0)] +\bar{D^*}},
\$
with $A_i^*$, $B_i^*$, $C_i^*$, $D_i^*$ defined in~\eqref{eq:rr_var}. To this end, we note that 
\$
&\sqrt{n}\big( \bar{A} - \bar{A^*} - \EE[Y_i(1)]\big) = \underbrace{\frac{1}{\sqrt{n}} \sum_{i=1}^n \Big(1 - \frac{T_i}{\hat p} \Big) \big(\hat\mu_1^Y(X_i) - \mu_{1,Y}^*(X_i) \big)}_{\textrm{(i)}} \\
&\quad + \underbrace{
    \sqrt{n}\Bigg(  \frac{1}{n} \sum_{i=1}^n \mu_{1,Y}^*(X_i) + \frac{1}{n_t} \sum_{T_i=1} \big(  Y_i(1) - \mu_{1,Y}^*(X_i) \big) - \EE[Y_i(1)] }_{\textrm{(ii)}}\\
&\qquad 
   \underbrace{ - \frac{1}{n} \sum_{i=1}^n \tilde\mu_{1,Y}^*(X_i) - \frac{1}{np} \sum_{T_i=1} \big(Y_i - \EE[Y_i(1)] -\tilde\mu_{1,Y}^*(X_i) \big)   \Bigg)
}_{\textrm{(ii)}}.
\$
We treat the two terms (i) and (ii) separately. Firstly, we denote  $\Delta_1(X_i) = \hat\mu_{1}^Y(X_i) - \mu_{1,Y}^*(X_i)$. Fix any $k\in[K]$, and denote $\EE_c$ as the expectation taken conditional on the sample splitting, $\cD^{(-k)}$ and $\cT$, under which $\Delta_1(X_i)$ for $i\in \cD^{(k)}$ are i.i.d. Under Assumption~\ref{assump:rr_conv}, we have $\EE_c[\Delta_1(X_i)^2]\stackrel{P}{\to}0$. Then  a straightforward calculation shows 
\$
\EE_c\Bigg[ \frac{1}{n_k} \bigg(  \sum_{i\in \cD^{(k)}} \Delta_1(X_i) - \frac{n_k}{n_{k,t}} \sum_{i\in \cD^{(k)},T_i=1} \Delta_1(X_i)   \bigg)^2\Bigg] = o_P(1).
\$
Then a limiting argument of a.s.~convergence and convergence in probability leads to 
\$
\eta_k:=\frac{1}{\sqrt{n_k}} \bigg(  \sum_{i\in \cD^{(k)}} \Delta_1(X_i) - \frac{n_k}{n_{k,t}} \sum_{i\in \cD^{(k)},T_i=1} \Delta_1(X_i)   \bigg) = o_P(1).
\$
Summing over $k\in[K]$ for a fixed number $K$ of folds, we have 
\$
\textrm{(i)}&=  \sum_{k=1}^K \frac{1}{\sqrt{n}}\bigg( \sum_{i\in \cD^{(k)}} \Delta_1(X_i) - \frac{n}{n_t} \sum_{i\in \cD^{(k)},T_i=1} \Delta_1(X_i) \bigg)\\
&= \sum_{k=1}^K \frac{\sqrt{n_k}}{\sqrt{n}} \bigg( \eta_k + \Big(\frac{n_k}{n_{k,t}} - \frac{n}{n_t} \Big) \sum_{i\in \cD^{(k)},T_i=1} \Delta_1(X_i)\bigg).
\$
Here 
\$
\bigg| \Big(\frac{n_k}{n_{k,t}} - \frac{n}{n_t} \Big) \sum_{i\in \cD^{(k)},T_i=1} \Delta_1(X_i)  \bigg| \leq \frac{n+n_t}{n_t\cdot (n_t/K-1)} \cdot \bigg| \sum_{i\in \cD^{(k)},T_i=1} \Delta_1(X_i)  \bigg| = o_P(1),
\$
since $\sum_{i\in \cD^{(k)},T_i=1} \Delta_1(X_i) = O_P(n)$ by $\EE_c[\Delta_1(X_i)^2]\stackrel{P}{\to}0$ and a standard limiting argument. Therefore we have $\textrm{(i)} = o_P(1)$. Moreover, note that 
\$
\textrm{(ii)}&= \sqrt{n}\bigg( \EE\big[\mu_{1,Y}^*(X_i)\big] + \frac{1}{n_t} \sum_{T_i=1} \big(  Y_i(1) -\EE[Y_i(1)]-  \mu_{1,Y}^*(X_i) \big)   - \frac{1}{np} \sum_{T_i=1} \big(Y_i - \EE[Y_i(1)] -\tilde\mu_{1,Y}^*(X_i) \big)   \bigg) \\
&=  \sqrt{n}\bigg( \frac{1}{n_t} -\frac{1}{np}\bigg) \cdot \sum_{i=1}^n T_i \Big(  Y_i(1) -\EE[Y_i(1)]-  \tilde\mu_{1,Y}^*(X_i) \Big) =   \sqrt{n}\bigg( \frac{1}{n_t} -\frac{1}{np}\bigg) \cdot O_P(\sqrt{n}) = o_P(1).
\$
Here the third equality follows from the fact that $T_i \big(  Y_i(1) -\EE[Y_i(1)]-  \tilde\mu_{1,Y}^*(X_i) \big)$ are i.i.d.~centered random variables. Putting them together, we arrive at 
\$
\bar{A} = o_P\big(1/\sqrt{n}\big) + \EE\big[Y_i(1)\big] + \bar{A^*}.
\$
Following exactly the same arguments, we can show that 
\$
\bar{B} &= o_P\big(1/\sqrt{n}\big) + \EE\big[Z_i(1)\big] + \bar{B^*},\\
\bar{C} &= o_P\big(1/\sqrt{n}\big) + \EE\big[Y_i(0)\big] + \bar{C^*},\\
\bar{D} &= o_P\big(1/\sqrt{n}\big) + \EE\big[Z_i(0)\big] + \bar{D^*}.
\$
Therefore, we have 
\$
\hat\delta = \bar{A} / \bar{B} - \bar{C} /\bar{D} = o_P\big(1/\sqrt{n}\big) + \frac{\EE\big[Y_i(1)\big] + \bar{A^*}}{\EE\big[Z_i(1)\big] + \bar{B^*}} + \frac{\EE\big[Y_i(0)\big] + \bar{C^*}}{\EE\big[Z_i(0)\big] + \bar{D^*}} = \hat\delta^* +  o_P\big(1/\sqrt{n}\big).
\$
Moerover, by Delta method or Taylor expansion for $\hat\delta^*$, we obtain the asymptotic expansion 
\$
\hat\delta = \hat\delta^* +  o_P\big(1/\sqrt{n}\big) =  \delta + \frac{1}{n} \sum_{i=1}^n \phi_\delta(Y_i,Z_i,X_i,T_i)  + o_P\big(1/\sqrt{n}\big),
\$
which leads to the asymptotic variance in~\eqref{eq:rr_var}. 

We now proceed to establish the consistency of the variance estimator. Firstly, for notational simplicity, we write $\phi_\delta(Y_i,Z_i,X_i,T_i) = \theta_A A_i^* + \theta_B B_i^* + \theta_C C_i^*+ \theta_D D_i^*$, where 
\$
\theta_A = \frac{1}{\EE[Z_i(1)]},~\theta_B = - \frac{\EE[Y_i(1)]}{\EE[Z_i(1)]^2},~ \theta_c = -\frac{1}{\EE[Z_i(0)]},~ \theta_D = \frac{\EE[Y_i(0)]}{\EE[Z_i(0)]^2}.
\$
Writing $\tilde{Y}_i(w)=Y_i(w)-\EE[Y_i(w)]$ and $\tilde{Z}_i(w) = Z_i(w) - \EE[Z_i(w)]$ for $w\in\{0,1\}$, we have 
\$
& \phi_\delta (Y_i,Z_i,X_i,T_i) = T_i \phi_1(Y_i(1),Z_i(1),X_i,T_i) + (1-T_i) \phi_0(Y_i(0),Z_i(0),X_i,T_i),
\$
where 
\$
&\phi_1(Y_i(1),Z_i(1),X_i,T_i) \\
&=   \theta_A \Big( \tilde\mu_{1,Y}^*(X_i) + \frac{1}{p}\big(\tilde{Y}_i(1) - \tilde\mu_{1,Y}^*(X_i) \big) \Big) \\
&\qquad +  
\theta_B \Big( \tilde\mu_{1,Z}^*(X_i) + \frac{1}{p}\big(\tilde{Z}_i(1) - \tilde\mu_{1,Z}^*(X_i) \big) \Big) + \theta_C \tilde\mu_{0,Y}^*(X_i) + \theta_D \tilde\mu_{0,Z}^*(X_i) , \\
&\phi_0(Y_i(0),Z_i(0),X_i,T_i) \\
&=    \theta_A   \tilde\mu_{1,Y}^*(X_i) + \theta_B  \tilde\mu_{1,Z}^*(X_i)   + \theta_C \Big(  \tilde\mu_{0,Y}^* (X_i) + \frac{1}{1-p}\big(\tilde{Y}_i(0)-\tilde\mu_{0,Y}^*(X_i)\Big) \\
&\qquad   + \theta_D \Big(  \tilde\mu_{0,Z}^* (X_i) + \frac{1}{1-p}\big(\tilde{Z}_i(0)-\tilde\mu_{0,Z}^*(X_i)\Big).
\$
Therefore, since $T_i$ are independent of all other random variables, we have 
\$
\Var\big(\phi_\delta(Y_i,Z_i,X_i,T_i)\big) &= p\cdot  \Var\big( \phi_1(Y_i(1),Z_i(1),X_i,T_i)\big) + (1-p)\cdot \Var\big( \phi_0(Y_i(0),Z_i(0),X_i,T_i)\big).
\$
Since $\hat{p}\stackrel{P}{\to} p$, it suffices to show that $\hat{V}_1 \stackrel{P}{\to} \Var\big( \phi_1(Y_i(1),Z_i(1),X_i,T_i)\big)$ and $\hat{V}_0  \stackrel{P}{\to}\Var\big( \phi_0(Y_i(0),Z_i(0),X_i,T_i)\big)$. We showcase the first convergence result, while the second follows exactly the same arguments. Since adding constants to random variables does not change the variance, we have 
\#\label{eq:rr_V1}
\Var\big( \phi_1(Y_i(1),Z_i(1),X_i,T_i) \big) = \EE[D_{1,i}^2] - \EE[D_{1,i}]^2 = o_P(1) + \frac{1}{n_t}\sum_{T_i=1} D_{1,i}^2 - \Big( \frac{1}{n_t}\sum_{T_i=1} D_{1,i} \Big)^2
\#
according to the law of large numbers, 
where 
\$
D_{1,i} &=  \theta_A \Big( \mu_{1,Y}^*(X_i) + \frac{1}{p}\big( {Y}_i(1) -  \mu_{1,Y}^*(X_i) \big) \Big) \\
&\qquad +  
\theta_B \Big( \mu_{1,Z}^*(X_i) + \frac{1}{p}\big( {Z}_i(1) -  \mu_{1,Z}^*(X_i) \big) \Big) + \theta_C  \mu_{0,Y}^*(X_i) + \theta_D  \mu_{0,Z}^*(X_i) \\
&= \theta_A\Big(1-\frac{1}{p}\Big) \mu_{1,Y}^*(X_i) + \frac{\theta_A}{p} Y_i(1) + \theta_B\Big(1-\frac{1}{p}\Big) \mu_{1,Z}^*(X_i) + \frac{\theta_B}{p}Z_i(1) + \theta_C  \mu_{0,Y}^*(X_i) + \theta_D  \mu_{0,Z}^*(X_i).
\$
Meanwhile, the law of large numbers gives $\hat\theta_A\stackrel{P}{\to} \theta_A$, $\hat\theta_B\stackrel{P}{\to} \theta_B$, $\hat\theta_C\stackrel{P}{\to} \theta_C$, $\hat\theta_D\stackrel{P}{\to} \theta_D$ for 
\$
\hat\theta_A = 1/\bar{Z}(1), \quad \hat\theta_B = - \bar{Y}(1)/\bar{Z}(1)^2,\quad \hat\theta_C = -1/\bar{Z}(0),\quad \hat\theta_D = \bar{Y}(0)/\bar{Z}(0)^2.
\$
Writing $D_{1,i} = D_{2,i} +E_{2,i}$, where 
\$
D_{2,i} =\hat\theta_A\Big(1-\frac{1}{\hat p}\Big) \mu_{1,Y}^*(X_i) + \frac{\hat\theta_A}{\hat p} Y_i(1) + \hat\theta_B\Big(1-\frac{1}{\hat p}\Big) \mu_{1,Z}^*(X_i) + \frac{\hat\theta_B}{\hat p}Z_i(1) + \hat\theta_C  \mu_{0,Y}^*(X_i) + \hat\theta_D  \mu_{0,Z}^*(X_i).
\$
Then by Cauchy-Schwarz inequality, 
\$
\bigg|\frac{1}{n_t}\sum_{T_i=1} D_{1,i}^2  - \frac{1}{n_t}\sum_{T_i=1} D_{2,i}^2\bigg| = \bigg|\frac{2}{n_t} \sum_{T_i=1} D_{2,i}E_{2,i} +  \frac{1}{n_t}\sum_{T_i=1} E_{2,i}^2 \bigg|\leq \sqrt{\frac{1}{n_t}\sum_{T_i=1} D_{2,i}^2} \cdot \sqrt{\frac{1}{n_t}\sum_{T_i=1} E_{2,i}^2} + \frac{1}{n_t}\sum_{T_i=1} E_{2,i}^2.
\$
Again by Cauchy-Schwarz inequality, 
\#
\frac{1}{n_t}\sum_{T_i=1} E_{2,i}^2 &= \frac{1}{n_t} \sum_{T_i=1}\bigg[ \bigg( \hat\theta_A\Big(1-\frac{1}{\hat p}\Big) -  \theta_A\Big(1-\frac{1}{  p}\Big) \bigg) \mu_{1,Y}^*(X_i) + \Big( \frac{\hat\theta_A}{\hat p} - \frac{ \theta_A}{ p}  \Big) Y_i(1)+ \big( \hat\theta_C - \theta_C\big)  \mu_{0,Y}^*(X_i)\notag \\
&\qquad \qquad  + \bigg( \hat\theta_B\Big(1-\frac{1}{\hat p}\Big) -  \theta_B\Big(1-\frac{1}{  p}\Big) \bigg)  \mu_{1,Z}^*(X_i) + \Big(  \frac{\hat\theta_B}{\hat p} - \frac{\theta_B}{p}\Big) Z_i(1)  + \big( \hat\theta_D - \theta_D\big)  \mu_{0,Z}^*(X_i)  \bigg]^2 \notag \\
&\leq \frac{6}{n_t} \sum_{T_i=1} \bigg( \hat\theta_A\Big(1-\frac{1}{\hat p}\Big) -  \theta_A\Big(1-\frac{1}{  p}\Big) \bigg)^2 \mu_{1,Y}^*(X_i)^2  
+ \frac{6}{n_t} \sum_{T_i=1}  \Big( \frac{\hat\theta_A}{\hat p} - \frac{ \theta_A}{ p}  \Big)^2 Y_i(1)^2 \notag  \\
&\qquad \qquad  +\frac{6}{n_t} \sum_{T_i=1} \bigg( \hat\theta_B\Big(1-\frac{1}{\hat p}\Big) -  \theta_B\Big(1-\frac{1}{  p}\Big) \bigg)^2  \mu_{1,Z}^*(X_i)^2 + \frac{6}{n_t} \sum_{T_i=1}\Big(  \frac{\hat\theta_B}{\hat p} - \frac{\theta_B}{p}\Big)^2 Z_i(1)^2\notag \\
&\qquad \qquad \qquad +\frac{6}{n_t} \sum_{T_i=1} \big( \hat\theta_C - \theta_C\big)^2  \mu_{0,Y}^*(X_i)^2 + \frac{6}{n_t} \sum_{T_i=1}\big( \hat\theta_D - \theta_D\big)  \mu_{0,Z}^*(X_i). \label{eq:rr_e2}
\#
Since $\hat\theta_A \pto \theta_A$ and $\hat p \pto p$, by the law of large numbers we have 
\$
\frac{6}{n_t} \sum_{T_i=1}  \hat\theta_A\Big(1-\frac{1}{\hat p}\Big) = o_P(1).
\$
Similarly, all terms in the right-most handed equation~\eqref{eq:rr_e2} are $o_P(1)$, which implies
\$
\frac{1}{n_t}\sum_{T_i=1} D_{1,i}^2  =  \frac{1}{n_t}\sum_{T_i=1} D_{2,i}^2 + o_P(1).
\$
Now writing $D_{2,i} = D_{3,i} + E_{3,i}$ so that 
\$
D_{3,i} =\hat\theta_A\Big(1-\frac{1}{\hat p}\Big) \hat\mu_{1,Y}(X_i) + \frac{\hat\theta_A}{\hat p} Y_i(1) + \hat\theta_B\Big(1-\frac{1}{\hat p}\Big) \hat \mu_{1,Z} (X_i) + \frac{\hat\theta_B}{\hat p}Z_i(1) + \hat\theta_C  \hat\mu_{0,Y} (X_i) + \hat\theta_D  \hat\mu_{0,Z} (X_i).
\$
It's straightforward to see that $D_{3,i} = d_{1,i}$ for $d_{1,i}$ defined in 
Theorem~\ref{thm:rr_inference}. 
Similar arguments yield 
\$
\bigg|\frac{1}{n_t}\sum_{T_i=1} D_{3,i}^2  - \frac{1}{n_t}\sum_{T_i=1} D_{2,i}^2\bigg| 
&= \bigg|-\frac{2}{n_t} \sum_{T_i=1} D_{2,i}E_{3,i} +  \frac{1}{n_t}\sum_{T_i=1} E_{3,i}^2 \bigg|\\
&\leq \sqrt{\frac{1}{n_t}\sum_{T_i=1} D_{2,i}^2} \cdot \sqrt{\frac{1}{n_t}\sum_{T_i=1} E_{3,i}^2} + \frac{1}{n_t}\sum_{T_i=1} E_{3,i}^2.
\$
Here similar to the arguments in~\eqref{eq:rr_e2}, Cauchy-Schwarz inequality implies 
\#
\frac{1}{n_t}\sum_{T_i=1} E_{3,i}^2 &\leq  \hat\theta_A^2\Big(1-\frac{1}{\hat p}\Big)^2  \cdot \frac{4}{n_t} \sum_{T_i=1}  \big[ \hat\mu_1^Y(X_i) - \mu_{1,Y}^*(X_i)\big]^2    +\hat\theta_B^2 \Big(1-\frac{1}{\hat p}\Big)^2 \cdot \frac{4}{n_t} \sum_{T_i=1}  \big[ \hat\mu_1^Z(X_i) - \mu_{1,Z}^*(X_i)\big]^2 \notag \\
&\qquad \qquad  +\hat\theta_C ^2 \cdot \frac{4}{n_t} \sum_{T_i=1}   \big[ \hat\mu_0^Y(X_i) -  \mu_{0,Y}^*(X_i)\big]^2 + \hat\theta_D ^2  \cdot \frac{4}{n_t} \sum_{T_i=1}  \big[ \hat\mu_0^Z(X_i) -  \mu_{0,Z}^*(X_i)\big]^2 . \label{eq:rr_e3}
\#
We decompose into the $K$ cross-fitting folds, so that 
\$
\frac{4}{n_t} \sum_{T_i=1}  \big[ \hat\mu_1^Y(X_i) - \mu_{1,Y}^*(X_i)\big]^2 = \sum_{k=1}^K \frac{4}{n_t} \sum_{i\in \cD^{(k)}, T_i=1} \big[ \hat\mu_1^Y(X_i) - \mu_{1,Y}^*(X_i)\big]^2.
\$
Note that under $\EE_c$, the conditional expectation given $\cD^{(-k)}$ and $\cT$, the summation terms $ \big[ \hat\mu_1^Y(X_i) - \mu_{1,Y}^*(X_i)\big]^2$ for $i\in \cD^{(k)}$ and $T_i=1$ are i.i.d., and 
\$
\EE_c\bigg[\frac{4}{n_t} \sum_{i\in \cD^{(k)}, T_i=1} \big[ \hat\mu_1^Y(X_i) - \mu_{1,Y}^*(X_i)\big]^2 \bigg] = \frac{4n_{k,t}}{n_t} \big\|\hat\mu_1^{Y,(k)} - \hat\mu_{1,Y}^*\big\|_2 = o_P(1).
\$
By a standard argument on a.s.~convergence and convergence in probability, we have 
\$
\frac{4}{n_t} \sum_{i\in \cD^{(k)}, T_i=1} \big[ \hat\mu_1^Y(X_i) - \mu_{1,Y}^*(X_i)\big]^2 = o_P(1)
\$
for all $k\in[K]$. Adding them up, we have 
\$
\frac{4}{n_t} \sum_{T_i=1}  \big[ \hat\mu_1^Y(X_i) - \mu_{1,Y}^*(X_i)\big]^2 = o_P(1).
\$
Exactly the same arguments for the rest terms in~\eqref{eq:rr_e3} leads to $\frac{1}{n_t}\sum_{T_i=1} E_{3,i}^2 = o_P(1)$. Also it's straightforward that $\frac{1}{n_t}\sum_{T_i=1} D_{2,i}^2 = O_P(1)$, hence 
\#\label{eq:rr_d1^2}
\frac{1}{n_t}\sum_{T_i=1} D_{1,i}^2  =  \frac{1}{n_t}\sum_{T_i=1} D_{2,i}^2 + o_P(1) =  \frac{1}{n_t}\sum_{T_i=1} D_{3,i}^2 + o_P(1).
\#
Moreover, by Cauchy-Schwarz inequality, we have 
\$
o_P(1) = \frac{1}{n_t }\sum_{T_i=1} E_{2,i}^2 \geq \bigg(\frac{1}{n_t }\sum_{T_i=1} E_{2,i}\bigg)^2,\quad o_P(1) = \frac{1}{n_t }\sum_{T_i=1} E_{3,i}^2 \geq \bigg(\frac{1}{n_t }\sum_{T_i=1} E_{3,i}\bigg)^2,
\$
which indicates
\#\label{eq:rr_d1}
\frac{1}{n_t}\sum_{T_i=1} D_{1,i} = \frac{1}{n_t}\sum_{T_i=1} D_{3,i} + \frac{1}{n_t }\sum_{T_i=1} E_{2,i} + \frac{1}{n_t }\sum_{T_i=1} E_{3,i} =  \frac{1}{n_t}\sum_{T_i=1} D_{3,i} + o_P(1).
\#
Combining equations~\eqref{eq:rr_V1},~\eqref{eq:rr_d1^2} and~\eqref{eq:rr_d1}, we have 
\$
\Var\big( \phi_1(Y_i(1),Z_i(1),X_i,T_i) \big) =    \frac{1}{n_t}\sum_{T_i=1} D_{3,i}^2 - \Big( \frac{1}{n_t}\sum_{T_i=1} D_{3,i} \Big)^2 + o_P(1) = \hat{V}_1 + o_P(1),
\$
By similar arguments for $\phi_0(Y_i(0),Z_i(0),X_i,T_i)$, we can show 
\$
\Var\big( \phi_0(Y_i(0),Z_i(0),X_i,T_i) \big) =   \hat{V}_0 + o_P(1).
\$
Furthermore, the fact that  $\hat{p}\pto p$ establishes the consistency of $\hat{\sigma}_\delta^2$ for $\sigma_\delta^2$. By Slutsky's theorem, 
\$
\sqrt{n} (\hat\delta - \delta) /  \hat{\sigma}_\delta \stackrel{d}{\to} N(0,1),
\$
Therefore, we conclude the proof of Theorem~\ref{thm:rr_inference}.
\end{proof}

\subsection{Proof of asymptotic inference: under SDA} \label{app:r_inference}
\begin{proof}[Proof of Theorem~\ref{thm:r_inference}]
Firstly, we note that $\frac{1}{n}\sum_{i=1}^n \Gamma_i$ coincides with $\tilde\theta_\deb$ defined in~\eqref{eq:db_count_0}, except that $\hat\mu_w(X_i)$ is substitued by $\hat\mu_w(X_i,Z_i)$ here for $w\in\{0,1\}$. Therefore, by exactly the same arguments as in the proof of Theorem~\ref{thm:asymp_var} for $\tilde\theta_{\deb}$, we can show that 
\$
\frac{1}{n}\sum_{i=1}^n \Gamma_i &= \EE\big[Y_i(1)\big]- \EE\big[Y_i(0)\big] + \frac{1}{n} \sum_{i=1}^n \big( \mu_1^*(X_i,Z_i) - \mu_0^*(X_i,Z_i)\big) \\
&\qquad + \frac{1}{n_t}\sum_{T_i=1} \big(Y_i(1) - \mu_1^*(X_i,Z_i)\big) - \frac{1}{n_c} \sum_{T_i=0} \big( Y_i(0) - \mu_0^*(X_i,Z_i)\big) + o_P\big(1/\sqrt{n}\big) \\
&= \EE\big[Y_i(1)\big]- \EE\big[Y_i(0)\big] + \frac{1}{n} \sum_{i=1}^n \big( \tilde\mu_1^*(X_i,Z_i) - \tilde\mu_0^*(X_i,Z_i)\big) \\
&\qquad + \frac{1}{n_t}\sum_{T_i=1} \big(\tilde Y_i(1) - \tilde\mu_1^*(X_i,Z_i)\big) - \frac{1}{n_c} \sum_{T_i=0} \big(\tilde Y_i(0) - \tilde\mu_0^*(X_i,Z_i)\big) + o_P\big(1/\sqrt{n}\big),
\$
where $\tilde{Y}_i(w) = Y_i(w) - \EE[Y_i(w)]$ and $\tilde\mu_w^*(X_i,Z_i) = \mu_w^*(X_i,Z_i) - \EE[\mu_w^*(X_i,Z_i)]$ for $w\in \{0,1\}$ are centered random variables. Therefore 
\$
&\frac{1}{n} \sum_{i=1}^n \big( \tilde\mu_1^*(X_i,Z_i) - \tilde\mu_0^*(X_i,Z_i)\big)  + 
\frac{1}{n_t}\sum_{T_i=1} \big(\tilde Y_i(1) - \tilde\mu_1^*(X_i,Z_i)\big) - \frac{1}{n_c} \sum_{T_i=0} \big(\tilde Y_i(0) - \tilde\mu_0^*(X_i,Z_i)\big) \\
&= \frac{1}{n} \sum_{i=1}^n \bigg[ \big( \tilde\mu_1^*(X_i,Z_i) - \tilde\mu_0^*(X_i,Z_i)\big)  +  \frac{T_i}{\hat p}
    \big(\tilde Y_i(1) - \tilde\mu_1^*(X_i,Z_i)\big) - \frac{1-T_i}{1-\hat p} \big(\tilde Y_i(0) - \tilde\mu_0^*(X_i,Z_i)\big) \bigg]\\
&= \frac{1}{n} \sum_{i=1}^n \Gamma_i^* +
\Big(\frac{1}{\hat p}-\frac{1}{p}\Big)\cdot \frac{1}{n} \sum_{i=1}^n T_i\big(\tilde Y_i(1) - \tilde\mu_1^*(X_i,Z_i)\big) \\
&\qquad +  
\Big(\frac{1}{1-\hat p}-\frac{1}{1-p}\Big) \cdot \frac{1}{n} \sum_{i=1}^n  (1-T_i) \cdot \big(\tilde Y_i(0) - \tilde\mu_0^*(X_i,Z_i)\big),
\$
where $\Gamma_i^*$ is defined in \eqref{eq:r_infl}. 
Since $\hat p\pto p$ and $\tilde{Y}_i(w)-\tilde\mu_w^*(X_i,Z_i)$ are i.i.d.~with mean zero, the last two terms in the above summation are both $o_P(1/\sqrt{n})$. Therefore 
\$
\frac{1}{n} \sum_{i=1}^n \Gamma_i = \EE\big[Y_i(1)\big]- \EE\big[Y_i(0)\big] + \frac{1}{n} \sum_{i=1}^n \Gamma_i^* + o_P(1/\sqrt{n}).
\$
By a standard argument of delta method or Taylor expansion, we see that 
\$
\hat\delta' &= \frac{\frac{1}{n}\sum_{i=1}^n \Gamma_i }{\frac{1}{n}\sum_{i=1}^n Z_i} = \frac{\EE [Y_i(1) ]- \EE [Y_i(0) ] + \frac{1}{n} \sum_{i=1}^n \Gamma_i^* + o_P(1/\sqrt{n})}{\EE[Z_i] +\frac{1}{n}\sum_{i=1}^n (Z_i - \EE[Z_i])} \\
&= \frac{\EE [Y_i(1) ]- \EE [Y_i(0) ]}{\EE[Z_i]} + \frac{1}{n}\sum_{i=1}^n \phi_{\delta}'(Y_i,Z_i,X_i,T_i) +o_P(1/\sqrt{n}).
\$
This leads to the asymptotic distribution $\sqrt{n}(\hat\delta' - \delta') \stackrel{d}{\to} N(0,\sigma_{\delta'}^2)$, where 
\$
\sigma_{\delta'}^2 = \Var\big( \phi_{\delta}'(Y_i,Z_i,X_i,T_i) \big).
\$
We now proceed to establish the asymptotic consistency of the variance estimator. To this end, writing $1/\EE[Z_i] = \theta_1$ and $- \frac{\EE[Y_i(1)]-\EE[Y_i(0)]}{\EE[Z_i]^2} = \theta_2$, we have 
\$
\phi_\delta'(Y_i,Z_i,X_i,T_i) &= T_i \cdot \phi_1(Y_i(1),Y_i(0),Z_i,X_i) + (1-T_i) \cdot \phi_0(Y_i(1),Y_i(0), Z_i,X_i),\\
\text{where}\quad \phi_1(Y_i(1),Y_i(0),Z_i,X_i) &= \theta_1 \tilde\mu_1^*(X_i,Z_i) - \theta_1\tilde\mu_0^*(X_i,Z_i) + \frac{\theta_1}{p}\big(\tilde{Y}_i(1) - \tilde\mu_1^*(X_i,Z_i)\big)  + \theta_2 \big(Z_i - \EE[Z_i]\big),\\
\phi_0(Y_i(1),Y_i(0),Z_i,X_i) &= \theta_1 \tilde\mu_1^*(X_i,Z_i) - \theta_1\tilde\mu_0^*(X_i,Z_i) - \frac{\theta_1}{1-p}\big(\tilde{Y}_i(0) - \tilde\mu_0^*(X_i,Z_i)\big)  + \theta_2 \big(Z_i - \EE[Z_i]\big).
\$
Since $T_i$'s are independent of all other random variables, we have 
\$
\Var\big( \phi_{\delta}'(Y_i,Z_i,X_i,T_i) \big) = p\cdot \Var\big( \phi_1(Y_i(1),Y_i(0),Z_i,X_i) \big) + (1-p)\cdot \Var\big(\phi_0(Y_i(1),Y_i(0),Z_i,X_i)   \big).
\$
Since adding constants does not change variance, the law of large numbers implies
\$
\Var\big( \phi_1(Y_i(1),Y_i(0),Z_i,X_i) \big) &= \Var\Big(   \theta_1  \mu_1^*(X_i,Z_i) - \theta_1 \mu_0^*(X_i,Z_i) + \frac{\theta_1}{p}\big( {Y}_i(1) -  \mu_1^*(X_i,Z_i)\big)  + \theta_2 Z_i \Big)\\
&= \frac{1}{n_t} \sum_{T_i=1} (D_i^*)^2 - \Big(  \frac{1}{n_t} \sum_{T_i=1} D_i^* \Big)^2 + o_P(1),
\$
where we write
\$
D_{i}^* &= \theta_1  \mu_1^*(X_i,Z_i) - \theta_1 \mu_0^*(X_i,Z_i) + \frac{\theta_1}{p}\big( {Y}_i(1) -  \mu_1^*(X_i,Z_i)\big)  + \theta_2 Z_i \\
\text{and}\quad D_{1,i} &= \hat \theta_1  \hat \mu_1(X_i,Z_i) -\hat  \theta_1\hat  \mu_0 (X_i,Z_i) + \frac{\hat \theta_1}{\hat  p}\big( {Y}_i(1) -  \hat \mu_1 (X_i,Z_i)\big)  + \hat \theta_2 Z_i,
\$
with $\hat\theta_1 = 1/\bar{Z} \pto \theta_1$, $\hat\theta_2 = -\frac{\bar{Y}(1)-\bar{Y}(0)}{\bar{Z}^2} \pto \theta_2$ and $\hat p= n_t/n\pto p$ being consistent estimators. Now we define $D_{1,i} = D_{2,i} + E_{2,i}$, where 
\$
D_{2,i} = \hat \theta_1   \mu_1^*(X_i,Z_i) -\hat  \theta_1  \mu_0^*(X_i,Z_i) + \frac{\hat \theta_1}{\hat  p}\big( {Y}_i(1) -  \mu_1^*(X_i,Z_i)\big)  + \hat \theta_2 Z_i.
\$
Also write $\Delta_w(X_i,Z_i) = \hat\mu_w(X_i,Z_i) - \mu_w^*(X_i,Z_i)$ for $w\in\{0,1\}$. Then by Cauchy-Schwarz inequality, 
\$
\frac{1}{n_t}\sum_{T_i=1} E_{2,i}^2 \leq \hat\theta_1^2 \Big(1-\frac{1}{\hat p}\Big)^2\cdot  \frac{2}{n_t} \sum_{T_i=1} \Delta_1(X_i,Z_i)^2 +  \hat\theta_1^2 \cdot  \frac{2}{n_t} \sum_{T_i=1} \Delta_0(X_i,Z_i)^2.
\$
Using similar arguments as in the proof of Theorem~\ref{thm:rr_inference}, we have 
\$
\frac{1}{n_t}\sum_{T_i=1} E_{2,i}^2 = o_P(1), \quad \frac{1}{n_t}\sum_{T_i=1} E_{2,i} = o_P(1),
\$
and
\$
\frac{1}{n_t} \sum_{T_i=1} D_{1,i} ^2 - \Big(  \frac{1}{n_t} \sum_{T_i=1} D_{1,i} \Big)^2
&= \frac{1}{n_t} \sum_{T_i=1} D_{2,i} ^2 - \Big(  \frac{1}{n_t} \sum_{T_i=1} D_{2,i} \Big)^2 + o_P(1).
\$
Writing $D_{2,i} = D_i^* + E_{3,i}$, where 
\$
E_{3,i} &= (\hat\theta_1 - \theta_1)  \mu_1^*(X_i,Z_i) -(\hat\theta_1 - \theta_1)   \mu_0^*(X_i,Z_i) +\Big(\frac{\hat\theta_1}{\hat p}-  \frac{\theta_1}{p}\Big)\big( {Y}_i(1) -  \mu_1^*(X_i,Z_i)\big)  + (\hat\theta_1 - \theta_2) Z_i,
\$
similar arguments as the proof of Theorem~\ref{thm:rr_inference} lead to 
\$
\frac{1}{n_t}\sum_{T_i=1} E_{3,i}^2 = o_P(1), \quad \frac{1}{n_t}\sum_{T_i=1} E_{3,i} = o_P(1),
\$
and (noting that $D_{1,i} = g_{1,i}$ as defined in~\eqref{eq:r_var_est})
\$
\Var\big( \phi_1(Y_i(1),Y_i(0),Z_i,X_i) \big) 
&= \frac{1}{n_t} \sum_{T_i=1} (D_i^*)^2 - \Big(  \frac{1}{n_t} \sum_{T_i=1} D_i^* \Big)^2 + o_P(1) \\
&=  \frac{1}{n_t} \sum_{T_i=1} D_{2,i} ^2 - \Big(  \frac{1}{n_t} \sum_{T_i=1} D_{2,i} \Big)^2 + o_P(1) \\
&=  \frac{1}{n_t} \sum_{T_i=1} D_{1,i} ^2 - \Big(  \frac{1}{n_t} \sum_{T_i=1} D_{1,i} \Big)^2 + o_P(1) = \hat{V}_1 + o_P(1)
\$
for $\hat{V}_1$ defined in~\eqref{eq:r_var_est}. With exactly the same arguments, we can show 
\$
\Var\big( \phi_0(Y_i(1),Y_i(0),Z_i,X_i) \big) = \hat{V}_0 + o_P(1).
\$
Combining the above two results and the fact that  $\hat p\pto p$, we know $\hat\sigma_{\delta'}^2 \pto \sigma_{\delta'}^2$. By Slutsky's theorem, we have 
\$
\sqrt{n}(\hat\delta' - \delta')/ \hat\sigma_{\delta'} \stackrel{d}{\to} N(0,1),
\$
which establishes the validity of confidence interval. Therefore, we complete the proof of Theorem~\ref{thm:r_inference}.

\end{proof}

\section{Results on semiparametric efficiency}

In this section, we provide proofs for semiparametric efficiency of the proposed procedures, which are based on variance bounds established in Appendix~\ref{app:var_bd}.

\subsection{Semiparametric efficiency for count metric} \label{app:count_efficiency}

\begin{proof}[Proof of Theorem~\ref{thm:efficiency}]
We first compute the asymptotic variance of $\hat\theta_\deb$ and $\tilde\theta_\deb$ under Assumptions~\ref{assump:treat} and~\ref{assump:consist}. Here under the consistency condition, Assumption~\ref{assump:conv} is satisfied with $\mu_1^*=\mu_1$ and $\mu_0^*=\mu_0$. Thus by Theorem~\ref{thm:asymp_var}, the asymptotic variance is 
\$
\sigma_*^2 &=  \frac{1}{p}\Var\Big(Y_i(1) - (1-p) \mu_1(X_i) - p \mu_0(X_i)\Big) + \frac{1}{1-p}\Var\Big( Y_i(0) - (1-p) \mu_1(X_i) - p \mu_0(X_i)\Big) .
\$
To establish the semiparametric efficiency, we quote the classical semiparametric efficiency results on treatment effects, see Lemma~\ref{lem:optimal_aipw}. Here since $T_i$'s are independent of all other random variables, the condition in Lemma~\ref{lem:optimal_aipw} is satisfied with $e(X_i) \equiv p$. Recall that $\mu_1(x)=\EE[Y_i(1)\given X_i=x]$ and $\mu_0(x)=\EE[Y_i(0)\given X_i=x]$. The asymptotic variance bound for $\tau$ is 
\$
V^* &= \EE\bigg[  \frac{\Var(Y_i(1)\given X_i)}{e(X_i)} + \frac{\Var(Y_i(0)\given X_i)}{1-e(X_i)} + \big( \EE[Y_i(1)\given X_i] - \EE[Y_i(0)\given X_i] -  \tau\big)^2   \bigg]\\
&= \frac{1}{p} \EE\big[ \Var(Y_i(1)\given X_i) ] + \frac{1}{1-p}\EE\big[ \Var(Y_i(0)\given X_i) ] + \Var\big( \mu_1(X_i) - \mu_0(X_i)  \big).
\$
Since $\mu_1(X_i)=\EE[Y_i(1)\given X_i]$, we know $Y_i(1) - \mu_1(X_i)$ is uncorrelated with $\mu_1(X_i)-\mu_0(X_i)$. Similarly for $Y_i(0)-\mu_0(X_i)$. Thus,
\$
\sigma_*^2 &=  \frac{1}{p}\Var\Big(Y_i(1) - \mu_1(X_i) + p \big[ \mu_1(X_i) -  \mu_0(X_i)\big]\Big) \\
&\qquad + \frac{1}{1-p}\Var\Big( Y_i(0) - \mu_0(X_i) - (1-p) \big[ \mu_1(X_i) -  \mu_0(X_i)\big]\Big) \\
&= \frac{1}{p}\Var\big(Y_i(1)-\mu_1(X_i)\big) + p\cdot \Var\big( \mu_1(X_i) - \mu_0(X_i)  \big) \\
&\qquad + \frac{1}{1-p} \Var\big(Y_i(0)-\mu_0(X_i)\big) + (1-p)\cdot \Var\big( \mu_1(X_i) - \mu_0(X_i)  \big) \\
&= \frac{1}{p}  \EE\big[ \Var(Y_i(1)\given X_i) ] + \frac{1}{1-p}\EE\big[ \Var(Y_i(0)\given X_i) ] + \Var\big( \mu_1(X_i) - \mu_0(X_i)  \big) = V^*.
\$
Therefore, $\sigma_*^2$ is the semiparametric asymptotic variance bound for $\tau$ in this case.
\end{proof}

\subsection{Semiparametric efficiency for ratio metric: without SDA}\label{app:rr_efficiency}

\begin{proof}[Proof of Theorem~\ref{thm:rr_optimal}]
We first show that $\sigma_{\delta,\dag}^2$ is the semiparametric asymptotic variance bound for $\delta$. Applying Theorem~\ref{thm:rr_var_bound}, the condition $T_i\indep (Y_i(1),Y_i(0),Z_i(1),Z_i(0))\given X_i$ is satisfied in the randomized experiment setting and the propensity score is $e(X_i)\equiv p$. Therefore, the efficient influence function is $\phi_{\delta}^*$ in~\eqref{eq:rr_eff_infl_general}, which coincides with $\phi_{\delta,\dag}$ in~\eqref{eq:rr_eff_infl}. So the variance bound is given by $\sigma^2_{\delta,\dag}$. 

It remains to show that $\sqrt{n}(\hat\delta - \delta)\stackrel{d}{\to} N(0,\sigma_{\delta,\dag}^2)$. By Assumption~\ref{assump:rr_consistency}, Assumption~\ref{assump:conv} is satisfied with $\mu_{w,Y}^* = \mu_{w,Y}$ and $\mu_{w,Z}^*=\mu_{w,Z}$ for $w\in\{0,1\}$. In this case, $\EE[\mu_{w,Y}^*(X_i)]=\EE[Y_i(w)]$ and $\EE[\mu_{w,Z}^*(X_i)] = \EE[Z_i(w)]$ for $w\in\{0,1\}$. Thus $A_i^*=A_i^\dag$. Therefore, applying Theorem~\ref{thm:rr_inference}, we have $\sqrt{n}(\hat\delta - \delta)\stackrel{d}{\to} N(0,\sigma_{\delta,\dag}^2)$, which completes the proof. 
\end{proof}

\subsection{Semiparametric efficiency for ratio metric: with SDA}\label{app:r_efficiency}

\begin{proof}[Proof of Theorem~\ref{thm:r_optimality}]
    We first show that $\sigma_{\delta',\dag}^2$ is the semiparametric variance bound for $\delta'$. The condition $T_i\indep (Y_i(1),Y_i(0))\given (X_i,Z_i)$ in Theorem~\ref{thm:r_var_bound} is satisfied here in the randomized experiment since $T_i$ is independent of all other random variables. Thus the efficient influence function is given by $\phi_{\delta'}^*$ in~\eqref{eq:r_eff_infl_general}, which coincides with $\phi_{\delta',\dag}$ in~\eqref{eq:r_eff_infl} since $e(X_i,Z_i)\equiv p$. Hence the variance bound is exactly $\sigma_{\delta',\dag}^2$. 
    
    Then we show that $\sqrt{n}(\hat\delta' - \delta') \stackrel{d}{\to} N(0,\sigma_{\delta,\dag}^2)$. According to Assumption~\ref{assump:r_consistency}, Assumption~\ref{assump:conv} is satisfied with $\mu_1^*=\mu_1$ and $\mu_0^*=\mu_0$, and $\EE[Y(1)]=\EE[ \mu_1(X,Z)]$, $\EE[Y(0)]=\EE[\mu_0(X,Z)]$. Hence $\Gamma_i^*=\Gamma_i^\dag$ and $\sigma_{\delta'}^2 = \sigma_{\delta',\dag}^2$, and applying Theorem~\ref{thm:r_inference} yields $\sqrt{n}(\hat\delta'-\delta')\stackrel{d}{\to} N(0,\sigma_{\delta',\dag}^2)$, hence completing the proof of Theorem~\ref{thm:r_optimality}.
    \end{proof}

\section{Semiparametric variance bounds}\label{app:var_bd}
In this part, we establish semiparametric variance bounds for the estimands $\tau$, $\delta$ and $\delta'$, which forms the basis for the optimality of the proposed procedures. 

\subsection{Backgrounds and definitions} \label{app:var_bd_background}
We first introduce basic definitions that are necessary for semiparametric variance bounds from the literature of semiparametric efficiency. Let $\mu$ be a fixed $\sigma$-finite measure on $(\cX,\cB)$, and let $L_2(\mu)$ be the space of all $L_2$-integrable functions with respect to $\mu$. 

Let $\mathcal{P}$ be a collection of measures on $(\cX,\cB)$ dominated by $\mu$, which can be equivalently viewed as a subset of $L_2(\mu)$ via the embedding 
\#\label{eq:embed_l2}
P ~\to ~p=\frac{\ud P}{\ud \mu} ~\to ~r = \sqrt{p}\in L_2(\mu)
\#
for all $P\in \mathcal{P}$, where $\frac{\ud P}{\ud \mu}$ is the Radon–Nikodym derivative.

\begin{definition}[Regular parametric space]
Suppose $\cQ$ is a collection of measures on $(\cX,\cB)$ dominated by $\mu$, which is parametrized by $\theta\in \Theta\subset \cR^k$, so that $\cQ= \{P_\theta:\theta\in \Theta\}$. Denote the embedding $r(\theta)=\sqrt{p(\theta)}\in L_2(\mu)$, where $p(\theta) = \frac{\ud P_\theta}{\ud \mu} $ is the Radon–Nikodym derivative. We say the parametrization $\theta\to p(\theta)$ is \textit{regular} if $\Theta$ is open, and for any $\theta_0\in \Theta$, the map $\theta\to r(\theta)$ is Fr\'echet differentiable at $\theta_0$ with derivative $\dot{r}(\theta_0)\in L_2(\mu)^k$, and the matrix $\int \dot{r}(\theta_0)\dot{r}(\theta_0)^\top \ud \mu$ is non-singular.
\end{definition}

\begin{definition}[Regular non-parametric space]
Let $\mathcal{P}$ be a collection of measures on $(\cX,\cB)$ dominated by $\mu$. An element $P_0\in \cP$ is \textit{regular} if it is contained in a regular parametric subspace of $\cP$. We say $\cP$ is \textit{regular} if every point $P\in \cP$ is regular.
\end{definition}

\begin{definition}[Tangent space]
Let $\mathcal{P}$ be a collection of measures on $(\cX,\cB)$ dominated by $\mu$, which can be viewed as a subset $\cR$ of $L_2(\mu)$ via the embedding in Equation \eqref{eq:embed_l2}. For a fixed $P_0\in \cP$, we define the \textit{tangent set} at $P_0$ as the union of all (one-dimensional) tangent spaces of curves $C\subset \cP$ passing through $P_0$. We define \textit{tangent space} at $P_0$ as the closed linear span of the tangent set at $P_0$. 
\end{definition}

Note that in the above definition, everything considered is in the Hilbert space $L_2(\mu)$. 

\begin{definition}[Pathwise differentiability]
	
Let $\mathcal{P}$ be a regular nonparametric model, and $\beta:\mathcal{P}\to \cR$ be a functional on $\mathcal{P}$. Under the embedding in Equation \eqref{eq:embed_l2}, let $\beta$ also denote the mapping $\beta:L_2(\mu)\to \cR$ with $\beta(r(P)) = \beta(P)$. For a fixed $P_0\in \cP$, let $\dot{\cP}$ denote the tangent space of $\cP$ at $P_0$, and let $r_0=r(P_0)$ be its embedding in $L_2(\mu)$, with tangent space $\dot{\cR}$ at $r_0$.
We say $\beta$ is pathwise differentiable on $\cP$ at $P_0$ (or $r_0$) if there exists a bounded linear functional $\dot{\beta}(r_0):\dot{\cP}\to \cR$ so that 
\$ 
{\beta}(r(\eta)) =\beta(r_0)+\eta \dot{\beta}(r_0)(t) +o(|\eta|)
\$ 
for any curve $r(\cdot)$ in $\cR$ passing through $r_0=r(0)$ and $\dot{r}(0)=t$. We also define $\dot{\beta}(P_0):\dot{\cP}\to \cR$ by $\dot{\beta}(P_0)(h)=\dot{\beta}(r_0)(hr_0/2)$, so that $\beta(P_\eta) = \beta(P_0) + \eta \dot{\beta}(P_0)(h)+o(|\eta|)$, where $P_\eta$ is a curve in $\cP$ corresponding to $r(\eta)$ in $\cR$ and $h=2t/r_0$.
\end{definition}


The following definition gives the notion of efficient influence function by defining it as the projection of pathwise derivative on the tangent space, and the efficient variance, or semiparametric variance bound, is defined as the variance of efficient influence function. The efficient variance is actually a lower bound for all asymptotically linear estimators, hence any asymptotically linear estimator whose asymptotic variance coincide with the efficient variance is semiparametrically efficient. More detailed results are in the textbook~\cite{bickel1993efficient}. 

\begin{definition}[Efficient influence function]\label{def:eff_inf_fct}
	Let $\mathcal{P}$ be a regular nonparametric model, and $\beta:\mathcal{P}\to \cR$ be a functional on $\mathcal{P}$ that is pathwise differentiable at $P_0$ with derivative $\dot{\beta}(P_0)$. We define the \textit{efficient influence function} as $\tilde{\phi}(\cdot,P_0\given \beta,\cP) = \Pi(\dot{\beta}(P_0)\given \dot{\cP})$, the projection of $\dot{\beta}(P_0)\in L_2(\mu)$ onto the tangent space $\dot{\cP}$. We also define the \textit{efficient variance} as 
	\$
	\widetilde{V}_\beta = \|\tilde\phi(\cdot,P_0\given \beta,\cP)\|_0^2~,
	\$
	where the norm $\|\cdot\|_0$ is defined as $\|f\|_0^2 = \langle f,f\rangle_0 = \int f^2 \ud P_0$ for any $f\in L_2(\mu)$.
\end{definition}

Based on these definitions, given an estimand $\beta=\beta(\cP)$, we outline the following protocol  of providing semiparametric variance bounds, which will be applied in Appedix~\ref{app:rr_var_bound} and~\ref{app:r_var_bound}.
\begin{itemize}
    \item Step 1: Characterize parametric submodels $\{P_\theta\colon \theta\in \Theta\}$ for $\cP$. 
    \item Step 2: Find the tangent space $\cT$ of any regular parametric submodel $\{P_\theta\colon \theta\in \Theta\}$.
    \item Step 3: Find pathwise derivative $\ud \beta(\theta)/\ud \theta$ of the estimand at the parametric submodel.
    \item Step 4: Project the pathwise derivative, as a member in the Hilbert space $L_2(\mu)$, onto the tangent space. Then the projection is the efficient influence function, whose variance is the efficient variance. 
\end{itemize}

\subsection{Semiparametric variance bound for $\tau$}
The semiparametric variance bound for $\tau$ has been established in a general setting, as in the following lemma. 
\begin{lemma}[Theorem 1 in~\cite{hahn1998role}]\label{lem:optimal_aipw}
Suppose $T_i\indep (Y_i(1),Y_i(0)\given X_i$, and let $e(x) = \PP(T_i=1\given X_i=x)$ be the propensity score. Then the asymptotic variance bounds for $\tau = \EE[Y_i(1)-Y_i(0)]$ is 
\$
\EE\bigg[  \frac{\Var(Y_i(1)\given X_i)}{e(X_i)} + \frac{\Var(Y_i(0)\given X_i)}{1-e(X_i)} + \big( \EE[Y_i(1)\given X_i] - \EE[Y_i(0)\given X_i] -  \tau\big)^2   \bigg].
\$
\end{lemma}

\subsection{Semiparametric variance bound for $\delta$} \label{app:rr_var_bound}

We establish the efficient influence function and efficient variance for $\delta$ in the general setting following the protocol in Appedix~\ref{app:var_bd_background}. Recall that $\mu_{w,Y}(x)=\EE[Y_i(w)\given X_i=x]$ and $\mu_{w,Z}(x)=\EE[Z_i(w)\given X_i=x]$, $w\in\{0,1\}$ are the true conditional mean functions.

\begin{theorem}\label{thm:rr_var_bound}
Suppose $(X_i,Z_i(0),Z_i(1),Y_i(0),Y_i(1))\iid \PP$ and $T_i\indep (Y_i(1),Y_i(0), Z_i(1),Z_i(0))\given X_i$. Let $e(x)=\PP(T_i=1\given X_i=x)$ be the propensity score. Then the efficient influence function for  $\delta = \EE[Y_i(1)]/\EE[Z_i(1)] - \EE[Y_i(0)]/\EE[Z_i(0)]$ is $\phi_\delta^*(Y_i,Z_i,X_i,T_i)$, where 
\#
\phi_\delta^*(Y_i,Z_i,X_i,T_i) &= \frac{A_i^{**}}{\EE[Z_i(1)]} - \frac{\EE[Y_i(1)]}{\EE[Z_i(1)]^2} B_i^{**} - \frac{C_i^{**}}{\EE[Z_i(0)]} + \frac{\EE[Y_i(0)]}{\EE[Z_i(0)]^2} D_i^{**}, \label{eq:rr_eff_infl_general}\\
\textrm{where}\quad A_i^{**} &=   \mu_{1,Y} (X_i ) + \frac{T_i}{e(X_i)}  \big(Y_i  - \mu_{1,Y} (X_i ) \big) - \EE[Y_i(1)], \notag \\
 B_i^{**} &=   \mu_{1,Z} (X_i ) + \frac{T_i}{e(X_i)}  \big(Z_i  - \mu_{1,Z} (X_i ) \big)-\EE[Z_i(1)], \notag \\
C_i^{**} &= \mu_{0,Y}^*(X_i ) + \frac{1-T_i}{1-e(X_i)}  \big(Y_i   -  \mu_{0,Y} (X_i ) \big)- \EE[Y_i(0)], \notag  \\
D_i^{**} &= \mu_{0,Z}^*(X_i ) + \frac{1-T_i}{1-e(X_i)}  \big(Z_i -  \mu_{0,Z} (X_i ) \big)- \EE[Z_i(0)]. \notag
\#
The asymptotic variance bound for $\delta$ is $V_\delta^* = \Var(\phi_\delta^*(Y_i,Z_i,X_i,T_i))$. 
\end{theorem}

\begin{proof}[Proof of Theorem~\ref{thm:rr_var_bound}]
We follow the general protocol in Appendix~\ref{app:var_bd_background}, which is similar to \cite{hahn1998role}. 

\paragraph{Step 1: Parametric submodel.} Under the strong ignorability $T_i\indep (Y_i(1),Y_i(0), Z_i(1),Z_i(0))\given X_i$, the density of the joint distribution of $(X_i,Z_i(0),Z_i(1),Y_i(0),Y_i(1))\iid \PP$ is given by 
\$
\bar{q}(x,y_0,y_1,z_0,z_1,t) = f(x) \cdot f(y_0,y_1,z_0,z_1\given x) \cdot e(x)^t (1-e(x))^{1-t},
\$
where $f(y_0,y_1,z_0,z_1\given x)$ is the conditional density function. The joint distribution of the observed quantities $(X_i,Y_i,Z_i,T_i)$ where $Y_i=Y_i(T_i)$ and $Z_i=Z_i(T_i)$ is thus 
\$
q(x,y,z,t) = \big[ f_1(y,z\given x) e(x)    \big]^{t}\cdot \big[ f_0(y,z\given x) (1-e(x))    \big]^{1-t} f(x),
\$
where $f_w(y,z\given x)$ is the conditional density function of $Y(w),Z(w)$ given $X$ for $w\in\{0,1\}$. 
Then a regular parametric submodel containing $\PP$ is characterized by 
\$
q(x,y,z,t;\theta) = \big[ f_1(y,z\given x,\theta ) e(x,\theta )    \big]^{t}\cdot \big[ f_0(y,z\given x,\theta) (1-e(x,\theta))    \big]^{1-t} f(x,\theta),
\$
where the submodel is parametrizedby $\theta\in \Theta$, and the density equals the true density $q(x,y,z,t)$ when $\theta=\theta_0$. Also, the functions $f_1(y,z\given x,\theta)$, $f_0(y,z\given x,\theta)$, $f(x,\theta)$ and $e(x,\theta)$ is differentiable w.r.t.~$\theta$ at $\theta_0$. 

\paragraph{Step 2: Tangent space.} By a straightforward calculation, the score function at $\PP(\theta)$ with density $q(x,y,z,t;\theta)$ is 
\$
s(\theta;x,y,z,t) = t\cdot \frac{\dot{f}_1(y,z\given x,\theta)}{f_1(y,z\given x,\theta)} + (1-t)\cdot \frac{\dot{f}_0(y,z\given x,\theta)}{f_0(y,z\given x,\theta)} + \frac{t-e(x,\theta)}{e(x,\theta)(1-e(x,\theta))}\cdot \dot{e}(x,\theta) + \frac{\dot{f}(x,\theta)}{f(x,\theta)},
\$
where 
\$
\dot{f}_w(y,z\given x,\theta) = \frac{\partial}{\partial \theta}f_w(y,z\given x,\theta),~w\in\{0,1\},\quad \dot{e}(x,\theta) = \frac{\partial}{\partial \theta} {e}(x,\theta),\quad \dot{f}(x,\theta) = \frac{\partial}{\partial \theta} f(x,\theta).
\$
Therefore, the tangent space at $\PP$ is 
\#\label{eq:rr_tangent}
\cT = \Big\{& t\cdot s_1(y,z\given x ) + (1-t)\cdot s_0(y,z\given x ) +  a(x) \big(t-e(x )\big) + s(x)  \colon \\
 &\quad \int s_w(y,z\given x ) f_w(y,z\given x ) \ud y\ud z= 0,\forall w\in \{0,1\},~\forall x,z,\\
&\qquad ~ \int s(x)f(x)\ud x = 0, ~a(x) \text{ is square integrable}\Big\}.\notag
\#
\paragraph{Step 3: Pathwise derivative.} The target $\delta$ can be written as 
\begin{displaymath}
\delta(\theta) = \frac{\int yf_1(y,z\given x,\theta) f(x,\theta) \ud x\ud y \ud z}{\int zf_1(y,z\given x,\theta) f(x,\theta) \ud x\ud y \ud z} - \frac{\int yf_0(y,z\given x,\theta) f(x,\theta) \ud x\ud y \ud z}{\int zf_0(y,z\given x,\theta) f(x,\theta) \ud x\ud y \ud z}.
\end{displaymath}
Taking derivative w.r.t.~$\theta$, we have 
\#\label{eq:diff_delta}
\frac{\partial \delta(\theta)}{\partial \theta} &= \frac{1}{\EE[Z(1)]} \bigg(\int y\dot{f}_1(y,z\given x,\theta) f(x,\theta) \ud x\ud y \ud z + \int y f_1(y,z\given x,\theta) \dot{f}(x,\theta) \ud x\ud y \ud z\bigg) \notag \\
&\quad - \frac{\EE[Y(1)]}{\EE[Z(1)]^2} \bigg(\int z\dot{f}_1(y,z\given x,\theta) f(x,\theta) \ud x\ud y \ud z + \int z f_1(y,z\given x,\theta) \dot{f}(x,\theta) \ud x\ud y \ud z\bigg) \notag \\
&\quad - \frac{1}{\EE[Z(0)]} \bigg(\int y\dot{f}_0(y,z\given x,\theta) f(x,\theta) \ud x\ud y \ud z + \int y f_0(y,z\given x,\theta) \dot{f}(x,\theta) \ud x\ud y \ud z\bigg) \notag \\
&\quad + \frac{\EE[Y(0)]}{\EE[Z(0)]^2} \bigg(\int z\dot{f}_0(y,z\given x,\theta) f(x,\theta) \ud x\ud y \ud z + \int z f_0(y,z\given x,\theta) \dot{f}(x,\theta) \ud x\ud y \ud z\bigg).
\#

\paragraph{Step 4: Projection onto tangent space.} We now establish the projection of $\partial\delta(\theta)/\partial\theta$ onto the tangent space $\cT$ by showing that the influence function $\phi_\delta^*$ in~\eqref{eq:rr_eff_infl} is in $\cT$, and is essentially the projection. Note that $\phi_\delta^*\in \cT$ by setting 
\$
s_1(y,z\given x) &= \frac{1}{e(x)} \frac{1}{\EE[Z(1)]} \big(y-\mu_{1,Y}(x)\big) - \frac{1}{e(x)} \frac{\EE[Y(1)]}{\EE[Z(1)]^2} \big(z-\mu_{1,z}(x)\big),\\
s_0(y,z\given x)&=-\frac{1}{1-e(x)} \frac{1}{\EE[Z(0)]} \big(y-\mu_{0,Y}(x)\big) + \frac{1}{1-e(x)} \frac{\EE[Y(0)]}{\EE[Z(0)]^2} \big(z-\mu_{0,z}(x)\big),\\
t(x)&= \frac{\mu_{1,Y}(x) -  \EE[Y(1)] }{\EE[Z(1)]} - \frac{\EE[Y(1)]}{\EE[Z(1)]^2}\big( \mu_{1,z}(x) -\EE[Z(1)]\big)\\
&\quad - \frac{\mu_{0,Y}(x) -  \EE[Y(0)] }{\EE[Z(0)]} + \frac{\EE[Y(0)]}{\EE[Z(0)]^2}\big( \mu_{0,z}(x) -\EE[Z(0)]\big).\quad a(x)\equiv 0
\$
in the definition of $\cT$ in~\eqref{eq:rr_tangent}. Now we show that $\phi_\delta^*$ is the projection onto $\cT$, for which it suffices to show  
\$
\frac{\partial \delta(\theta)}{\partial \theta} = \EE_\theta\big[  s(\theta;X_i,Y_i,Z_i,T_i) \cdot \phi_\delta^*(Y_i,Z_i,X_i,T_i)  \big],
\$
which implies $\phi_\delta^*(\cdot)$ is the Riesz representation of $\partial\delta(\theta)/\partial\theta$ in $L_2(\mu)$. In light of~\eqref{eq:rr_eff_infl} and~\eqref{eq:diff_delta}, it suffices to show the following four relations
\$
\EE\big[A_i^{**} \cdot s(\theta;X_i,Y_i,Z_i,T_i) \big] &= \int y\dot{f}_1(y,z\given x,\theta) f(x,\theta) \ud x\ud y \ud z + \int y f_1(y,z\given x,\theta) \dot{f}(x,\theta) \ud x\ud y \ud z,\\
\EE\big[B_i^{**} \cdot s(\theta;X_i,Y_i,Z_i,T_i) \big] &=\int z\dot{f}_1(y,z\given x,\theta) f(x,\theta) \ud x\ud y \ud z + \int z f_1(y,z\given x,\theta) \dot{f}(x,\theta) \ud x\ud y \ud z,\\
\EE\big[C_i^{**} \cdot s(\theta;X_i,Y_i,Z_i,T_i) \big] &= \int y\dot{f}_0(y,z\given x,\theta) f(x,\theta) \ud x\ud y \ud z + \int y f_0(y,z\given x,\theta) \dot{f}(x,\theta) \ud x\ud y \ud z\\
\EE\big[D_i^{**} \cdot s(\theta;X_i,Y_i,Z_i,T_i) \big] &= \int z\dot{f}_0(y,z\given x,\theta) f(x,\theta) \ud x\ud y \ud z + \int z f_0(y,z\given x,\theta) \dot{f}(x,\theta) \ud x\ud y \ud z.
\$
By symmetric roles of $Y_i,Z_i$ and $T_i,1-T_i$, we only show the first equation, and others follow exactly the same arguments. Note that 
\$
&\int y\dot{f}_1(y,z\given x,\theta) f(x,\theta) \ud x\ud y \ud z + \int y f_1(y,z\given x,\theta) \dot{f}(x,\theta) \ud x\ud y \ud z \\
&= \int y\frac{\dot{f}_1(y,z\given x,\theta)}{f_1(y,z\given x,\theta)}\cdot    f_1(y,z\given x,\theta) f(x,\theta) \ud x\ud y \ud z + \int y \frac{ \dot{f}(x,\theta)}{  {f}(x,\theta)}\cdot f_1(y,z\given x,\theta){f}(x,\theta) \ud x\ud y \ud z \\
&= \EE\bigg[ Y(1) \cdot \frac{\dot{f}_1(Y(1),Z(1)\given X,\theta)}{f_1(Y(1),Z(1)\given X,\theta)}   \bigg] + \EE\bigg[ Y(1) \cdot \frac{\dot{f} (X,\theta)}{f(X,\theta)}   \bigg].
\$
On the other hand, 
\$
\EE\big[A_i^{**} \cdot s(\theta;X_i,Y_i,Z_i,T_i) \big] &= \EE\bigg[A_i^{**} \cdot T_i \cdot \frac{\dot{f}_1(Y ,Z \given X,\theta)}{f_1(Y ,Z \given X,\theta)}\bigg] + \EE\bigg[A_i^{**} \cdot (1-T_i) \cdot \frac{\dot{f}_0(Y ,Z \given X,\theta)}{f_0(Y ,Z \given X,\theta)}\bigg] \\
&\qquad + \EE\bigg[A_i^{**} \cdot \frac{T_i-e(X_i)}{e(X_i)(1-e(X_i))} \cdot\dot{e}(X_i,\theta)\bigg] +  \EE\bigg[A_i^{**} \cdot \frac{\dot{f}(X_i,\theta)}{f(X_i,\theta)}     \bigg].
\$
By the definition of $A_i^{**}$, 
\$
\EE\bigg[A_i^{**} \cdot T_i \cdot \frac{\dot{f}_1(Y ,Z \given X,\theta)}{f_1(Y ,Z \given X,\theta)}\bigg] 
&= \EE\bigg[ \big( \mu_{1,Y}(X ) - \EE[Y (1)] \big)\cdot T \cdot \frac{\dot{f}_1(Y  ,Z  \given X,\theta)}{f_1(Y ,Z \given X,\theta)}\bigg] \\
&\qquad + \EE\bigg[ \big(Y (1) -  \mu_{1,Y}(X )   \big)\cdot \frac{T }{e(X)} \cdot \frac{\dot{f}_1(Y ,Z \given X,\theta)}{f_1(Y ,Z \given X,\theta)}\bigg],
\$
where by the tower property and takeout property of conditional expectations, 
\$
&\EE\bigg[ \big( \mu_{1,Y}(X ) - \EE[Y (1)] \big)\cdot T \cdot \frac{\dot{f}_1(Y  ,Z  \given X,\theta)}{f_1(Y ,Z \given X,\theta)}\bigg] \\
&= \EE\bigg[ \big( \mu_{1,Y}(X ) - \EE[Y (1)] \big)\cdot \EE[T\given X] \cdot \EE\Big[ \frac{\dot{f}_1(Y  ,Z  \given X,\theta)}{f_1(Y ,Z \given X,\theta)}\Biggiven X\Big]\bigg] \\
&= \EE\bigg[ \big( \mu_{1,Y}(X ) - \EE[Y (1)] \big)\cdot e(X) \cdot \frac{\dot{f}_1(Y  ,Z  \given X,\theta)}{f_1(Y ,Z \given X,\theta)}\bigg] = 0,
\$
where the last equality follows the fact that for any measurable function $g(\cdot)$, 
\$
\EE\bigg[g(X)\frac{\dot{f}_1(Y  ,Z  \given X,\theta)}{f_1(Y ,Z \given X,\theta)}\bigg] = \int_x g(x)f(x,\theta)  \int_{y,z} \dot{f}_1(y,z\given x,\theta)\ud y \ud z \ud x= 0.
\$
Similarly, 
\$
&\EE\bigg[ \big(Y (1) -  \mu_{1,Y}(X )   \big)\cdot \frac{T }{e(X)} \cdot \frac{\dot{f}_1(Y ,Z \given X,\theta)}{f_1(Y ,Z \given X,\theta)}\bigg] \\
&=\EE\bigg[ \big(Y (1) -  \mu_{1,Y}(X )   \big)\cdot  \frac{\dot{f}_1(Y ,Z \given X,\theta)}{f_1(Y ,Z \given X,\theta)}\bigg]  = \EE\bigg[ Y(1) \cdot \frac{\dot{f}_1(Y(1),Z(1)\given X,\theta)}{f_1(Y(1),Z(1)\given X,\theta)}   \bigg].
\$
Meanwhile, still by the conditional independence of $Y_i(1),Z_i(1)$ and $T_i$ given $X_i$ as well as properties of conditional expectations, 
\$
\EE\bigg[A_i^{**} \cdot (1-T_i) \cdot \frac{\dot{f}_0(Y ,Z \given X,\theta)}{f_0(Y ,Z \given X,\theta)}\bigg] &= \EE\bigg[ \big(\mu_{1,Y}(X) - \EE[Y(1)]\big)\cdot (1-T ) \cdot \frac{\dot{f}_0(Y ,Z \given X,\theta)}{f_0(Y ,Z \given X,\theta)}\bigg]\\
&= \EE\bigg[ \big(\mu_{1,Y}(X) - \EE[Y(1)]\big)\cdot \big(1-e(X)\big) \cdot \frac{\dot{f}_0(Y ,Z \given X,\theta)}{f_0(Y ,Z \given X,\theta)}\bigg] = 0.
\$
For the third term, 
\$
&\EE\bigg[A_i^{**} \cdot \frac{T_i-e(X_i)}{e(X_i)(1-e(X_i))} \cdot\dot{e}(X_i,\theta)\bigg] \\
&= \EE\bigg[\Big(\mu_{1,Y} (X  )  - \EE[Y (1)]\Big) \cdot \frac{T -e(X )}{e(X )(1-e(X ))} \cdot\dot{e}(X ,\theta)\bigg] \\
&\qquad +   \EE\bigg[  \frac{T }{e(X )}  \big(Y  - \mu_{1,Y} (X  ) \big)  \cdot \frac{T -e(X )}{e(X )(1-e(X ))} \cdot\dot{e}(X ,\theta)\bigg] \\
&= \EE\bigg[\Big(\mu_{1,Y} (X  )  - \EE[Y (1)]\Big) \cdot \frac{\EE[T\given X] -e(X )}{e(X )(1-e(X ))} \cdot\dot{e}(X ,\theta)\bigg] \\
&\qquad +   \EE\bigg[  \frac{T }{e(X )}  \big(Y(1)  - \mu_{1,Y} (X  ) \big)  \cdot \frac{T -e(X )}{e(X )(1-e(X ))} \cdot\dot{e}(X ,\theta)\bigg]\\
&=  \EE\bigg[  \frac{T }{e(X )}  \big(Y(1)  - \mu_{1,Y} (X  ) \big)  \cdot \frac{T -e(X )}{e(X )(1-e(X ))} \cdot\dot{e}(X ,\theta)\bigg] \\
&=  \EE\bigg[  \frac{T }{e(X )}  \big(\EE[Y(1)\given X]  - \mu_{1,Y} (X  ) \big)  \cdot \frac{T -e(X )}{e(X )(1-e(X ))} \cdot\dot{e}(X ,\theta)\bigg]  = 0.
\$
Here the second equality follows from the tower property and takeout property and the fact that $TY=TY(1)$, the forth equality follows from the conditional independence of $Y(1)$ and $T$ given $X$ as well as the tower property and takeout property. The last term satisfies 
\$
\EE\bigg[A_i^{**} \cdot \frac{\dot{f}(X_i,\theta)}{f(X_i,\theta)} \bigg] &= \EE\bigg[ \Big(  \mu_{1,Y} (X )- \EE[Y (1)]\Big) \cdot \frac{\dot{f}(X ,\theta)}{f(X ,\theta)} \bigg]  + \EE\bigg[  \frac{T }{e(X )}  \big(Y   - \mu_{1,Y} (X) \big) \cdot \frac{\dot{f}(X ,\theta)}{f(X,\theta)} \bigg]\\
&= \EE\bigg[  \frac{T }{e(X )}  \big(Y(1)  - \mu_{1,Y} (X) \big) \cdot \frac{\dot{f}(X ,\theta)}{f(X,\theta)} \bigg] =  \EE\bigg[  Y(1)  \cdot \frac{\dot{f}(X ,\theta)}{f(X,\theta)} \bigg].
\$
The above derivation  uses the fact that $TY=TY(1)$, $Y(1)\indep T\given X$ and
$
\EE\big [ g(X)\cdot \frac{\dot{f}(X ,\theta)}{f(X ,\theta)} \big] = 0
$
for all measurable function $g(\cdot)$. Combining the four terms, we arrive at the first of the desired relations. By symmetric roles of $Y_i,Z_i$ and $T_i,1-T_i$, the other four follow similarly. Thus, we complete the proof that $\phi_\delta^*$ is the projection of $\partial \delta(\theta)/\partial \theta$ onto $\cT$. Hence, $\phi_\delta^*$ is the efficient influence function, and its variance is the semiparametric asymptotic variance bound for estimating $\delta$. 
\end{proof}

\subsection{Semiparametric variance bound for $\delta'$} \label{app:r_var_bound}

Similar to the case of $\delta$, we establish the efficient influence function and efficient variance for $\delta'$ in the general setting following the protocol in Appedix~\ref{app:var_bd_background}. Recall that $\mu_{w}(x,z)=\EE[Y_i(w)\given X_i=x,Z_i=z]$, $w\in\{0,1\}$ are the true conditional mean functions.

\begin{theorem}\label{thm:r_var_bound}
Suppose $(X_i,Z_i,Y_i(0),Y_i(1))\iid \PP$ and $T_i\indep (Y_i(1),Y_i(0))\given (X_i,Z_i)$. Let $e(x,z)=\PP(T_i=1\given X_i=x,Z_i=z)$ be the propensity score. Then the efficient influence function for  $\delta' = \EE[Y_i(1)]/\EE[Z_i] - \EE[Y_i(0)]/\EE[Z_i]$ is $\phi_{\delta'}^*(Y_i,Z_i,X_i,T_i)$, where 
\#
\phi_{\delta'}^*(Y_i,Z_i,X_i,T_i) &= \frac{\Gamma_i^{**}}{\EE[Z_i]} - \frac{\EE[Y_i(1)]-\EE[Y_i(0)]}{\EE[Z_i]^2}\big(Z_i-\EE[Z_i]), \label{eq:r_eff_infl_general}\\
\textrm{where}\quad \Gamma_i^{**} &=   \mu_{1} (X_i,Z_i ) - \mu_0(X_i,Z_i) - \big( \EE[Y_i(1)] - \EE[Y_i(0)] \big) \\
&\quad  + \frac{T_i}{e(X_i,Z_i)}  \big(Y_i  - \mu_{1 } (X_i,Z_i ) \big) - \frac{1-T_i}{1-e(X_i,Z_i)}\big(Y_i  - \mu_{0 } (X_i,Z_i ) \big) . \notag 
\#
The asymptotic variance bound for $\delta$ is $V_{\delta'}^* = \Var(\phi_{\delta'}^*(Y_i,Z_i,X_i,T_i))$. 
\end{theorem}

\begin{proof}[Proof of Theorem~\ref{thm:r_var_bound}]
The proof follows the protocol in Appendix~\ref{app:var_bd_background}.

\paragraph{Step 1: Parametric submodel.} Under the strong ignorability $T_i\indep (Y_i(1),Y_i(0))\given (X_i,Z_i)$, the density of the joint distribution of $(X_i,Z_i,Y_i(0),Y_i(1),T_i)\iid \PP$ is given by 
\$
\bar{q}(x,y_0,y_1,z,t) = f(x,z) \cdot f(y_0,y_1,\given x,z) \cdot e(x)^t (1-e(x))^{1-t},
\$
where $f(y_0,y_1\given x,z)$ is the conditional density function. The joint distribution of the observed quantities $(X_i,Y_i,Z_i,T_i)$ where $Y_i=Y_i(T_i)$ is 
\$
q(x,y,z,t) = \big[ f_1(y\given x,z) e(x,z)    \big]^{t}\cdot \big[ f_0(y \given x,z) (1-e(x,z))    \big]^{1-t} f(x,z),
\$
where $f_w(y\given x,z)$ is the conditional density function of $Y(w)$ given $X,Z$ for $w\in\{0,1\}$. 
Then a regular parametric submodel containing $\PP$ is characterized by 
\$
q(x,y,z,t;\theta) = \big[ f_1(y\given x,z,\theta ) e(x,z,\theta )    \big]^{t}\cdot \big[ f_0(y \given x,z,\theta) (1-e(x,z,\theta))    \big]^{1-t} f(x,z,\theta),
\$
where the submodel is parametrizedby $\theta\in \Theta$, and the density equals the true density $q(x,y,z,t)$ when $\theta=\theta_0$. Also, the functions $f_1(y,z\given x,\theta)$, $f_0(y\given x,z,\theta)$, $f(x,z,\theta)$ and $e(x,z,\theta)$ are differentiable w.r.t.~$\theta$ at $\theta=\theta_0$. 

\paragraph{Step 2: Tangent space.} 
By the joint distribution of the observable, the score function at $\theta$ or $\PP(\theta)$ is 
\#\label{eq:r_score}
s(\theta;x,y,z,t) = t\cdot \frac{\dot{f}_1(y \given x,z,\theta)}{f_1(y\given x,z,\theta)} + (1-t)\cdot \frac{\dot{f}_0(y\given x,z,\theta)}{f_0(y\given x,z,\theta)} + \frac{t-e(x,z,\theta)}{e(x,z,\theta)(1-e(x,z,\theta))}\cdot \dot{e}(x,z,\theta) + \frac{\dot{f}(x,z,\theta)}{f(x,z,\theta)},
\#
where 
\$
\dot{f}_w(y\given x,z,\theta) = \frac{\partial}{\partial \theta}f_w(y\given x,z,\theta),~w\in\{0,1\},\quad \dot{e}(x,z,\theta) = \frac{\partial}{\partial \theta} {e}(x,z,\theta),\quad \dot{f}(x,z,\theta) = \frac{\partial}{\partial \theta} f(x,z,\theta).
\$
Therefore, the tangent space $\cT$ at $\theta$ is 
\#\label{eq:r_tangent}
\cT = \Big\{& t\cdot s_1(y\given x,z ) + (1-t)\cdot s_0(y\given x,z ) +  a(x,z) \big(t-e(x,z )\big) + s(x,z)  \colon \\
 &\int s_w(y\given x,z ) f_w(y\given x,z ) \ud y = 0,\forall w\in \{0,1\},~\forall x,z,\\
&\qquad  ~ \int s(x,z)f(x,z)\ud x = 0, ~a(x,z) \text{ is square integrable}\Big\}.\notag
\#

\paragraph{Step 3: Pathwise derivative.} The target $\delta'$ can be parametrized as 
\$
\delta'(\theta) = \frac{\int y f_1(y\given x,z,\theta) f(x,z,\theta) \ud x \ud y \ud z - \int y f_0(y\given x,z,\theta) f(x,z,\theta) \ud x \ud y \ud z}{\int z f(x,z,\theta) \ud x \ud z}.
\$
Therefore, the pathwise derivative at $\theta$ is 
\$
\frac{\partial \delta'(\theta)}{\partial \theta} &= \frac{1}{\EE[Z]}\bigg(  \int y \dot{f}_1(y\given x,z,\theta) f(x,z,\theta) \ud x \ud y \ud z - \int y \dot{f}_0(y\given x,z,\theta) f(x,z,\theta) \ud x \ud y \ud z   \bigg) \\
&\quad + \frac{1}{\EE[Z]}\bigg(  \int y {f}_1(y\given x,z,\theta) \dot{f}(x,z,\theta) \ud x \ud y \ud z - \int y {f}_0(y\given x,z,\theta) \dot{f}(x,z,\theta) \ud x \ud y \ud z   \bigg) \\
&\quad + \frac{\EE[Y(1)]-\EE[Y(0)]}{\EE[Z]^2} \int z\dot{f}(x,z,\theta) \ud x\ud z \\
&= \frac{1}{\EE[Z]} \bigg(  \EE\bigg[ Y(1) \frac{\dot{f}_1(Y(1)\given X,Z,\theta)}{ {f}_1(Y(1)\given X,Z,\theta)}  \bigg]  - \EE\bigg[ Y(0) \frac{\dot{f}_0(Y(0)\given X,Z,\theta)}{ {f}_0(Y(0)\given X,Z,\theta)}  \bigg]    \bigg) \\
&\quad +
\frac{1}{\EE[Z]} \cdot  \EE\bigg[ \big(Y(1)-Y(0)\big) \frac{\dot{f}(  X,Z,\theta)}{ {f}(  X,Z,\theta)}  \bigg]  
 + \frac{\EE[Y(1)]-\EE[Y(0)]}{\EE[Z]^2} \cdot \EE\bigg[ Z \cdot \frac{\dot{f}(  X,Z,\theta)}{ {f}(  X,Z,\theta)}  \bigg],
\$
where the second equality follows similar arguments as in the proof of Theorem~\ref{thm:rr_var_bound}.

\paragraph{Step 4: Projection onto tangent space.} We show that $\phi_{\delta'}^*$ is in $\cT$ and is the projection of $\frac{\partial \delta'(\theta)}{\partial \theta}$ onto $\cT$. Firstly, note that $\phi_{\delta'}^* \in \cT$ by setting $a(x,z)\equiv 0$ and 
\$
&s_1(y\given x,z) = \frac{y-\mu_1(x,z)}{e(x,z)},\quad s_0(y\given x,z) = \frac{y-\mu_0(x,z)}{1-e(x,z)},\\
 &s(x,z) = \mu_1(x,z)-\mu_0(x,z) - \EE[Y(1)] +\EE[Y(0)]
\$
in the specification of~\ref{eq:r_tangent}. Then to show that $\phi_{\delta'}^*$ is the projection of $\frac{\partial \delta'(\theta)}{\partial \theta}$ onto $\cT$, it suffices to show that 
\#\label{eq:r_zzz}
\frac{\partial \delta'(\theta)}{\partial \theta} = \EE_\theta\big[  s(\theta;X_i,Y_i,Z_i,T_i) \cdot \phi_{\delta'}^*(Y_i,Z_i,X_i,T_i)  \big]
\#
for the score function in~\eqref{eq:r_score}. To this end, we now show that 
\$
& 
\EE\big[ \big(Z_i-\EE[Z_i]\big) \cdot s(\theta;X_i,Y_i,Z_i,T_i) ] = \EE\bigg[ Z \cdot \frac{\dot{f}(  X,Z,\theta)}{ {f}(  X,Z,\theta)}  \bigg],\quad\text{and}\\
&\EE\big[\Gamma_i^{**} \cdot s(\theta;X_i,Y_i,Z_i,T_i) ] = \EE\bigg[ Y(1) \frac{\dot{f}_1(Y(1)\given X,Z,\theta)}{ {f}_1(Y(1)\given X,Z,\theta)}  \bigg]  - \EE\bigg[ Y(0) \frac{\dot{f}_0(Y(0)\given X,Z,\theta)}{ {f}_0(Y(0)\given X,Z,\theta)}  \bigg] \\
&\qquad \qquad +  \EE\bigg[ \big(Y(1)-Y(0)\big) \frac{\dot{f}(  X,Z,\theta)}{ {f}(  X,Z,\theta)}  \bigg].
\$
For the first equation, note that $\EE[s(\theta;X_i,Y_i,Z_i,T_i)]=0$, hence
\$
&\EE\big[ \big(Z_i-\EE[Z_i]\big) \cdot s(\theta;X_i,Y_i,Z_i,T_i) ] = \EE\big[  Z_i  \cdot s(\theta;X_i,Y_i,Z_i,T_i) ]  \\
&= \EE\bigg[  T \cdot Z \cdot \frac{\dot{f}_1(Y(1) \given X,Z,\theta)}{f_1(Y(1)\given X,Z,\theta)} + (1-T)\cdot Z \cdot \frac{\dot{f}_0(Y(0)\given X,Z,\theta)}{f_0(Y(0)\given X,Z,\theta)} \bigg]\\
&\qquad + \EE\bigg[  Z\cdot \frac{T-e(X,Z,\theta)}{e(X,Z,\theta)(1-e(X,Z,\theta))}\cdot \dot{e}(X,Z,\theta) + Z\cdot \frac{\dot{f}(X,Z,\theta)}{f(X,Z,\theta)}\bigg] \\
&= \EE\bigg[ e(X,Z) \cdot Z \cdot \frac{\dot{f}_1(Y(1) \given X,Z,\theta)}{f_1(Y(1)\given X,Z,\theta)} + (1-e(X,Z))\cdot Z \cdot \frac{\dot{f}_0(Y(0)\given X,Z,\theta)}{f_0(Y(0)\given X,Z,\theta)} \bigg]\\
&\qquad + \EE\bigg[  Z\cdot \frac{\EE[T\given X,Z]-e(X,Z,\theta)}{e(X,Z,\theta)(1-e(X,Z,\theta))}\cdot \dot{e}(X,Z,\theta) + Z\cdot \frac{\dot{f}(X,Z,\theta)}{f(x,Z,\theta)}\bigg] = \EE\bigg[ Z\cdot \frac{\dot{f}(X,Z,\theta)}{f(X,Z,\theta)}\bigg],
\$
where the third equality follows from the conditional independence of $T$ and $Y(1),Y(0)$ given $X,Z$ and the last equality follows the fact that 
\#\label{eq:r_zero}
\EE\bigg[g(X,Z)\frac{\dot{f}_w(Y(w)\given X,Z,\theta)}{f_w(Y(w)\given X,Z,\theta)}\bigg]=0,\quad \forall~ w\in\{0,1\}
\#
and any measurable function $g(\cdot,\cdot)$. For the second equality, note that 
\$
&\EE\bigg[\Gamma_i^{**} \cdot  T_i \cdot \frac{\dot{f}_1(Y_i(1) \given X_i,Z_i,\theta)}{f_1(Y_i(1)\given X_i,Z_i,\theta)}  \bigg] \\
&=  \EE\bigg[ \Big(\mu_1(X,Z)-\mu_0(X,Z) - \EE[Y(1)]+\EE[Y(0)]\Big)\cdot  T  \cdot \frac{\dot{f}_1(Y (1) \given X ,Z ,\theta)}{f_1(Y (1)\given X ,Z ,\theta)}  \bigg] \\
&\qquad + \EE\bigg[ \frac{T}{e(X,Z)} \big(Y(1)-\mu_1(X,Z)\big)  \cdot \frac{\dot{f}_1(Y (1) \given X ,Z ,\theta)}{f_1(Y (1)\given X ,Z ,\theta)}  \bigg] \\
&= \EE\bigg[ \Big(\mu_1(X,Z)-\mu_0(X,Z) - \EE[Y(1)]+\EE[Y(0)]\Big)\cdot  e(X,Z)  \cdot \frac{\dot{f}_1(Y (1) \given X ,Z ,\theta)}{f_1(Y (1)\given X ,Z ,\theta)}  \bigg] \\
&\qquad + \EE\bigg[ \big(Y(1)-\mu_1(X,Z)\big)  \cdot \frac{\dot{f}_1(Y (1) \given X ,Z ,\theta)}{f_1(Y (1)\given X ,Z ,\theta)}  \bigg] = \EE\bigg[  Y(1) \cdot \frac{\dot{f}_1(Y (1) \given X ,Z ,\theta)}{f_1(Y (1)\given X ,Z ,\theta)}  \bigg].
\$
Here the second equality follows from the conditional independence of $T$ and $Y(1),Y(0)$ given $X,Z$ and the last equality follows from~\eqref{eq:r_zero}. Following exactly the same arguments, we have 
\$
\EE\bigg[\Gamma_i^{**} \cdot  (1-T_i) \cdot \frac{\dot{f}_0(Y_i(0) \given X_i,Z_i,\theta)}{f_0(Y_i(0)\given X_i,Z_i,\theta)}  \bigg]  =  \EE\bigg[ Y(0)\cdot \frac{\dot{f}_0(Y_i(0) \given X_i,Z_i,\theta)}{f_0(Y_i(0)\given X_i,Z_i,\theta)}  \bigg].
\$
Meanwhile, note that 
\$
&\EE\bigg[\Gamma_i^{**} \cdot   \frac{\dot{f} (  X_i,Z_i,\theta)}{f ( X_i,Z_i,\theta)}  \bigg]  \\
&=\EE\bigg[ \Big( \mu_1(X,Z) - \mu_0(X,Z) - \EE[Y(1)]+\EE[Y(0)]\Big) \cdot   \frac{\dot{f} (  X ,Z ,\theta)}{f ( X ,Z ,\theta)}  \bigg] \\
&\quad + \EE\bigg[\frac{T}{ e(X,Z)} \cdot \big(Y(1)-\mu_1(X,Z)\big) \cdot  \frac{\dot{f} (  X ,Z ,\theta)}{f ( X ,Z ,\theta)}  \bigg]
+ \EE\bigg[\frac{1-T}{ 1-e(X,Z)} \cdot \big(Y(0)-\mu_0(X,Z)\big) \cdot  \frac{\dot{f} (  X ,Z ,\theta)}{f ( X ,Z ,\theta)}  \bigg] \\
&= \EE\bigg[ \Big( \mu_1(X,Z) - \mu_0(X,Z)  \Big) \cdot   \frac{\dot{f} (  X ,Z ,\theta)}{f ( X ,Z ,\theta)}  \bigg] = \EE\bigg[ \big( Y(1)-Y(0)  \big) \cdot   \frac{\dot{f} (  X ,Z ,\theta)}{f ( X ,Z ,\theta)}  \bigg], 
\$
where the second equality follows from the conditional independence and the fact that $\EE\big[\frac{\dot{f} (  X ,Z ,\theta)}{f ( X ,Z ,\theta)}\big]=0$, and the last equality follows from the takeout and tower property of conditional expectations. Also, 
\$
&\EE\bigg[\Gamma_i^{**} \cdot   \frac{T_i-e(X_i,Z_i)}{e(X_i,Z_i)(1-e(X_i,Z_i))} \bigg] \\
&= \EE\bigg[\Big( \mu_1(X,Z) - \mu_0(X,Z) - \EE[Y(1)]+\EE[Y(0)]\Big) \cdot    \frac{T -e(X ,Z )}{e(X ,Z )(1-e(X ,Z))} \bigg] \\
&\quad + \EE\bigg[\frac{T}{ e(X,Z)} \cdot \big(Y(1)-\mu_1(X,Z)\big)\cdot  \frac{T -e(X ,Z )}{e(X ,Z )(1-e(X ,Z))} \bigg] \\
&\quad +  \EE\bigg[\frac{1-T}{ 1-e(X,Z)} \cdot \big(Y(0)-\mu_0(X,Z)\big)\cdot  \frac{T -e(X ,Z )}{e(X ,Z )(1-e(X ,Z))} \bigg] =0,
\$
where all the three terms in the summation equals zero because of the tower and takeout property of conditional expectations. Putting them together, we've proved~\eqref{eq:r_zzz}. Therefore, $\phi_{\delta'}^*$ is the projection of $\frac{\partial \delta'(\theta)}{\partial \theta}$ onto $\cT$. Hence, $\phi_{\delta'}^*$ is the efficient influence function, and its variance is the semiparametric asymptotic variance bound for estimating $\delta'$. 
\end{proof}

\section{Technical lemmas}

\begin{lemma}\label{lem:cond_to_op}
Let $\cF_n$ be a sequence of $\sigma$-algebra, and let $A_n\geq 0$ be a sequence of nonnegative random variables. If $\EE[A_n\given \cF_n] = o_P(1)$, then $A_n=o_P(1)$. 
\end{lemma}

\begin{proof}[Proof of Lemma~\ref{lem:cond_to_op}]
By Markov's inequality, for any $\epsilon>0$, we have 
\$
B_n := \PP( A_n >\epsilon \given \cF_n) \leq \frac{\EE[A_n\given \cF_n]}{\epsilon } = o_P(1),
\$
and $B_n\in[0,1]$ are bounded random variables. For any subsequence $\{n_k\}_{k\geq 1}$ of $\NN$, since $B_{n_k}\pto 0$, there exists a subsequence $\{n_{k_i}\}_{i\geq 1} \subset \{n_k\}_{k\geq 1}$ such that $B_{n_{k_i}} \stackrel{\text{a.s.}}{\to} 0$ as $i\to \infty$. By the dominated convergence theorem, we have $\EE[B_{n_{k_i}}]\to 0$, or equivalently, 
$
\PP(A_{n_{k_i}} >\epsilon) \to 0.
$
Therefore, for any subsequence $\{n_k\}_{k\geq 1}$ of $\NN$, there exists a subsequence $\{n_{k_i}\}_{i\geq 1} \subset \{n_k\}_{k\geq 1}$ such that $A_{n_{k_i}} \pto 0$ as $i\to \infty$. By the arbitrariness of $\{n_k\}_{k\geq 1}$, we know $A_n\pto 0$ as $n\to \infty$, which completes the proof. 
\end{proof}